\documentclass[draftclsnofoot,onecolumn]{IEEEtran} 	
\usepackage{amsmath,amssymb}
\usepackage{amsthm}
\usepackage{color}
\usepackage[mathscr]{eucal}
\usepackage{graphics,graphicx,multicol}
\usepackage{epsfig}
\usepackage{enumerate}
\usepackage{subfigure}
\usepackage{algorithmic}
\usepackage{algorithm}
\usepackage{morefloats}
\usepackage{multirow}
\usepackage{array}
\usepackage{pstricks, pst-node, pst-plot, pst-circ}
\usepackage{moredefs}
\usepackage{courier}
\usepackage{colortbl}
\usepackage{xcolor}
\usepackage{cite}

\newtheorem{thm}{Theorem}[section] 
\newtheorem{lem}[thm]{Lemma}
\newtheorem{cor}[thm]{Corollary}
\newtheorem{rem}[thm]{Remark}
\begin{document}

\title{Optimal Downlink Power Allocation in Cellular Networks}
\author{Ahmed~Abdelhadi,~\IEEEmembership{}
        Awais~Khawar,~\IEEEmembership{}
        and~T. Charles~Clancy~\IEEEmembership{}

\thanks{Ahmed Abdelhadi (aabdelhadi@vt.edu) is with Virginia Polytechnic Institute and State University, Arlington, VA, 22203. 

Awais Khawar (awais@vt.edu) is with Virginia Polytechnic Institute and State University, Arlington, VA, 22203. 

T. Charles Clancy (tcc@vt.edu) is with Virginia Polytechnic Institute and State University, Arlington, VA, 22203. }
}
\maketitle

\begin{abstract}

In this paper, we introduce a novel approach for power allocation in cellular networks. In our model, we use sigmoidal-like utility functions to represent different users' modulation schemes. Each utility function is a representation of the probability of successfully transmitted packets per unit of power consumed by a user, when using a certain modulation scheme. We consider power allocation with utility proportional fairness policy, where the fairness among users is in utility percentage i.e. percentage of successfully transmitted packets of the corresponding modulation scheme. We formulate our network utility maximization problem as a product of utilities of all users and prove that our power allocation optimization problem is convex and therefore the optimal solution is tractable. We present a distributed algorithm to allocate base station (BS) powers optimally with priority given to users running lower modulation schemes while ensuring non-zero power allocation to users running higher modulation schemes.
 Our algorithm prevents fluctuation in the power allocation process and is capable of traffic and modulation dependent pricing i.e. charges different price per unit power from different users depending in part on their modulation scheme and total power available at the BS. This is used to flatten traffic and decrease the service price for users.

\end{abstract}
%

\providelength{\AxesLineWidth}       \setlength{\AxesLineWidth}{0.5pt}%
\providelength{\plotwidth}           \setlength{\plotwidth}{8cm}
\providelength{\LineWidth}           \setlength{\LineWidth}{0.7pt}%
\providelength{\MarkerSize}          \setlength{\MarkerSize}{3pt}%
\newrgbcolor{GridColor}{0.8 0.8 0.8}%
\newrgbcolor{GridColor2}{0.5 0.5 0.5}%

\section{Introduction}\label{sec:intro}

In the past, cellular networks were dominated by voice-only traffic. However, in recent years, there has been a significant growth in the amount of multimedia-rich traffic over cellular networks. In order to support such traffic, networks require higher data rates which can be achieved by using higher modulation schemes. This is the reason current and emerging cellular standards are supporting various higher modulation schemes. For example, Long term evolution (LTE), the fourth-generation (4G) wireless standard specified by 3$^\text{rd}$ Generation Partnership Project (3GPP), supports higher modulation schemes such as QPSK, 16-QAM, and 64-QAM according to 3GPP Release 10 or more commonly known as LTE-Advanced (LTE-A). The next evolution of LTE -- LTE Release 12 and beyond -- also refereed to as LTE-B supports higher modulation schemes upto 256-QAM \cite{EricsonLTERel12, Huawei_LTE_B}. Higher modulation schemes require higher transmit power to achieve a certain signal-to-noise ratio (SNR) which can guarantee 
minimum successful transmission probability of packets. This motivates numerous research efforts to optimally allocate power for users seeking better quality-of-service (QoS), where QoS can be the minimum successful transmission probability of packets. In a cellular network, different users can run different services, requiring different modulation schemes, and thus can have different power and QoS requirements. 

One aspect of improving power allocation and achieving better QoS is to use network utility maximization framework. This framework was first explored by the seminal works of Kelly et al. \cite{kelly98powercontrol} and Low et al. \cite{Low99optimizationflow} for wired networks, such as the Internet, and later by Marbach et al. \cite{MB02} and Lee et al. \cite{DL_PowerAllocation} for wireless networks, such as Code Division Multiple Access (CDMA)-based systems. The utility function can be considered as a controlling parameter through which a user's QoS can be guaranteed. In this treatment, each user's utility function is a function of its power allocation and the goal is to allocate powers in order to maximize network utility, which is defined as a product of all users' utilities. The traditional approaches to model network utility maximization problem is by the summation of users' utility functions \cite{MB02}\cite{DL_PowerAllocation}. However, such a formulation can drop users in order to maximize system 
performance and thus all users are not treated fairly, as an example consult \cite{DL_PowerAllocation}. In this paper, we deviate from this trend and introduce a novel network utility maximization framework which is a product of all users' utility functions. The motivation behind this approach is to ensure that all users are entertained and no user is dropped in order to maximize network utility and at the same time maintain minimum QoS for all users.

A utility function is a representation of QoS of a user. For example, utility functions have been defined to maximize signal-to-interference-plus-noise ratio (SINR) \cite{DL_PowerAllocation}, Shannon capacity \cite{MPS05}, ratio of throughput to transmit power \cite{GM99} etc. Thus, the type of utility function represents each user's QoS characteristics and it is possible that in a network supporting different services we have to deal with various types of utility functions. The shape of a utility function depends upon the characteristics of service it is representing and thus can have many shapes including concave, convex, sigmoidal-like or S-shaped, and inverse-S-shaped \cite{LK08}. The optimality of the solution of network utility maximization problem depends upon the shape of a utility function. For example, services represented by concave utility functions satisfy the convexity conditions of network utility maximization problem, when it is represented by sum of users' utility functions, and thus yields 
optimal solution. This is not the case for non-concave utility functions. Moreover, algorithms proposed for downlink power allocation are utility function specific and can't be generalized to other utility functions due to the various shapes a utility function can take. Thus, an algorithm can be optimal for one class of utility functions and at the same time sub-optimal for other classes of utility functions. In this framework, not only our power allocation and pricing algorithms are optimal but the solution to our network utility maximization problem is also optimal.

In this paper we consider sigmoidal-like, or S-shaped, utility functions, which are first convex and then concave, for network utility maximization problem because our QoS criterion is probability of successful packet transmission and its cumulative distribution function is also S-shaped i.e. first convex and then concave. Thus, sigmoidal-like utility function is a natural choice to represent probability of successful packet transmission which can be a function of SINR \cite{DL_PowerAllocation} or signal-to-interference ratio (SIR) \cite{XSC03}.

The wireless broadcast channel is random in nature and users experience independent and heterogeneous communication environment. In such a challenging propagation environment it is difficult to design resource allocation algorithms that maximize system efficiency, ensure fairness, and meet QoS requirements of all users. For example, in order to improve system efficiency, opportunistic  resource allocation algorithms are used \cite{VTL02}. These algorithms favor users with whom they share good channel quality or those typically closer to the BS, and tend to avoid users that are in deep-fade or the ones at cell edges. Despite of maximizing system efficiency such algorithms fail to (a) satisfy QoS requirements of users and (b) maintain fair allocation of resources among users. This can be avoided by fair allocation of resources among all users and thus utility proportional fairness problems are of more interest \cite{SPO+13}.


Pricing can be used as a control measure by network providers i.e. users can be motivated to make rational decisions on the usage of network resources in order to maximize network utility. It has been successfully applied to control congestion in the Internet \cite{TWL07} and rate and power in wireless networks \cite{SMG02, CL07}. In network utility maximization problems it is typical of network providers to set shadow prices for its resources. Similarly, we set shadow prices for powers with the goal of achieving optimization of individual users' QoS and the distributed optimization of the power allocation process in order to maximize network performance. Therefore, we charge price from the user depending upon the modulation scheme used and the total power available at the BS. So users using lower modulation schemes require less transmit power and thus pay lower prices as compared to users using higher modulation schemes. Also, the price per unit power which the users pay depends upon the total power 
available at the BS. So at the time of high network utilization, when the demand for power is high, powers will be scarce and thus users will pay more for the same amount of per unit power as compared to the time of low network utilization. Thus, pricing can be used by network providers to flatten traffic i.e. users can be motivated to use the network during off-peak hours as they pay less for the same service.

%




\subsection{Related Work}\label{sec:related}


In \cite{WC13}, the authors consider a proportional fairness scheduling problem in a network supporting different users' modulation schemes. However, their approach is limited since they propose a two-user opportunistic proportional fairness scheduling and utility maximization problem. Their algorithm for opportunistic proportional fairness scheduling is optimal but for utility maximization their proposed algorithm is suboptimal.

Lee et al. consider a downlink utility-based power allocation algorithm for Code Division Multiple Access (CDMA)-based systems in \cite{DL_PowerAllocation}. The algorithm maximizes the total system utility but can drop users to maximize the overall system utilization, therefore, it does not guarantee minimum QoS for all users. Moreover, they use an approximation to solve non-concave sigmoidal-like utility maximization problem which is Pareto optimal. An alternate approach to approximate concave and non-concave utility functions using minimum mean-square error was proposed by \cite{RebeccaThesis}. The approximated utility function is used to solve the power allocation problem using a modified version of the distributed power allocation algorithm presented in \cite{kelly98powercontrol}. The same problem is considered by Tychogiorgos et al. with a non-convex optimization formulation for maximization of utility function \cite{DBLP:conf/globecom/TychogiorgosGL11, DBLP:conf/pimrc/TychogiorgosGL11}. The proposed 
algorithm solves the problem 
when the duality gap is zero but doesn't converge to the optimal solution for a positive duality gap.

evolved NodeB (eNodeB) can also be used to allocate powers fairly among users by taking into account user positions. This utility proportional fairness problem can achieve QoS targets both in terms of fairness and throughput \cite{Power_alloc3}. In addition, the authors have shown that their signal-to-interference-plus-noise ratio (SINR)-based fair power allocation algorithm, which starts from the cell-edge users and converges to the center, can achieve maximum sum rate and optimal fairness. On the other hand, an alternate approach to improve the throughput for the cell-edge user, an uplink power control and resource allocation scheme is considered in \cite{Power_alloc4}. The algorithm takes into account the interference to and from adjacent cells. It first allocates resources independently among the cells and then adjusts resources and powers for the center users based on resource allocation and users' position among adjacent cells. In \cite{DBLP:conf/qosip/Harks05}, the authors propose a utility max-min 
fairness power allocation for users with elastic and 
real-time traffic sharing a single path in the network. In \cite{UtilityFairness}, the authors proposed a utility proportional fair optimization formulation for high-SINR wireless networks using a utility max-min architecture. They compare their algorithm to the traditional bandwidth proportional fair algorithms \cite{utility_fair} and present a closed form solution that prevents oscillations in the network.

In \cite{Power_Rate_alloc6, Power_Rate_alloc9}, the authors consider downlink power allocation for semi-elastic applications, for example video conferencing, in 4G cellular networks using a sigmoidal-like utility function. They consider a utility maximization problem over multiple time slots since their algorithm allocates powers and subcarriers in each time slot in order to optimize average user utility over time by an exchange of price and demand among users, the network, and an intermediate power allocation module. A similar network utility maximization problem for video streaming over cellular networks, for both uplink and downlink, is considered in \cite{alloc_video_streaming}. The authors consider a joint optimization of video and network resources to maximize total video reception quality of a limited number of users without interrupting the service of other voice users. 
the convergence of their algorithms.  



In \cite{PricingSIGCOMM12}, the authors conduct a trial of time-dependent pricing (TDP) system called TUBE for iPhone or iPad users, using 3G cellular services, who are charged according to the proposed TDP algorithm. Our results show that TDP benefits both operators and customers, flattening the temporal fluctuation of demand while allowing users to save money by choosing the time and volume of their usage.

For earlier studies on optimal rate allocation, we refer to \cite{Ahmed_Utility1}. Rate pricing was introduced in \cite{Ahmed_Utility2} for single cell model. Rate allocation with carrier aggregation is studied in \cite{Ahmed_Utility4, Haya_Utility1}. Rate allocation with prioritization of mobile users based on their subscription is investigated in \cite{Ahmed_Utility3}. Rate allocation with guaranteed bit rate (GBR) to mobile users running specific services is discussed in \cite{Haya_Utility2}.

\subsection{Our Contributions}\label{sec:contributions}
Our contributions in this paper are summarized as:
\begin{itemize}

\item \textbf{Utility Proportional Fairness:} We introduce a utility proportional fairness optimization problem, where the fairness among users is in utility percentage. Utility percentage is the percentage of packet successful transmission versus power. Each modulation is represented by a utility function that is a sigmoidal-like function. We prove that the proposed optimization problem is convex and therefore the global optimal solution is tractable. In addition, the optimization problem formulation gives priority to lower modulation users when allocating powers while ensuring non-zero power allocation to higher modulation users. 

\item \textbf{Robust and Convergent Power Allocation Algorithm:} We present a robust distributed power allocation algorithm that converges to the optimal powers for users with lower and higher modulation schemes by introducing a fluctuation decay function that damps the fluctuations. Our algorithm senses fluctuations and tunes the fluctuation decay function to reach convergence.

\item \textbf{Traffic-dependent Bidding/Pricing:} We present a pricing policy for network providers that can reduce the demand in power i.e. flatten traffic load on the network and decrease the overall service cost to subscribers. The pricing policy is dependent on users' modulation schemes and the total available power at the BS. Thus, network providers can charge higher for the same service when the demand for power is high and vice versa. This serves as an incentive to subscribers to use the service when the demand is low as they pay less for the same service.
\end{itemize}


The remainder of this paper is organized as follows. Section \ref{sec:upf} motivates the use of utility proportional fairness over other methods, introduces system topology, and discusses representation of modulations by sigmoidal-like utility functions. Section \ref{sec:Problem_formulation} presents the problem formulation. Section \ref{sec:Proof} proves the global optimal solution exists and is tractable. In Section \ref{sec:Dual}, we introduce the dual problem. In Section \ref{sec:Algorithm}, we present our distributed optimization algorithm. Section \ref{sec:conv_analy} analyzes the power allocation algorithm and discusses its convergence. In Section \ref{sec:OuP_Algorithm}, we present a more robust distributed power allocation algorithm for the utility proportional fairness optimization problem. Section \ref{sec:sim} discusses simulation setup and provides quantitative results along with the discussion. Section \ref{sec:conclude} concludes the paper.

%

\begin{figure}
\centering
\includegraphics[trim=0in 0in 0in 0in,width=3.4in]{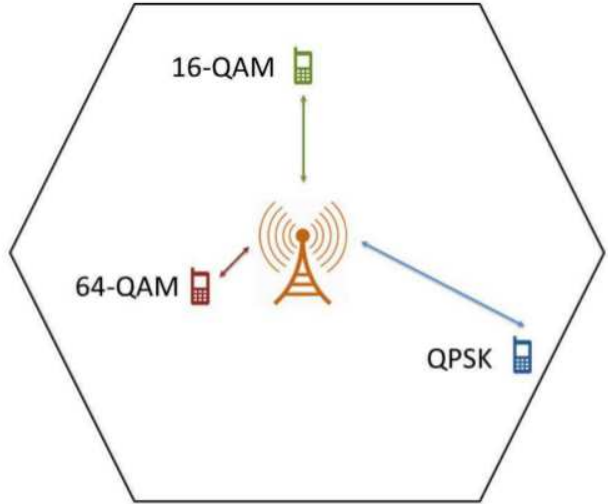}
\caption{Adaptive modulation schemes are used to enhance high data-rate access to users in 3G and 4G cellular networks. When the signal-to-noise (SNR) ratio is the highest, a higher modulation scheme is used, for example 256-QAM. This usually happens for users near the BS as they share good channel quality with the BS. Higher modulation schemes require higher transmit power to maintain an acceptable SNR for successful transmission of packets, as seen from Fig. \ref{fig:prob}.} 
\label{fig:AdapMod}
\end{figure}

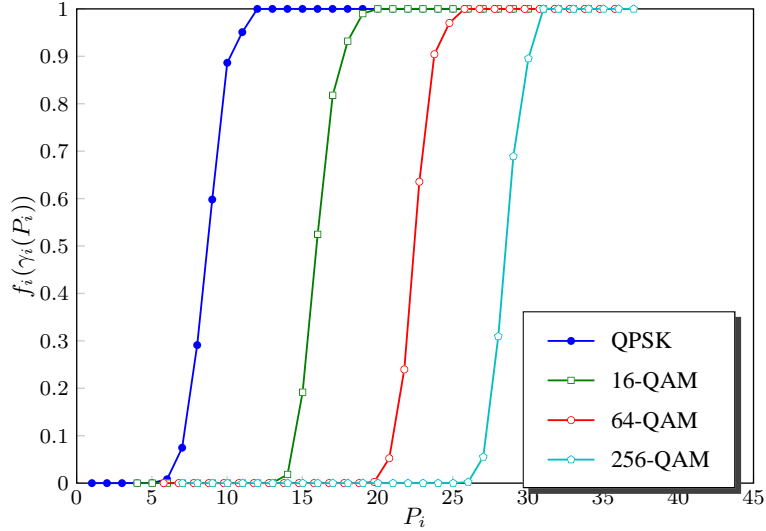
\begin{figure}[!t]
\centering
%
\psset{xunit=0.025\plotwidth,yunit=0.788710\plotwidth}%
\begin{pspicture}(-4.976959,-0.111111)(45.829493,1.023392)%


\psline[linewidth=\AxesLineWidth,linecolor=GridColor](0.000000,0.000000)(0.000000,0.015215)
\psline[linewidth=\AxesLineWidth,linecolor=GridColor](5.000000,0.000000)(5.000000,0.015215)
\psline[linewidth=\AxesLineWidth,linecolor=GridColor](10.000000,0.000000)(10.000000,0.015215)
\psline[linewidth=\AxesLineWidth,linecolor=GridColor](15.000000,0.000000)(15.000000,0.015215)
\psline[linewidth=\AxesLineWidth,linecolor=GridColor](20.000000,0.000000)(20.000000,0.015215)
\psline[linewidth=\AxesLineWidth,linecolor=GridColor](25.000000,0.000000)(25.000000,0.015215)
\psline[linewidth=\AxesLineWidth,linecolor=GridColor](30.000000,0.000000)(30.000000,0.015215)
\psline[linewidth=\AxesLineWidth,linecolor=GridColor](35.000000,0.000000)(35.000000,0.015215)
\psline[linewidth=\AxesLineWidth,linecolor=GridColor](40.000000,0.000000)(40.000000,0.015215)
\psline[linewidth=\AxesLineWidth,linecolor=GridColor](45.000000,0.000000)(45.000000,0.015215)
\psline[linewidth=\AxesLineWidth,linecolor=GridColor](0.000000,0.000000)(0.540000,0.000000)
\psline[linewidth=\AxesLineWidth,linecolor=GridColor](0.000000,0.100000)(0.540000,0.100000)
\psline[linewidth=\AxesLineWidth,linecolor=GridColor](0.000000,0.200000)(0.540000,0.200000)
\psline[linewidth=\AxesLineWidth,linecolor=GridColor](0.000000,0.300000)(0.540000,0.300000)
\psline[linewidth=\AxesLineWidth,linecolor=GridColor](0.000000,0.400000)(0.540000,0.400000)
\psline[linewidth=\AxesLineWidth,linecolor=GridColor](0.000000,0.500000)(0.540000,0.500000)
\psline[linewidth=\AxesLineWidth,linecolor=GridColor](0.000000,0.600000)(0.540000,0.600000)
\psline[linewidth=\AxesLineWidth,linecolor=GridColor](0.000000,0.700000)(0.540000,0.700000)
\psline[linewidth=\AxesLineWidth,linecolor=GridColor](0.000000,0.800000)(0.540000,0.800000)
\psline[linewidth=\AxesLineWidth,linecolor=GridColor](0.000000,0.900000)(0.540000,0.900000)
\psline[linewidth=\AxesLineWidth,linecolor=GridColor](0.000000,1.000000)(0.540000,1.000000)

{ \footnotesize 
\rput[t](0.000000,-0.015215){$0$}
\rput[t](5.000000,-0.015215){$5$}
\rput[t](10.000000,-0.015215){$10$}
\rput[t](15.000000,-0.015215){$15$}
\rput[t](20.000000,-0.015215){$20$}
\rput[t](25.000000,-0.015215){$25$}
\rput[t](30.000000,-0.015215){$30$}
\rput[t](35.000000,-0.015215){$35$}
\rput[t](40.000000,-0.015215){$40$}
\rput[t](45.000000,-0.015215){$45$}
\rput[r](-0.540000,0.000000){$0$}
\rput[r](-0.540000,0.100000){$0.1$}
\rput[r](-0.540000,0.200000){$0.2$}
\rput[r](-0.540000,0.300000){$0.3$}
\rput[r](-0.540000,0.400000){$0.4$}
\rput[r](-0.540000,0.500000){$0.5$}
\rput[r](-0.540000,0.600000){$0.6$}
\rput[r](-0.540000,0.700000){$0.7$}
\rput[r](-0.540000,0.800000){$0.8$}
\rput[r](-0.540000,0.900000){$0.9$}
\rput[r](-0.540000,1.000000){$1$}
} 

\psframe[linewidth=\AxesLineWidth,dimen=middle](0.000000,0.000000)(45.000000,1.000000)

{ \small 
\rput[b](22.500000,-0.111111){
\begin{tabular}{c}
$P_i$\\
\end{tabular}
}

\rput[t]{90}(-4.976959,0.500000){
\begin{tabular}{c}
$f_i(\gamma_i(P_i))$\\
\end{tabular}
}
} 

\newrgbcolor{color14.0104}{0  0  1}
\psline[plotstyle=line,linejoin=1,showpoints=true,dotstyle=*,dotsize=\MarkerSize,linestyle=solid,linewidth=\LineWidth,linecolor=color14.0104]
(1.010300,0.000000)(2.010300,0.000000)(3.010300,0.000000)(4.010300,0.000013)(5.010300,0.000346)
(6.010300,0.008213)(7.010300,0.074633)(8.010300,0.291049)(9.010300,0.598035)(10.010300,0.886385)
(11.010300,0.950990)(12.010300,1.000000)(13.010300,1.000000)(14.010300,1.000000)(15.010300,1.000000)
(16.010300,1.000000)(17.010300,1.000000)(18.010300,1.000000)(19.010300,1.000000)(20.010300,1.000000)
(21.010300,1.000000)(22.010300,1.000000)(23.010300,1.000000)(24.010300,1.000000)(25.010300,1.000000)
(26.010300,1.000000)(27.010300,1.000000)(28.010300,1.000000)(29.010300,1.000000)(30.010300,1.000000)
(31.010300,1.000000)

\newrgbcolor{color15.0099}{0     0.49804           0}
\psline[plotstyle=line,linejoin=1,showpoints=true,dotstyle=Bsquare,dotsize=\MarkerSize,linestyle=solid,linewidth=\LineWidth,linecolor=color15.0099]
(4.020600,0.000000)(5.020600,0.000000)(6.020600,0.000000)(7.020600,0.000000)(8.020600,0.000000)
(9.020600,0.000000)(10.020600,0.000000)(11.020600,0.000000)(12.020600,0.000008)(13.020600,0.000747)
(14.020600,0.018305)(15.020600,0.191702)(16.020600,0.524576)(17.020600,0.817823)(18.020600,0.931970)
(19.020600,0.990000)(20.020600,1.000000)(21.020600,1.000000)(22.020600,1.000000)(23.020600,1.000000)
(24.020600,1.000000)(25.020600,1.000000)(26.020600,1.000000)(27.020600,1.000000)(28.020600,1.000000)
(29.020600,1.000000)(30.020600,1.000000)(31.020600,1.000000)(32.020600,1.000000)(33.020600,1.000000)
(34.020600,1.000000)

\newrgbcolor{color16.0099}{1  0  0}
\psline[plotstyle=line,linejoin=1,showpoints=true,dotstyle=Bo,dotsize=\MarkerSize,linestyle=solid,linewidth=\LineWidth,linecolor=color16.0099]
(5.781513,0.000000)(6.781513,0.000000)(7.781513,0.000000)(8.781513,0.000000)(9.781513,0.000000)
(10.781513,0.000000)(11.781513,0.000000)(12.781513,0.000000)(13.781513,0.000000)(14.781513,0.000000)
(15.781513,0.000000)(16.781513,0.000000)(17.781513,0.000000)(18.781513,0.000028)(19.781513,0.002465)
(20.781513,0.052488)(21.781513,0.239690)(22.781513,0.635535)(23.781513,0.904382)(24.781513,0.970299)
(25.781513,1.000000)(26.781513,1.000000)(27.781513,1.000000)(28.781513,1.000000)(29.781513,1.000000)
(30.781513,1.000000)(31.781513,1.000000)(32.781513,1.000000)(33.781513,1.000000)(34.781513,1.000000)
(35.781513,1.000000)

\newrgbcolor{color17.0099}{0     0.74902     0.74902}
\psline[plotstyle=line,linejoin=1,showpoints=true,dotstyle=Bpentagon,dotsize=\MarkerSize,linestyle=solid,linewidth=\LineWidth,linecolor=color17.0099]
(7.030900,0.000000)(8.030900,0.000000)(9.030900,0.000000)(10.030900,0.000000)(11.030900,0.000000)
(12.030900,0.000000)(13.030900,0.000000)(14.030900,0.000000)(15.030900,0.000000)(16.030900,0.000000)
(17.030900,0.000000)(18.030900,0.000000)(19.030900,0.000000)(20.030900,0.000000)(21.030900,0.000000)
(22.030900,0.000000)(23.030900,0.000000)(24.030900,0.000000)(25.030900,0.000030)(26.030900,0.002385)
(27.030900,0.055192)(28.030900,0.309425)(29.030900,0.689027)(30.030900,0.895338)(31.030900,1.000000)
(32.030900,1.000000)(33.030900,1.000000)(34.030900,1.000000)(35.030900,1.000000)(36.030900,1.000000)
(37.030900,1.000000)

{ \small 
\rput(36.9,0.17){%
\psshadowbox[framesep=0pt,linewidth=\AxesLineWidth]{\psframebox*{\begin{tabular}{l}
\Rnode{a1}{\hspace*{0.0ex}} \hspace*{0.4cm} \Rnode{a2}{~~QPSK} \\
\Rnode{a3}{\hspace*{0.0ex}} \hspace*{0.4cm} \Rnode{a4}{~~16-QAM} \\
\Rnode{a5}{\hspace*{0.0ex}} \hspace*{0.4cm} \Rnode{a6}{~~64-QAM} \\
\Rnode{a7}{\hspace*{0.0ex}} \hspace*{0.4cm} \Rnode{a8}{~~256-QAM} \\
\end{tabular}}
\ncline[linestyle=solid,linewidth=\LineWidth,linecolor=color14.0104]{a1}{a2} \ncput{\psdot[dotstyle=*,dotsize=\MarkerSize,linecolor=color14.0104]}
\ncline[linestyle=solid,linewidth=\LineWidth,linecolor=color15.0099]{a3}{a4} \ncput{\psdot[dotstyle=Bsquare,dotsize=\MarkerSize,linecolor=color15.0099]}
\ncline[linestyle=solid,linewidth=\LineWidth,linecolor=color16.0099]{a5}{a6} \ncput{\psdot[dotstyle=Bo,dotsize=\MarkerSize,linecolor=color16.0099]}
\ncline[linestyle=solid,linewidth=\LineWidth,linecolor=color17.0099]{a7}{a8} \ncput{\psdot[dotstyle=Bpentagon,dotsize=\MarkerSize,linecolor=color17.0099]}
}%
}%
} 

\end{pspicture}%

\caption{Cumulative distribution function of successful packet transmission for QPSK, 16-QAM, 64-QAM, and 256-QAM modulation schemes. Each modulation represents a user's utility function $U_i(\gamma_i(P_i))$ which is a function of power.}
\label{fig:prob}
\end{figure}

\begin {table}[]
\caption {Mathematical Notations}
\label{table:notations}
\begin{center}
\renewcommand{\arraystretch}{1} 
\begin{tabular}{| l | l |}
  \hline
 Notation 					&Description \\
\hline
 $M$ 							  &Total number of users in a cell \\
 $P_T$ 							&Total BS power available\\ 
 $P_i$ 							&Power assigned to the $i^{\text{th}}$ UE\\
 $\mathbf P$			  &Vector of all users' powers\\
 $\gamma_i(P_i)$ 		&SINR of the $i^{\text{th}}$ UE\\
 $U_i(\gamma_i(P_i))$ &Utility function of the $i^{\text{th}}$ UE \\
 $G_i$							&Accounts for path loss, shadowing, and fading \\
						        &between BS and the $i^{\text{th}}$ UE \\
 $I_i$							&Accounts for interference and background \\
									  &noise at the $i^{\text{th}}$ UE \\
 $z_i$							&Slack variable of the $i^{\text{th}}$ UE\\
 $p$								&Shadow price or the total price per unit power \\
     							  &for all UEs\\
$f_i(\gamma_i(P_i))$     		&Probability of packet transmission success as a\\
								&function of user power \\
\hline
\end{tabular}
\end{center}
\end {table}

\section{Utility Proportional Fairness}\label{sec:upf}

In this treatment, we consider utility proportional fairness rather than bandwidth proportional fairness network utility maximization problem. Regular proportional fairness problem suits the case when all users have the same modulation scheme \cite{kelly98powercontrol}. However, modern cellular network's users have different QoS needs and therefore can handle different modulation schemes. Thus, one has to give up regular proportional fairness schemes in favor of utility proportional fairness for power allocation in modern cellular systems \cite{WPL06, UtilityFairness}.


\textbf{System Topology:} In this paper, we consider, without loss of generality, a single cellular system consisting of a single BS and $M$ UEs. Our objective is to optimally allocate powers to UEs depending upon the modulation scheme used and price paid for the power such that all UEs are served with non-zero allocation of power. We assume a time-slotted system in which the power allocation algorithm executes in every time slot. We assume the time slot is of arbitrary interval in which a single or several packets can be transmitted and the propagation conditions such as path loss, fading, noise and intercell interference stay the same for each UE. The BS allocates power within the power limit available such that the power allocated to the $i^{\text{th}}$ UE by the BS is given by $P_i$. Each UE has its own utility function $U_i(\gamma_i(P_i))$ that corresponds to the type of modulation scheme being handled by the UE where $\gamma_i(P_i)$ is the ``generic'' signal quality metric for the $i^{\text{th}}$ UE, 
as in \cite{DL_PowerAllocation}. For cellular systems, this metric is commonly referred to as SINR as it not only depends on the $i^{\text{th}}$ UE's power allocation but also on the power allocation of all other UEs. For CDMA systems, this metric represents the bit energy to interference density ratio of the $i^{\text{th}}$ UE \cite{DL_PowerAllocation}.

\textbf{Sigmoidal-like Utility Functions:} Our objective is to assign optimal power levels to the UEs so as to have a minimum QoS for each UE. Therefore, we assume the utility functions $U_i(\gamma_i(P_i))$ to be sigmoidal-like functions. The utility functions have the following properties: 

\begin{itemize}
\item $U_i(0) = 0$ and $U_i(\gamma_i(P_i))$ is an increasing function of $P_i$.
\item $U_i(\gamma_i(P_i))$ is twice continuously differentiable in $P_i$.
\end{itemize}

The SINR $\gamma_i(P_i)$ is represented, as in \cite{DL_PowerAllocation}
\begin{equation}
\gamma_i(P_i) = \dfrac{G_i P_i}{G_i \sum_{m=1}^{M} P_m - G_iP_i + I_i}
\end{equation}
where $G_i$ accounts for path loss, shadowing, and fading between the BS and the $i^{\text{th}}$ UE and $I_i$ accounts for background noise and intercell interference to the $i^{\text{th}}$ UE. In our model, we use the normalized sigmoidal-like utility function, as in \cite{DL_PowerAllocation,XSC01,XSC03}, that can be expressed as 
\begin{equation}\label{eqn:sigmoid}
U_i(\gamma_i(P_i)) = c_i\Big(\frac{1}{1+e^{-a_i(P_i-b_i)}}-d_i\Big)
\end{equation}
where $c_i = \frac{1+e^{a_ib_i}}{e^{a_ib_i}}$ and $d_i = \frac{1}{1+e^{a_ib_i}}$. So, it satisfies $U(0)=0$ and $U(\infty)=1$. The values assigned to $a$ and $b$ play a role in the total system utility. For example, a UE with a larger value of $a$ or a UE with a smaller value of $b$ requires less power to achieve the same utility, given that other conditions are same \cite{DL_PowerAllocation}. In addition, we can tune parameters $a$ and $b$ to get an approximation of utility functions of various applications. 


\textbf{Link Adaptation in 4G Cellular Systems:} Current and emerging cellular standards adapt to the RF transmission conditions and select modulation and coding schemes which result in enhanced QoS for users. This is known as link adaptation. The choice of modulation and/or coding scheme is dynamically selected based on the channel-quality between the base station and the user. For example, in LTE systems, each user sends a signal-quality level, known as channel quality indicator (CQI) measurement, to the BS. The CQI measurement is based on the received signal-strength of the reference signal, transmitted by the BS with a constant power level and fixed modulation scheme.
This CQI measurements aids the BS to assign modulation and/or coding scheme to the user. Typically, a BS can assign upto 256-QAM scheme to users that report the highest value of CQI. This is usually for the users that share good channel with the BS or are close to the BS, see Fig. \ref{fig:AdapMod}. However, higher power needs to be allocated in order to assure same QoS as that of a lower modulation scheme, say QPSK, see Fig. \ref{fig:prob}.

\textbf{Representing Modulations by Sigmoidal-like Utility Functions:} In order to further motivate the use of sigmoidal-like utility functions, in Fig. \ref{fig:prob}, we provide the probability of packet transmission success for different modulation schemes such as QPSK, 16-QAM, 64-QAM, and 256-QAM. We assume that our packets consist of 800 symbols and we set $P_T=31$. The probability of packet success can be given by $f_i(\gamma_i(P_i)) = \text{Prob} (\gamma_i(P_i) \geq \Gamma) $, where $\Gamma$ is some pre-established threshold. This packet transmission success probability depends on many parameters including modulation schemes, coding rate, packet size, hybrid automatic repeat request (H-ARQ) schemes, SINR, and power. It is important to note that the cumulative distribution function of modulation schemes has a sigmoidal-like shape i.e. first convex and then concave. Thus, a modulation can indeed be represented by sigmoidal-like utility function of its power allocation \cite{DL_PowerAllocation}. For this 
reason, we represent different user modulations by sigmoidal-
like utility functions.

\section{Problem Formulation}\label{sec:Problem_formulation}

We consider the utility proportional fairness objective function given by 
\begin{equation}\label{eqn:utility_fairness}
\underset{\textbf{P}}\max \prod_{i=1}^{M}U_i(\gamma_i(P_i))
\end{equation}
where $\textbf{P} =\{P_1,P_2,...,P_M\}$ and $M$ is the number of UEs in the coverage area of the BS. The goal of this power allocation objective function is to allocate power to each UE that maximizes the total mobile system objective (i.e. the product of the utilities of all the UEs) while ensuring proportional fairness among individual utilities. This power allocation objective function ensures non-zero power allocation for all users. Therefore, the corresponding power allocation optimization problem guarantees minimum QoS for all users. In addition, this approach allocates more power to users with lower modulation schemes providing improvement in the QoS of cellular system. 


\textbf{Optimization Problem:} The basic formulation of the utility proportional fairness power allocation problem is given by the following optimization problem with two constraints:
\begin{equation}\label{eqn:opt_prob_fairness}
\begin{aligned}
& \underset{\textbf{P}}{\text{max}}
& & \prod_{i=1}^{M}U_i(\gamma_i(P_i)) \\
& \text{subject to}
& & \sum_{i=1}^{M}P_i \leq P_T\\
& & &  P_i \geq 0, \;\;\;\;\; \text{for} \; \; i = 1,2, ...,M \; \; \text{and} \; \; P_T \geq 0.
\end{aligned}
\end{equation}
where $P_T$ is the total power of the BS covering the $M$ UEs, and $\textbf{P} =\{P_1,P_2,...,P_M\}$. Our optimization problem has two constraints which are discussed as follows.

\textbf{Total BS Power Constraint i.e. $\sum_{i=1}^{M}P_i \leq P_T$:} The BS has to allocate powers to all users by staying within its available power budget.

\textbf{Minimum QoS Constraint i.e. $ P_i \geq 0$ for $i = 1,2, ...,M$ and $P_T \geq 0$:} We ensure that all UEs are served by the BS by allocating non-zero powers to all UEs, i.e., when $P_T \neq 0 $, $P_i > 0$ for all users. This is to meet a minimum QoS criteria for all users of the network. The case $P_i = 0$ is only when $P_T = 0$ and is included to make the problem as general as possible.

We prove in Section \ref{sec:Proof} that there exists a tractable global optimal solution to the optimization problem (\ref{eqn:opt_prob_fairness}).

\section{The Global Optimal Solution}\label{sec:Proof}

In the optimization problem (\ref{eqn:opt_prob_fairness}), since the objective function $\arg \underset{\textbf{P}} \max \prod_{i=1}^{M}U_i(\gamma_i(P_i))$ is equivalent to $\arg \underset{\textbf{P}} \max \sum_{i=1}^{M}\log(U_i(\gamma_i(P_i)))$, it can be expressed as:
\begin{equation}\label{eqn:opt_prob_fairness_mod}
\begin{aligned}
& \underset{\textbf{P}}{\text{max}}
& & \sum_{i=1}^{M}\log(U_i(\gamma_i(P_i))) \\
& \text{subject to}
& & \sum_{i=1}^{M}P_i \leq P_T\\
& & &  P_i \geq 0, \;\;\;\;\; \text{for} \; \; i = 1,2, ...,M \; \; \text{and} \; \; P_T \geq 0.
\end{aligned}
\end{equation}

\begin{lem}\label{lem:concavity}
The utility functions $\log(U_i(\gamma_i(P_i)))$, in the optimization problem (\ref{eqn:opt_prob_fairness_mod}), are strictly concave functions. 
\end{lem}
\begin{proof}
In this paper, we assume that all the utility functions of the UEs are sigmoidal-like functions. The utility function of the normalized sigmoidal-like function is given by equation (\ref{eqn:sigmoid}) as $U_i(\gamma_i(P_i)) = c\Big(\frac{1}{1+e^{-a_i(P_i-b_i)}}-d\Big)$. For $0<P_i<P_T$, we have 
\begin{equation*}\label{eqn:sigmoid_bound}
0<1-d_i({1+e^{-a_i(P_i-b_i)}})<\frac{1}{1+c_id_i}
\end{equation*}
It follows that for $0<P_i<P_T$, we have the first and second derivative as
\begin{equation*}\label{eqn:sigmoid_derivative}
\begin{aligned}
\frac{d}{dP_i}\log U_i(\gamma_i(P_i)) =& \frac{a_id_i e^{-a_i(P_i-b_i)}}{1-d_i(1+e^{-a_i(P_i-b_i)})} + \frac{a_ie^{-a_i(P_i-b_i)}}{(1+e^{-a_i(P_i-b_i)})}>0\\
\frac{d^2}{dP_i^2}\log U_i(\gamma_i(P_i)) =& \frac{-a_i^2d_ie^{-a_i(P_i-b_i)}}{c_i\Big(1-d_i(1+e^{-a(P_i-b_i)})\Big)^2} + \frac{-a_i^2e^{-a_i(P_i-b_i)}}{(1+e^{-a_i(P_i-b_i)})^2} < 0. \\
\end{aligned}
\end{equation*}
Therefore, the sigmoidal-like utility function's $U_i(\gamma_i(P_i))$ natural logarithm  $\log(U_i(\gamma_i(P_i)))$ is strictly concave function. Therefore, all the utility functions in our system model have strictly concave natural logarithms.
\end{proof}

The natural logarithms of the utility functions of Figure \ref{fig:sigmoid} are shown in Figure \ref{fig:log_sigmoid} and the derivatives of  natural logarithms of the utility functions are shown in Figure \ref{fig:diff_log_sigmoid}.
\begin{thm}\label{thm:global_soln}
The optimization problem (\ref{eqn:opt_prob_fairness}) is a convex optimization problem and there exists a unique tractable global optimal solution. 
\end{thm}
\begin{proof}
It follows from Lemma \ref{lem:concavity} that all UEs' utility functions are strictly concave. Therefore, the optimization problem (\ref{eqn:opt_prob_fairness_mod}) is a convex optimization problem \cite{Boyd2004}. The optimization problem (\ref{eqn:opt_prob_fairness_mod}) is equivalent to the optimization problem (\ref{eqn:opt_prob_fairness}), therefore it is also a convex optimization problem. For a convex optimization problem there exists a unique tractable global optimal solution \cite{Boyd2004}.
\end{proof}

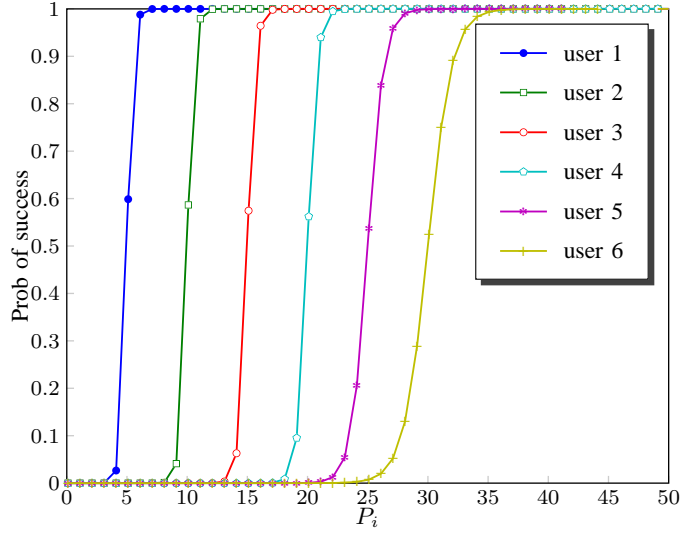
\begin{figure}[!t]
\centering

\psset{xunit=0.020000\plotwidth,yunit=0.788710\plotwidth}%
\begin{pspicture}(-5.529954,-0.111111)(50.921659,1.023392)%


\psline[linewidth=\AxesLineWidth,linecolor=GridColor](0.000000,0.000000)(0.000000,0.015215)
\psline[linewidth=\AxesLineWidth,linecolor=GridColor](5.000000,0.000000)(5.000000,0.015215)
\psline[linewidth=\AxesLineWidth,linecolor=GridColor](10.000000,0.000000)(10.000000,0.015215)
\psline[linewidth=\AxesLineWidth,linecolor=GridColor](15.000000,0.000000)(15.000000,0.015215)
\psline[linewidth=\AxesLineWidth,linecolor=GridColor](20.000000,0.000000)(20.000000,0.015215)
\psline[linewidth=\AxesLineWidth,linecolor=GridColor](25.000000,0.000000)(25.000000,0.015215)
\psline[linewidth=\AxesLineWidth,linecolor=GridColor](30.000000,0.000000)(30.000000,0.015215)
\psline[linewidth=\AxesLineWidth,linecolor=GridColor](35.000000,0.000000)(35.000000,0.015215)
\psline[linewidth=\AxesLineWidth,linecolor=GridColor](40.000000,0.000000)(40.000000,0.015215)
\psline[linewidth=\AxesLineWidth,linecolor=GridColor](45.000000,0.000000)(45.000000,0.015215)
\psline[linewidth=\AxesLineWidth,linecolor=GridColor](50.000000,0.000000)(50.000000,0.015215)
\psline[linewidth=\AxesLineWidth,linecolor=GridColor](0.000000,0.000000)(0.600000,0.000000)
\psline[linewidth=\AxesLineWidth,linecolor=GridColor](0.000000,0.100000)(0.600000,0.100000)
\psline[linewidth=\AxesLineWidth,linecolor=GridColor](0.000000,0.200000)(0.600000,0.200000)
\psline[linewidth=\AxesLineWidth,linecolor=GridColor](0.000000,0.300000)(0.600000,0.300000)
\psline[linewidth=\AxesLineWidth,linecolor=GridColor](0.000000,0.400000)(0.600000,0.400000)
\psline[linewidth=\AxesLineWidth,linecolor=GridColor](0.000000,0.500000)(0.600000,0.500000)
\psline[linewidth=\AxesLineWidth,linecolor=GridColor](0.000000,0.600000)(0.600000,0.600000)
\psline[linewidth=\AxesLineWidth,linecolor=GridColor](0.000000,0.700000)(0.600000,0.700000)
\psline[linewidth=\AxesLineWidth,linecolor=GridColor](0.000000,0.800000)(0.600000,0.800000)
\psline[linewidth=\AxesLineWidth,linecolor=GridColor](0.000000,0.900000)(0.600000,0.900000)
\psline[linewidth=\AxesLineWidth,linecolor=GridColor](0.000000,1.000000)(0.600000,1.000000)

{ \footnotesize 
\rput[t](0.000000,-0.015215){$0$}
\rput[t](5.000000,-0.015215){$5$}
\rput[t](10.000000,-0.015215){$10$}
\rput[t](15.000000,-0.015215){$15$}
\rput[t](20.000000,-0.015215){$20$}
\rput[t](25.000000,-0.015215){$25$}
\rput[t](30.000000,-0.015215){$30$}
\rput[t](35.000000,-0.015215){$35$}
\rput[t](40.000000,-0.015215){$40$}
\rput[t](45.000000,-0.015215){$45$}
\rput[t](50.000000,-0.015215){$50$}
\rput[r](-0.600000,0.000000){$0$}
\rput[r](-0.600000,0.100000){$0.1$}
\rput[r](-0.600000,0.200000){$0.2$}
\rput[r](-0.600000,0.300000){$0.3$}
\rput[r](-0.600000,0.400000){$0.4$}
\rput[r](-0.600000,0.500000){$0.5$}
\rput[r](-0.600000,0.600000){$0.6$}
\rput[r](-0.600000,0.700000){$0.7$}
\rput[r](-0.600000,0.800000){$0.8$}
\rput[r](-0.600000,0.900000){$0.9$}
\rput[r](-0.600000,1.000000){$1$}
} 

\psframe[linewidth=\AxesLineWidth,dimen=middle](0.000000,0.000000)(50.000000,1.000000)

{ \small 
\rput[b](25.000000,-0.111111){
\begin{tabular}{c}
$P_i$\\
\end{tabular}
}

\rput[t]{90}(-5.529954,0.500000){
\begin{tabular}{c}
Prob of success\\
\end{tabular}
}
} 

\newrgbcolor{color403.0076}{0  0  1}
\psline[plotstyle=line,linejoin=1,showpoints=false,dotstyle=*,dotsize=\MarkerSize,linestyle=solid,linewidth=\LineWidth,linecolor=color403.0076]
(49.100000,1.000000)(50.000000,1.000000)
\psline[plotstyle=line,linejoin=1,showpoints=true,dotstyle=*,dotsize=\MarkerSize,linestyle=solid,linewidth=\LineWidth,linecolor=color403.0076]
(0.100000,0.000000)(1.100000,0.000000)(2.100000,0.000009)(3.100000,0.000500)(4.100000,0.026597)
(5.100000,0.598688)(6.100000,0.987872)(7.100000,0.999775)(8.100000,0.999996)(9.100000,1.000000)
(10.100000,1.000000)(11.100000,1.000000)(12.100000,1.000000)(13.100000,1.000000)(14.100000,1.000000)
(15.100000,1.000000)(16.100000,1.000000)(17.100000,1.000000)(18.100000,1.000000)(19.100000,1.000000)
(20.100000,1.000000)(21.100000,1.000000)(22.100000,1.000000)(23.100000,1.000000)(24.100000,1.000000)
(25.100000,1.000000)(26.100000,1.000000)(27.100000,1.000000)(28.100000,1.000000)(29.100000,1.000000)
(30.100000,1.000000)(31.100000,1.000000)(32.100000,1.000000)(33.100000,1.000000)(34.100000,1.000000)
(35.100000,1.000000)(36.100000,1.000000)(37.100000,1.000000)(38.100000,1.000000)(39.100000,1.000000)
(40.100000,1.000000)(41.100000,1.000000)(42.100000,1.000000)(43.100000,1.000000)(44.100000,1.000000)
(45.100000,1.000000)(46.100000,1.000000)(47.100000,1.000000)(48.100000,1.000000)(49.100000,1.000000)

\newrgbcolor{color404.0071}{0         0.5           0}
\psline[plotstyle=line,linejoin=1,showpoints=false,dotstyle=Bsquare,dotsize=\MarkerSize,linestyle=solid,linewidth=\LineWidth,linecolor=color404.0071]
(49.100000,1.000000)(50.000000,1.000000)
\psline[plotstyle=line,linejoin=1,showpoints=true,dotstyle=Bsquare,dotsize=\MarkerSize,linestyle=solid,linewidth=\LineWidth,linecolor=color404.0071]
(0.100000,0.000000)(1.100000,0.000000)(2.100000,0.000000)(3.100000,0.000000)(4.100000,0.000000)
(5.100000,0.000000)(6.100000,0.000001)(7.100000,0.000039)(8.100000,0.001292)(9.100000,0.041091)
(10.100000,0.586618)(11.100000,0.979164)(12.100000,0.999358)(13.100000,0.999981)(14.100000,0.999999)
(15.100000,1.000000)(16.100000,1.000000)(17.100000,1.000000)(18.100000,1.000000)(19.100000,1.000000)
(20.100000,1.000000)(21.100000,1.000000)(22.100000,1.000000)(23.100000,1.000000)(24.100000,1.000000)
(25.100000,1.000000)(26.100000,1.000000)(27.100000,1.000000)(28.100000,1.000000)(29.100000,1.000000)
(30.100000,1.000000)(31.100000,1.000000)(32.100000,1.000000)(33.100000,1.000000)(34.100000,1.000000)
(35.100000,1.000000)(36.100000,1.000000)(37.100000,1.000000)(38.100000,1.000000)(39.100000,1.000000)
(40.100000,1.000000)(41.100000,1.000000)(42.100000,1.000000)(43.100000,1.000000)(44.100000,1.000000)
(45.100000,1.000000)(46.100000,1.000000)(47.100000,1.000000)(48.100000,1.000000)(49.100000,1.000000)

\newrgbcolor{color405.0071}{1  0  0}
\psline[plotstyle=line,linejoin=1,showpoints=false,dotstyle=Bo,dotsize=\MarkerSize,linestyle=solid,linewidth=\LineWidth,linecolor=color405.0071]
(49.100000,1.000000)(50.000000,1.000000)
\psline[plotstyle=line,linejoin=1,showpoints=true,dotstyle=Bo,dotsize=\MarkerSize,linestyle=solid,linewidth=\LineWidth,linecolor=color405.0071]
(0.100000,0.000000)(1.100000,0.000000)(2.100000,0.000000)(3.100000,0.000000)(4.100000,0.000000)
(5.100000,0.000000)(6.100000,0.000000)(7.100000,0.000000)(8.100000,0.000000)(9.100000,0.000000)
(10.100000,0.000000)(11.100000,0.000008)(12.100000,0.000167)(13.100000,0.003335)(14.100000,0.062973)
(15.100000,0.574443)(16.100000,0.964429)(17.100000,0.998167)(18.100000,0.999909)(19.100000,0.999995)
(20.100000,1.000000)(21.100000,1.000000)(22.100000,1.000000)(23.100000,1.000000)(24.100000,1.000000)
(25.100000,1.000000)(26.100000,1.000000)(27.100000,1.000000)(28.100000,1.000000)(29.100000,1.000000)
(30.100000,1.000000)(31.100000,1.000000)(32.100000,1.000000)(33.100000,1.000000)(34.100000,1.000000)
(35.100000,1.000000)(36.100000,1.000000)(37.100000,1.000000)(38.100000,1.000000)(39.100000,1.000000)
(40.100000,1.000000)(41.100000,1.000000)(42.100000,1.000000)(43.100000,1.000000)(44.100000,1.000000)
(45.100000,1.000000)(46.100000,1.000000)(47.100000,1.000000)(48.100000,1.000000)(49.100000,1.000000)

\newrgbcolor{color406.0071}{0        0.75        0.75}
\psline[plotstyle=line,linejoin=1,showpoints=false,dotstyle=Bpentagon,dotsize=\MarkerSize,linestyle=solid,linewidth=\LineWidth,linecolor=color406.0071]
(49.100000,1.000000)(50.000000,1.000000)
\psline[plotstyle=line,linejoin=1,showpoints=true,dotstyle=Bpentagon,dotsize=\MarkerSize,linestyle=solid,linewidth=\LineWidth,linecolor=color406.0071]
(0.100000,0.000000)(1.100000,0.000000)(2.100000,0.000000)(3.100000,0.000000)(4.100000,0.000000)
(5.100000,0.000000)(6.100000,0.000000)(7.100000,0.000000)(8.100000,0.000000)(9.100000,0.000000)
(10.100000,0.000000)(11.100000,0.000000)(12.100000,0.000000)(13.100000,0.000000)(14.100000,0.000000)
(15.100000,0.000005)(16.100000,0.000058)(17.100000,0.000710)(18.100000,0.008577)(19.100000,0.095349)
(20.100000,0.562177)(21.100000,0.939913)(22.100000,0.994780)(23.100000,0.999569)(24.100000,0.999965)
(25.100000,0.999997)(26.100000,1.000000)(27.100000,1.000000)(28.100000,1.000000)(29.100000,1.000000)
(30.100000,1.000000)(31.100000,1.000000)(32.100000,1.000000)(33.100000,1.000000)(34.100000,1.000000)
(35.100000,1.000000)(36.100000,1.000000)(37.100000,1.000000)(38.100000,1.000000)(39.100000,1.000000)
(40.100000,1.000000)(41.100000,1.000000)(42.100000,1.000000)(43.100000,1.000000)(44.100000,1.000000)
(45.100000,1.000000)(46.100000,1.000000)(47.100000,1.000000)(48.100000,1.000000)(49.100000,1.000000)

\newrgbcolor{color407.0071}{0.75           0        0.75}
\psline[plotstyle=line,linejoin=1,showpoints=false,dotstyle=Basterisk,dotsize=\MarkerSize,linestyle=solid,linewidth=\LineWidth,linecolor=color407.0071]
(49.100000,1.000000)(50.000000,1.000000)
\psline[plotstyle=line,linejoin=1,showpoints=true,dotstyle=Basterisk,dotsize=\MarkerSize,linestyle=solid,linewidth=\LineWidth,linecolor=color407.0071]
(0.100000,0.000000)(1.100000,0.000000)(2.100000,0.000000)(3.100000,0.000000)(4.100000,0.000000)
(5.100000,0.000000)(6.100000,0.000000)(7.100000,0.000000)(8.100000,0.000000)(9.100000,0.000000)
(10.100000,0.000000)(11.100000,0.000000)(12.100000,0.000000)(13.100000,0.000000)(14.100000,0.000000)
(15.100000,0.000000)(16.100000,0.000002)(17.100000,0.000007)(18.100000,0.000032)(19.100000,0.000143)
(20.100000,0.000642)(21.100000,0.002872)(22.100000,0.012742)(23.100000,0.054681)(24.100000,0.205870)
(25.100000,0.537430)(26.100000,0.838891)(27.100000,0.958909)(28.100000,0.990529)(29.100000,0.997871)
(30.100000,0.999524)(31.100000,0.999894)(32.100000,0.999976)(33.100000,0.999995)(34.100000,0.999999)
(35.100000,1.000000)(36.100000,1.000000)(37.100000,1.000000)(38.100000,1.000000)(39.100000,1.000000)
(40.100000,1.000000)(41.100000,1.000000)

\newrgbcolor{color408.0071}{0.75        0.75           0}
\psline[plotstyle=line,linejoin=1,showpoints=false,dotstyle=B+,dotsize=\MarkerSize,linestyle=solid,linewidth=\LineWidth,linecolor=color408.0071]
(49.100000,1.000000)(50.000000,1.000000)
\psline[plotstyle=line,linejoin=1,showpoints=true,dotstyle=B+,dotsize=\MarkerSize,linestyle=solid,linewidth=\LineWidth,linecolor=color408.0071]
(0.100000,0.000000)(1.100000,0.000000)(2.100000,0.000000)(3.100000,0.000000)(4.100000,0.000000)
(5.100000,0.000000)(6.100000,0.000000)(7.100000,0.000000)(8.100000,0.000000)(9.100000,0.000000)
(10.100000,0.000000)(11.100000,0.000000)(12.100000,0.000000)(13.100000,0.000000)(14.100000,0.000000)
(15.100000,0.000000)(16.100000,0.000001)(17.100000,0.000002)(18.100000,0.000007)(19.100000,0.000018)
(20.100000,0.000050)(21.100000,0.000136)(22.100000,0.000371)(23.100000,0.001007)(24.100000,0.002732)
(25.100000,0.007392)(26.100000,0.019840)(27.100000,0.052154)(28.100000,0.130108)(29.100000,0.289050)
(30.100000,0.524979)(31.100000,0.750260)(32.100000,0.890903)(33.100000,0.956893)(34.100000,0.983698)
(35.100000,0.993940)(36.100000,0.997762)(37.100000,0.999176)(38.100000,0.999697)(39.100000,0.999888)
(40.100000,0.999959)(41.100000,0.999985)(42.100000,0.999994)(43.100000,0.999998)(44.100000,0.999999)

{ \small 
\rput[tr](48.800000,0.969571){%
\psshadowbox[framesep=0pt,linewidth=\AxesLineWidth]{\psframebox*{\begin{tabular}{l}
\Rnode{a1}{\hspace*{0.0ex}} \hspace*{0.4cm} \Rnode{a2}{~~user 1} \\
\Rnode{a3}{\hspace*{0.0ex}} \hspace*{0.4cm} \Rnode{a4}{~~user 2} \\
\Rnode{a5}{\hspace*{0.0ex}} \hspace*{0.4cm} \Rnode{a6}{~~user 3} \\
\Rnode{a7}{\hspace*{0.0ex}} \hspace*{0.4cm} \Rnode{a8}{~~user 4} \\
\Rnode{a9}{\hspace*{0.0ex}} \hspace*{0.4cm} \Rnode{a10}{~~user 5} \\
\Rnode{a11}{\hspace*{0.0ex}} \hspace*{0.4cm} \Rnode{a12}{~~user 6} \\
\end{tabular}}
\ncline[linestyle=solid,linewidth=\LineWidth,linecolor=color403.0076]{a1}{a2} \ncput{\psdot[dotstyle=*,dotsize=\MarkerSize,linecolor=color403.0076]}
\ncline[linestyle=solid,linewidth=\LineWidth,linecolor=color404.0071]{a3}{a4} \ncput{\psdot[dotstyle=Bsquare,dotsize=\MarkerSize,linecolor=color404.0071]}
\ncline[linestyle=solid,linewidth=\LineWidth,linecolor=color405.0071]{a5}{a6} \ncput{\psdot[dotstyle=Bo,dotsize=\MarkerSize,linecolor=color405.0071]}
\ncline[linestyle=solid,linewidth=\LineWidth,linecolor=color406.0071]{a7}{a8} \ncput{\psdot[dotstyle=Bpentagon,dotsize=\MarkerSize,linecolor=color406.0071]}
\ncline[linestyle=solid,linewidth=\LineWidth,linecolor=color407.0071]{a9}{a10} \ncput{\psdot[dotstyle=Basterisk,dotsize=\MarkerSize,linecolor=color407.0071]}
\ncline[linestyle=solid,linewidth=\LineWidth,linecolor=color408.0071]{a11}{a12} \ncput{\psdot[dotstyle=B+,dotsize=\MarkerSize,linecolor=color408.0071]}
}%
}%
} 

\end{pspicture}%
\caption{The sigmoidal-like utility functions (representing users with different modulation schemes) $U_i(\gamma_i(P_i))$. We use sigmoidal-like utility functions as their shape resembles cumulative distribution function of successful packet transmission of modulation schemes, see Figure \ref{fig:prob}.}
\label{fig:sigmoid}
\end{figure}
\begin{figure}[!t]
\centering

%
\psset{xunit=0.020000\plotwidth,yunit=0.013145\plotwidth}%
\begin{pspicture}(-5.529954,-66.666667)(50.921659,1.403509)%


\psline[linewidth=\AxesLineWidth,linecolor=GridColor](0.000000,-60.000000)(0.000000,-59.087117)
\psline[linewidth=\AxesLineWidth,linecolor=GridColor](5.000000,-60.000000)(5.000000,-59.087117)
\psline[linewidth=\AxesLineWidth,linecolor=GridColor](10.000000,-60.000000)(10.000000,-59.087117)
\psline[linewidth=\AxesLineWidth,linecolor=GridColor](15.000000,-60.000000)(15.000000,-59.087117)
\psline[linewidth=\AxesLineWidth,linecolor=GridColor](20.000000,-60.000000)(20.000000,-59.087117)
\psline[linewidth=\AxesLineWidth,linecolor=GridColor](25.000000,-60.000000)(25.000000,-59.087117)
\psline[linewidth=\AxesLineWidth,linecolor=GridColor](30.000000,-60.000000)(30.000000,-59.087117)
\psline[linewidth=\AxesLineWidth,linecolor=GridColor](35.000000,-60.000000)(35.000000,-59.087117)
\psline[linewidth=\AxesLineWidth,linecolor=GridColor](40.000000,-60.000000)(40.000000,-59.087117)
\psline[linewidth=\AxesLineWidth,linecolor=GridColor](45.000000,-60.000000)(45.000000,-59.087117)
\psline[linewidth=\AxesLineWidth,linecolor=GridColor](50.000000,-60.000000)(50.000000,-59.087117)
\psline[linewidth=\AxesLineWidth,linecolor=GridColor](0.000000,-60.000000)(0.600000,-60.000000)
\psline[linewidth=\AxesLineWidth,linecolor=GridColor](0.000000,-50.000000)(0.600000,-50.000000)
\psline[linewidth=\AxesLineWidth,linecolor=GridColor](0.000000,-40.000000)(0.600000,-40.000000)
\psline[linewidth=\AxesLineWidth,linecolor=GridColor](0.000000,-30.000000)(0.600000,-30.000000)
\psline[linewidth=\AxesLineWidth,linecolor=GridColor](0.000000,-20.000000)(0.600000,-20.000000)
\psline[linewidth=\AxesLineWidth,linecolor=GridColor](0.000000,-10.000000)(0.600000,-10.000000)
\psline[linewidth=\AxesLineWidth,linecolor=GridColor](0.000000,0.000000)(0.600000,0.000000)

{ \footnotesize 
\rput[t](0.000000,-60.912883){$0$}
\rput[t](5.000000,-60.912883){$5$}
\rput[t](10.000000,-60.912883){$10$}
\rput[t](15.000000,-60.912883){$15$}
\rput[t](20.000000,-60.912883){$20$}
\rput[t](25.000000,-60.912883){$25$}
\rput[t](30.000000,-60.912883){$30$}
\rput[t](35.000000,-60.912883){$35$}
\rput[t](40.000000,-60.912883){$40$}
\rput[t](45.000000,-60.912883){$45$}
\rput[t](50.000000,-60.912883){$50$}
\rput[r](-0.600000,-60.000000){$-60$}
\rput[r](-0.600000,-50.000000){$-50$}
\rput[r](-0.600000,-40.000000){$-40$}
\rput[r](-0.600000,-30.000000){$-30$}
\rput[r](-0.600000,-20.000000){$-20$}
\rput[r](-0.600000,-10.000000){$-10$}
\rput[r](-0.600000,0.000000){$0$}
} 

\psframe[linewidth=\AxesLineWidth,dimen=middle](0.000000,-60.000000)(50.000000,0.000000)

{ \small 
\rput[b](25.000000,-66.666667){
\begin{tabular}{c}
$P_i$\\
\end{tabular}
}

\rput[t]{90}(-5.529954,-30.000000){
\begin{tabular}{c}
$\log U_i$\\
\end{tabular}
}
} 

\newrgbcolor{color71.0027}{0  0  1}
\psline[plotstyle=line,linejoin=1,showpoints=false,dotstyle=*,dotsize=\MarkerSize,linestyle=solid,linewidth=\LineWidth,linecolor=color71.0027]
(49.100000,0.000000)(50.000000,0.000000)
\psline[plotstyle=line,linejoin=1,showpoints=true,dotstyle=*,dotsize=\MarkerSize,linestyle=solid,linewidth=\LineWidth,linecolor=color71.0027]
(0.100000,-20.709633)(1.100000,-15.612353)(2.100000,-11.600234)(3.100000,-7.600504)(4.100000,-3.626957)
(5.100000,-0.513015)(6.100000,-0.012203)(7.100000,-0.000225)(8.100000,-0.000004)(9.100000,-0.000000)
(10.100000,-0.000000)(11.100000,-0.000000)(12.100000,-0.000000)(13.100000,-0.000000)(14.100000,-0.000000)
(15.100000,0.000000)(16.100000,0.000000)(17.100000,0.000000)(18.100000,0.000000)(19.100000,0.000000)
(20.100000,0.000000)(21.100000,0.000000)(22.100000,0.000000)(23.100000,0.000000)(24.100000,0.000000)
(25.100000,0.000000)(26.100000,0.000000)(27.100000,0.000000)(28.100000,0.000000)(29.100000,0.000000)
(30.100000,0.000000)(31.100000,0.000000)(32.100000,0.000000)(33.100000,0.000000)(34.100000,0.000000)
(35.100000,0.000000)(36.100000,0.000000)(37.100000,0.000000)(38.100000,0.000000)(39.100000,0.000000)
(40.100000,0.000000)(41.100000,0.000000)(42.100000,0.000000)(43.100000,0.000000)(44.100000,0.000000)
(45.100000,0.000000)(46.100000,0.000000)(47.100000,0.000000)(48.100000,0.000000)(49.100000,0.000000)

\newrgbcolor{color72.0022}{0         0.5           0}
\psline[plotstyle=line,linejoin=1,showpoints=false,dotstyle=Bsquare,dotsize=\MarkerSize,linestyle=solid,linewidth=\LineWidth,linecolor=color72.0022]
(49.100000,-0.000000)(50.000000,-0.000000)
\psline[plotstyle=line,linejoin=1,showpoints=true,dotstyle=Bsquare,dotsize=\MarkerSize,linestyle=solid,linewidth=\LineWidth,linecolor=color72.0022]
(0.100000,-35.869723)(1.100000,-31.171509)(2.100000,-27.650643)(3.100000,-24.150019)(4.100000,-20.650001)
(5.100000,-17.150000)(6.100000,-13.650001)(7.100000,-10.150039)(8.100000,-6.651293)(9.100000,-3.191959)
(10.100000,-0.533382)(11.100000,-0.021056)(12.100000,-0.000642)(13.100000,-0.000019)(14.100000,-0.000001)
(15.100000,-0.000000)(16.100000,-0.000000)(17.100000,-0.000000)(18.100000,-0.000000)(19.100000,-0.000000)
(20.100000,-0.000000)(21.100000,-0.000000)(22.100000,-0.000000)(23.100000,-0.000000)(24.100000,-0.000000)
(25.100000,-0.000000)(26.100000,-0.000000)(27.100000,-0.000000)(28.100000,-0.000000)(29.100000,-0.000000)
(30.100000,-0.000000)(31.100000,-0.000000)(32.100000,-0.000000)(33.100000,-0.000000)(34.100000,-0.000000)
(35.100000,-0.000000)(36.100000,-0.000000)(37.100000,-0.000000)(38.100000,-0.000000)(39.100000,-0.000000)
(40.100000,-0.000000)(41.100000,-0.000000)(42.100000,-0.000000)(43.100000,-0.000000)(44.100000,-0.000000)
(45.100000,-0.000000)(46.100000,-0.000000)(47.100000,-0.000000)(48.100000,-0.000000)(49.100000,-0.000000)

\newrgbcolor{color73.0022}{1  0  0}
\psline[plotstyle=line,linejoin=1,showpoints=false,dotstyle=Bo,dotsize=\MarkerSize,linestyle=solid,linewidth=\LineWidth,linecolor=color73.0022]
(49.100000,0.000000)(50.000000,0.000000)
\psline[plotstyle=line,linejoin=1,showpoints=true,dotstyle=Bo,dotsize=\MarkerSize,linestyle=solid,linewidth=\LineWidth,linecolor=color73.0022]
(0.100000,-46.050226)(1.100000,-41.737581)(2.100000,-38.701838)(3.100000,-35.700091)(4.100000,-32.700005)
(5.100000,-29.700000)(6.100000,-26.700000)(7.100000,-23.700000)(8.100000,-20.700000)(9.100000,-17.700000)
(10.100000,-14.700000)(11.100000,-11.700008)(12.100000,-8.700167)(13.100000,-5.703340)(14.100000,-2.765044)
(15.100000,-0.554355)(16.100000,-0.036219)(17.100000,-0.001835)(18.100000,-0.000091)(19.100000,-0.000005)
(20.100000,-0.000000)(21.100000,-0.000000)(22.100000,-0.000000)(23.100000,-0.000000)(24.100000,-0.000000)
(25.100000,-0.000000)(26.100000,-0.000000)(27.100000,-0.000000)(28.100000,0.000000)(29.100000,0.000000)
(30.100000,0.000000)(31.100000,0.000000)(32.100000,0.000000)(33.100000,0.000000)(34.100000,0.000000)
(35.100000,0.000000)(36.100000,0.000000)(37.100000,0.000000)(38.100000,0.000000)(39.100000,0.000000)
(40.100000,0.000000)(41.100000,0.000000)(42.100000,0.000000)(43.100000,0.000000)(44.100000,0.000000)
(45.100000,0.000000)(46.100000,0.000000)(47.100000,0.000000)(48.100000,0.000000)(49.100000,0.000000)

\newrgbcolor{color74.0022}{0        0.75        0.75}
\psline[plotstyle=line,linejoin=1,showpoints=false,dotstyle=Bpentagon,dotsize=\MarkerSize,linestyle=solid,linewidth=\LineWidth,linecolor=color74.0022]
(49.100000,0.000000)(50.000000,0.000000)
\psline[plotstyle=line,linejoin=1,showpoints=true,dotstyle=Bpentagon,dotsize=\MarkerSize,linestyle=solid,linewidth=\LineWidth,linecolor=color74.0022]
(0.100000,-51.258692)(1.100000,-47.316063)(2.100000,-44.755261)(3.100000,-42.250431)(4.100000,-39.750035)
(5.100000,-37.250003)(6.100000,-34.750000)(7.100000,-32.250000)(8.100000,-29.750000)(9.100000,-27.250000)
(10.100000,-24.750000)(11.100000,-22.250000)(12.100000,-19.750000)(13.100000,-17.250000)(14.100000,-14.750000)
(15.100000,-12.250005)(16.100000,-9.750058)(17.100000,-7.250710)(18.100000,-4.758614)(19.100000,-2.350207)
(20.100000,-0.575939)(21.100000,-0.061968)(22.100000,-0.005234)(23.100000,-0.000431)(24.100000,-0.000035)
(25.100000,-0.000003)(26.100000,-0.000000)(27.100000,-0.000000)(28.100000,-0.000000)(29.100000,-0.000000)
(30.100000,-0.000000)(31.100000,-0.000000)(32.100000,-0.000000)(33.100000,-0.000000)(34.100000,-0.000000)
(35.100000,0.000000)(36.100000,0.000000)(37.100000,0.000000)(38.100000,0.000000)(39.100000,0.000000)
(40.100000,0.000000)(41.100000,0.000000)(42.100000,0.000000)(43.100000,0.000000)(44.100000,0.000000)
(45.100000,0.000000)(46.100000,0.000000)(47.100000,0.000000)(48.100000,0.000000)(49.100000,0.000000)

\newrgbcolor{color75.0022}{0.75           0        0.75}
\psline[plotstyle=line,linejoin=1,showpoints=false,dotstyle=Basterisk,dotsize=\MarkerSize,linestyle=solid,linewidth=\LineWidth,linecolor=color75.0022]
(49.100000,-0.000000)(50.000000,-0.000000)
\psline[plotstyle=line,linejoin=1,showpoints=true,dotstyle=Basterisk,dotsize=\MarkerSize,linestyle=solid,linewidth=\LineWidth,linecolor=color75.0022]
(0.100000,-39.321183)(1.100000,-36.063255)(2.100000,-34.393797)(3.100000,-32.859608)(4.100000,-31.352136)
(5.100000,-29.850476)(6.100000,-28.350106)(7.100000,-26.850024)(8.100000,-25.350005)(9.100000,-23.850001)
(10.100000,-22.350000)(11.100000,-20.850000)(12.100000,-19.350000)(13.100000,-17.850000)(14.100000,-16.350000)
(15.100000,-14.850000)(16.100000,-13.350002)(17.100000,-11.850007)(18.100000,-10.350032)(19.100000,-8.850143)
(20.100000,-7.350642)(21.100000,-5.852876)(22.100000,-4.362824)(23.100000,-2.906233)(24.100000,-1.580509)
(25.100000,-0.620957)(26.100000,-0.175674)(27.100000,-0.041959)(28.100000,-0.009516)(29.100000,-0.002131)
(30.100000,-0.000476)(31.100000,-0.000106)(32.100000,-0.000024)(33.100000,-0.000005)(34.100000,-0.000001)
(35.100000,-0.000000)(36.100000,-0.000000)(37.100000,-0.000000)(38.100000,-0.000000)(39.100000,-0.000000)
(40.100000,-0.000000)(41.100000,-0.000000)

\newrgbcolor{color76.0022}{0.75        0.75           0}
\psline[plotstyle=line,linejoin=1,showpoints=false,dotstyle=B+,dotsize=\MarkerSize,linestyle=solid,linewidth=\LineWidth,linecolor=color76.0022]
(49.100000,-0.000000)(50.000000,-0.000000)
\psline[plotstyle=line,linejoin=1,showpoints=true,dotstyle=B+,dotsize=\MarkerSize,linestyle=solid,linewidth=\LineWidth,linecolor=color76.0022]
(0.100000,-32.252168)(1.100000,-29.304772)(2.100000,-28.030629)(3.100000,-26.946095)(4.100000,-25.916712)
(5.100000,-24.906115)(6.100000,-23.902245)(7.100000,-22.900825)(8.100000,-21.900304)(9.100000,-20.900112)
(10.100000,-19.900041)(11.100000,-18.900015)(12.100000,-17.900006)(13.100000,-16.900002)(14.100000,-15.900001)
(15.100000,-14.900001)(16.100000,-13.900001)(17.100000,-12.900003)(18.100000,-11.900007)(19.100000,-10.900018)
(20.100000,-9.900050)(21.100000,-8.900136)(22.100000,-7.900371)(23.100000,-6.901007)(24.100000,-5.902736)
(25.100000,-4.907419)(26.100000,-3.920040)(27.100000,-2.953563)(28.100000,-2.039387)(29.100000,-1.241154)
(30.100000,-0.644397)(31.100000,-0.287335)(32.100000,-0.115520)(33.100000,-0.044064)(34.100000,-0.016437)
(35.100000,-0.006078)(36.100000,-0.002240)(37.100000,-0.000825)(38.100000,-0.000303)(39.100000,-0.000112)
(40.100000,-0.000041)(41.100000,-0.000015)(42.100000,-0.000006)(43.100000,-0.000002)(44.100000,-0.000001)

{ \small 
\rput[tr](48.800000,-1.825767){%
\psshadowbox[framesep=0pt,linewidth=\AxesLineWidth]{\psframebox*{\begin{tabular}{l}
\Rnode{a1}{\hspace*{0.0ex}} \hspace*{0.4cm} \Rnode{a2}{~~user 1} \\
\Rnode{a3}{\hspace*{0.0ex}} \hspace*{0.4cm} \Rnode{a4}{~~user 2} \\
\Rnode{a5}{\hspace*{0.0ex}} \hspace*{0.4cm} \Rnode{a6}{~~user 3} \\
\Rnode{a7}{\hspace*{0.0ex}} \hspace*{0.4cm} \Rnode{a8}{~~user 4} \\
\Rnode{a9}{\hspace*{0.0ex}} \hspace*{0.4cm} \Rnode{a10}{~~user 5} \\
\Rnode{a11}{\hspace*{0.0ex}} \hspace*{0.4cm} \Rnode{a12}{~~user 6} \\
\end{tabular}}
\ncline[linestyle=solid,linewidth=\LineWidth,linecolor=color71.0027]{a1}{a2} \ncput{\psdot[dotstyle=*,dotsize=\MarkerSize,linecolor=color71.0027]}
\ncline[linestyle=solid,linewidth=\LineWidth,linecolor=color72.0022]{a3}{a4} \ncput{\psdot[dotstyle=Bsquare,dotsize=\MarkerSize,linecolor=color72.0022]}
\ncline[linestyle=solid,linewidth=\LineWidth,linecolor=color73.0022]{a5}{a6} \ncput{\psdot[dotstyle=Bo,dotsize=\MarkerSize,linecolor=color73.0022]}
\ncline[linestyle=solid,linewidth=\LineWidth,linecolor=color74.0022]{a7}{a8} \ncput{\psdot[dotstyle=Bpentagon,dotsize=\MarkerSize,linecolor=color74.0022]}
\ncline[linestyle=solid,linewidth=\LineWidth,linecolor=color75.0022]{a9}{a10} \ncput{\psdot[dotstyle=Basterisk,dotsize=\MarkerSize,linecolor=color75.0022]}
\ncline[linestyle=solid,linewidth=\LineWidth,linecolor=color76.0022]{a11}{a12} \ncput{\psdot[dotstyle=B+,dotsize=\MarkerSize,linecolor=color76.0022]}
}%
}%
} 

\end{pspicture}%

\caption{The natural logarithm of sigmoidal-like utility functions $\log U_i(\gamma_i(P_i))$ which are strictly concave. Thus all modulation schemes considered can be represented by sigmoidal-like utility functions that are strictly concave.}
\label{fig:log_sigmoid}
\end{figure}
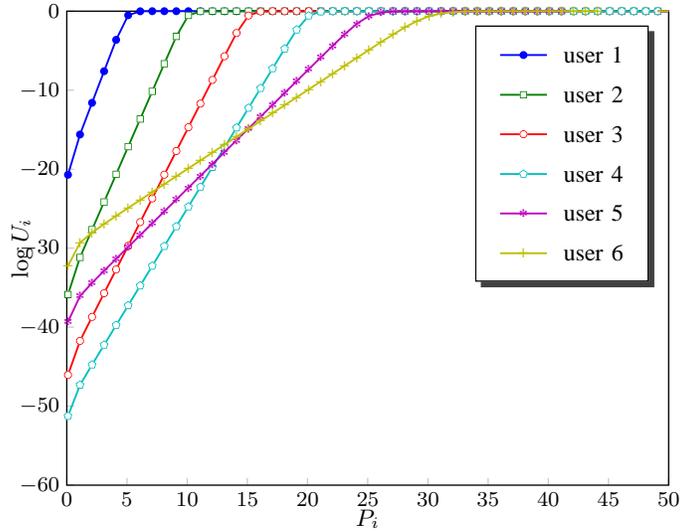
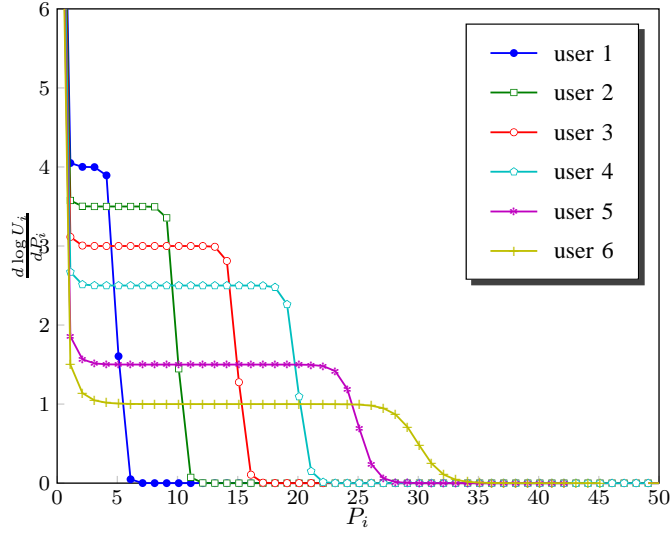
\begin{figure}[!t]
\centering
%
\psset{xunit=0.020000\plotwidth,yunit=0.131452\plotwidth}%
\begin{pspicture}(-4.262673,-0.666667)(50.921659,6.140351)%


\psline[linewidth=\AxesLineWidth,linecolor=GridColor](0.000000,0.000000)(0.000000,0.091288)
\psline[linewidth=\AxesLineWidth,linecolor=GridColor](5.000000,0.000000)(5.000000,0.091288)
\psline[linewidth=\AxesLineWidth,linecolor=GridColor](10.000000,0.000000)(10.000000,0.091288)
\psline[linewidth=\AxesLineWidth,linecolor=GridColor](15.000000,0.000000)(15.000000,0.091288)
\psline[linewidth=\AxesLineWidth,linecolor=GridColor](20.000000,0.000000)(20.000000,0.091288)
\psline[linewidth=\AxesLineWidth,linecolor=GridColor](25.000000,0.000000)(25.000000,0.091288)
\psline[linewidth=\AxesLineWidth,linecolor=GridColor](30.000000,0.000000)(30.000000,0.091288)
\psline[linewidth=\AxesLineWidth,linecolor=GridColor](35.000000,0.000000)(35.000000,0.091288)
\psline[linewidth=\AxesLineWidth,linecolor=GridColor](40.000000,0.000000)(40.000000,0.091288)
\psline[linewidth=\AxesLineWidth,linecolor=GridColor](45.000000,0.000000)(45.000000,0.091288)
\psline[linewidth=\AxesLineWidth,linecolor=GridColor](50.000000,0.000000)(50.000000,0.091288)
\psline[linewidth=\AxesLineWidth,linecolor=GridColor](0.000000,0.000000)(0.600000,0.000000)
\psline[linewidth=\AxesLineWidth,linecolor=GridColor](0.000000,1.000000)(0.600000,1.000000)
\psline[linewidth=\AxesLineWidth,linecolor=GridColor](0.000000,2.000000)(0.600000,2.000000)
\psline[linewidth=\AxesLineWidth,linecolor=GridColor](0.000000,3.000000)(0.600000,3.000000)
\psline[linewidth=\AxesLineWidth,linecolor=GridColor](0.000000,4.000000)(0.600000,4.000000)
\psline[linewidth=\AxesLineWidth,linecolor=GridColor](0.000000,5.000000)(0.600000,5.000000)
\psline[linewidth=\AxesLineWidth,linecolor=GridColor](0.000000,6.000000)(0.600000,6.000000)

{ \footnotesize 
\rput[t](0.000000,-0.091288){$0$}
\rput[t](5.000000,-0.091288){$5$}
\rput[t](10.000000,-0.091288){$10$}
\rput[t](15.000000,-0.091288){$15$}
\rput[t](20.000000,-0.091288){$20$}
\rput[t](25.000000,-0.091288){$25$}
\rput[t](30.000000,-0.091288){$30$}
\rput[t](35.000000,-0.091288){$35$}
\rput[t](40.000000,-0.091288){$40$}
\rput[t](45.000000,-0.091288){$45$}
\rput[t](50.000000,-0.091288){$50$}
\rput[r](-0.600000,0.000000){$0$}
\rput[r](-0.600000,1.000000){$1$}
\rput[r](-0.600000,2.000000){$2$}
\rput[r](-0.600000,3.000000){$3$}
\rput[r](-0.600000,4.000000){$4$}
\rput[r](-0.600000,5.000000){$5$}
\rput[r](-0.600000,6.000000){$6$}
} 

\psframe[linewidth=\AxesLineWidth,dimen=middle](0.000000,0.000000)(50.000000,6.000000)

{ \small 
\rput[b](25.000000,-0.666667){
\begin{tabular}{c}
$P_i$\\
\end{tabular}
}

\rput[t]{90}(-4.262673,3.000000){
\begin{tabular}{c}
$\frac{d \log U_i}{d P_i}$\\
\end{tabular}
}
} 

\newrgbcolor{color787.0049}{0  0  1}
\psline[plotstyle=line,linejoin=1,showpoints=false,dotstyle=*,dotsize=\MarkerSize,linestyle=solid,linewidth=\LineWidth,linecolor=color787.0049]
(0.858726,6.000000)(1.100000,4.049719)
\psline[plotstyle=line,linejoin=1,showpoints=false,dotstyle=*,dotsize=\MarkerSize,linestyle=solid,linewidth=\LineWidth,linecolor=color787.0049]
(49.100000,0.000000)(50.000000,0.000000)
\psline[plotstyle=line,linejoin=1,showpoints=true,dotstyle=*,dotsize=\MarkerSize,linestyle=solid,linewidth=\LineWidth,linecolor=color787.0049]
(1.100000,4.049719)(2.100000,4.000863)(3.100000,3.998016)(4.100000,3.893612)(5.100000,1.605249)
(6.100000,0.048514)(7.100000,0.000899)(8.100000,0.000016)(9.100000,0.000000)(10.100000,0.000000)
(11.100000,0.000000)(12.100000,0.000000)(13.100000,0.000000)(14.100000,0.000000)(15.100000,0.000000)
(16.100000,0.000000)(17.100000,0.000000)(18.100000,0.000000)(19.100000,0.000000)(20.100000,0.000000)
(21.100000,0.000000)(22.100000,0.000000)(23.100000,0.000000)(24.100000,0.000000)(25.100000,0.000000)
(26.100000,0.000000)(27.100000,0.000000)(28.100000,0.000000)(29.100000,0.000000)(30.100000,0.000000)
(31.100000,0.000000)(32.100000,0.000000)(33.100000,0.000000)(34.100000,0.000000)(35.100000,0.000000)
(36.100000,0.000000)(37.100000,0.000000)(38.100000,0.000000)(39.100000,0.000000)(40.100000,0.000000)
(41.100000,0.000000)(42.100000,0.000000)(43.100000,0.000000)(44.100000,0.000000)(45.100000,0.000000)
(46.100000,0.000000)(47.100000,0.000000)(48.100000,0.000000)(49.100000,0.000000)

\newrgbcolor{color788.0044}{0         0.5           0}
\psline[plotstyle=line,linejoin=1,showpoints=false,dotstyle=Bsquare,dotsize=\MarkerSize,linestyle=solid,linewidth=\LineWidth,linecolor=color788.0044]
(0.807109,6.000000)(1.100000,3.576098)
\psline[plotstyle=line,linejoin=1,showpoints=false,dotstyle=Bsquare,dotsize=\MarkerSize,linestyle=solid,linewidth=\LineWidth,linecolor=color788.0044]
(49.100000,0.000000)(50.000000,0.000000)
\psline[plotstyle=line,linejoin=1,showpoints=true,dotstyle=Bsquare,dotsize=\MarkerSize,linestyle=solid,linewidth=\LineWidth,linecolor=color788.0044]
(1.100000,3.576098)(2.100000,3.502251)(3.100000,3.500068)(4.100000,3.500002)(5.100000,3.500000)
(6.100000,3.499996)(7.100000,3.499863)(8.100000,3.495477)(9.100000,3.356181)(10.100000,1.446838)
(11.100000,0.072927)(12.100000,0.002248)(13.100000,0.000068)(14.100000,0.000002)(15.100000,0.000000)
(16.100000,0.000000)(17.100000,0.000000)(18.100000,0.000000)(19.100000,0.000000)(20.100000,0.000000)
(21.100000,0.000000)(22.100000,0.000000)(23.100000,0.000000)(24.100000,0.000000)(25.100000,0.000000)
(26.100000,0.000000)(27.100000,0.000000)(28.100000,0.000000)(29.100000,0.000000)(30.100000,0.000000)
(31.100000,0.000000)(32.100000,0.000000)(33.100000,0.000000)(34.100000,0.000000)(35.100000,0.000000)
(36.100000,0.000000)(37.100000,0.000000)(38.100000,0.000000)(39.100000,0.000000)(40.100000,0.000000)
(41.100000,0.000000)(42.100000,0.000000)(43.100000,0.000000)(44.100000,0.000000)(45.100000,0.000000)
(46.100000,0.000000)(47.100000,0.000000)(48.100000,0.000000)(49.100000,0.000000)

\newrgbcolor{color789.0044}{1  0  0}
\psline[plotstyle=line,linejoin=1,showpoints=false,dotstyle=Bo,dotsize=\MarkerSize,linestyle=solid,linewidth=\LineWidth,linecolor=color789.0044]
(0.758970,6.000000)(1.100000,3.114887)
\psline[plotstyle=line,linejoin=1,showpoints=false,dotstyle=Bo,dotsize=\MarkerSize,linestyle=solid,linewidth=\LineWidth,linecolor=color789.0044]
(49.100000,0.000000)(50.000000,0.000000)
\psline[plotstyle=line,linejoin=1,showpoints=true,dotstyle=Bo,dotsize=\MarkerSize,linestyle=solid,linewidth=\LineWidth,linecolor=color789.0044]
(1.100000,3.114887)(2.100000,3.005519)(3.100000,3.000274)(4.100000,3.000014)(5.100000,3.000001)
(6.100000,3.000000)(7.100000,3.000000)(8.100000,3.000000)(9.100000,3.000000)(10.100000,2.999999)
(11.100000,2.999975)(12.100000,2.999500)(13.100000,2.989996)(14.100000,2.811080)(15.100000,1.276672)
(16.100000,0.106714)(17.100000,0.005499)(18.100000,0.000274)(19.100000,0.000014)(20.100000,0.000001)
(21.100000,0.000000)(22.100000,0.000000)(23.100000,0.000000)(24.100000,0.000000)(25.100000,0.000000)
(26.100000,0.000000)(27.100000,0.000000)(28.100000,0.000000)(29.100000,0.000000)(30.100000,0.000000)
(31.100000,0.000000)(32.100000,0.000000)(33.100000,0.000000)(34.100000,0.000000)(35.100000,0.000000)
(36.100000,0.000000)(37.100000,0.000000)(38.100000,0.000000)(39.100000,0.000000)(40.100000,0.000000)
(41.100000,0.000000)(42.100000,0.000000)(43.100000,0.000000)(44.100000,0.000000)(45.100000,0.000000)
(46.100000,0.000000)(47.100000,0.000000)(48.100000,0.000000)(49.100000,0.000000)

\newrgbcolor{color790.0044}{0        0.75        0.75}
\psline[plotstyle=line,linejoin=1,showpoints=false,dotstyle=Bpentagon,dotsize=\MarkerSize,linestyle=solid,linewidth=\LineWidth,linecolor=color790.0044]
(0.714280,6.000000)(1.100000,2.670734)
\psline[plotstyle=line,linejoin=1,showpoints=false,dotstyle=Bpentagon,dotsize=\MarkerSize,linestyle=solid,linewidth=\LineWidth,linecolor=color790.0044]
(49.100000,0.000000)(50.000000,0.000000)
\psline[plotstyle=line,linejoin=1,showpoints=true,dotstyle=Bpentagon,dotsize=\MarkerSize,linestyle=solid,linewidth=\LineWidth,linecolor=color790.0044]
(1.100000,2.670734)(2.100000,2.513188)(3.100000,2.501077)(4.100000,2.500088)(5.100000,2.500007)
(6.100000,2.500001)(7.100000,2.500000)(8.100000,2.500000)(9.100000,2.500000)(10.100000,2.500000)
(11.100000,2.500000)(12.100000,2.500000)(13.100000,2.500000)(14.100000,2.499999)(15.100000,2.499988)
(16.100000,2.499854)(17.100000,2.498226)(18.100000,2.478556)(19.100000,2.261626)(20.100000,1.094559)
(21.100000,0.150217)(22.100000,0.013050)(23.100000,0.001076)(24.100000,0.000088)(25.100000,0.000007)
(26.100000,0.000001)(27.100000,0.000000)(28.100000,0.000000)(29.100000,0.000000)(30.100000,0.000000)
(31.100000,0.000000)(32.100000,0.000000)(33.100000,0.000000)(34.100000,0.000000)(35.100000,0.000000)
(36.100000,0.000000)(37.100000,0.000000)(38.100000,0.000000)(39.100000,0.000000)(40.100000,0.000000)
(41.100000,0.000000)(42.100000,0.000000)(43.100000,0.000000)(44.100000,0.000000)(45.100000,0.000000)
(46.100000,0.000000)(47.100000,0.000000)(48.100000,0.000000)(49.100000,0.000000)

\newrgbcolor{color791.0044}{0.75           0        0.75}
\psline[plotstyle=line,linejoin=1,showpoints=false,dotstyle=Basterisk,dotsize=\MarkerSize,linestyle=solid,linewidth=\LineWidth,linecolor=color791.0044]
(0.635081,6.000000)(1.100000,1.856550)
\psline[plotstyle=line,linejoin=1,showpoints=false,dotstyle=Basterisk,dotsize=\MarkerSize,linestyle=solid,linewidth=\LineWidth,linecolor=color791.0044]
(49.100000,0.000000)(50.000000,0.000000)
\psline[plotstyle=line,linejoin=1,showpoints=true,dotstyle=Basterisk,dotsize=\MarkerSize,linestyle=solid,linewidth=\LineWidth,linecolor=color791.0044]
(1.100000,1.856550)(2.100000,1.567156)(3.100000,1.514481)(4.100000,1.503207)(5.100000,1.500714)
(6.100000,1.500159)(7.100000,1.500036)(8.100000,1.500008)(9.100000,1.500002)(10.100000,1.500000)
(11.100000,1.500000)(12.100000,1.500000)(13.100000,1.500000)(14.100000,1.500000)(15.100000,1.499999)
(16.100000,1.499998)(17.100000,1.499989)(18.100000,1.499952)(19.100000,1.499785)(20.100000,1.499037)
(21.100000,1.495693)(22.100000,1.480886)(23.100000,1.417978)(24.100000,1.191194)(25.100000,0.693855)
(26.100000,0.241663)(27.100000,0.061637)(28.100000,0.014207)(29.100000,0.003193)(30.100000,0.000714)
(31.100000,0.000159)(32.100000,0.000036)(33.100000,0.000008)(34.100000,0.000002)(35.100000,0.000000)
(36.100000,0.000000)(37.100000,0.000000)(38.100000,0.000000)(39.100000,0.000000)(40.100000,0.000000)
(41.100000,0.000000)(42.100000,0.000000)

\newrgbcolor{color793.0039}{0.75        0.75           0}
\psline[plotstyle=line,linejoin=1,showpoints=false,dotstyle=B+,dotsize=\MarkerSize,linestyle=solid,linewidth=\LineWidth,linecolor=color793.0039]
(0.600405,6.000000)(1.100000,1.498961)
\psline[plotstyle=line,linejoin=1,showpoints=false,dotstyle=B+,dotsize=\MarkerSize,linestyle=solid,linewidth=\LineWidth,linecolor=color793.0039]
(49.100000,0.000000)(50.000000,0.000000)
\psline[plotstyle=line,linejoin=1,showpoints=true,dotstyle=B+,dotsize=\MarkerSize,linestyle=solid,linewidth=\LineWidth,linecolor=color793.0039]
(1.100000,1.498961)(2.100000,1.139545)(3.100000,1.047174)(4.100000,1.016852)(5.100000,1.006134)
(6.100000,1.002248)(7.100000,1.000826)(8.100000,1.000304)(9.100000,1.000112)(10.100000,1.000041)
(11.100000,1.000015)(12.100000,1.000006)(13.100000,1.000002)(14.100000,1.000001)(15.100000,1.000000)
(16.100000,0.999999)(17.100000,0.999998)(18.100000,0.999993)(19.100000,0.999982)(20.100000,0.999950)
(21.100000,0.999864)(22.100000,0.999629)(23.100000,0.998993)(24.100000,0.997268)(25.100000,0.992608)
(26.100000,0.980160)(27.100000,0.947846)(28.100000,0.869892)(29.100000,0.710950)(30.100000,0.475021)
(31.100000,0.249740)(32.100000,0.109097)(33.100000,0.043107)(34.100000,0.016302)(35.100000,0.006060)
(36.100000,0.002238)(37.100000,0.000824)(38.100000,0.000303)(39.100000,0.000112)(40.100000,0.000041)
(41.100000,0.000015)(42.100000,0.000006)(43.100000,0.000002)(44.100000,0.000001)(45.100000,0.000000)

{ \small 
\rput[tr](48.800000,5.817423){%
\psshadowbox[framesep=0pt,linewidth=\AxesLineWidth]{\psframebox*{\begin{tabular}{l}
\Rnode{a1}{\hspace*{0.0ex}} \hspace*{0.4cm} \Rnode{a2}{~~user 1} \\
\Rnode{a3}{\hspace*{0.0ex}} \hspace*{0.4cm} \Rnode{a4}{~~user 2} \\
\Rnode{a5}{\hspace*{0.0ex}} \hspace*{0.4cm} \Rnode{a6}{~~user 3} \\
\Rnode{a7}{\hspace*{0.0ex}} \hspace*{0.4cm} \Rnode{a8}{~~user 4} \\
\Rnode{a9}{\hspace*{0.0ex}} \hspace*{0.4cm} \Rnode{a10}{~~user 5} \\
\Rnode{a11}{\hspace*{0.0ex}} \hspace*{0.4cm} \Rnode{a12}{~~user 6} \\
\end{tabular}}
\ncline[linestyle=solid,linewidth=\LineWidth,linecolor=color787.0049]{a1}{a2} \ncput{\psdot[dotstyle=*,dotsize=\MarkerSize,linecolor=color787.0049]}
\ncline[linestyle=solid,linewidth=\LineWidth,linecolor=color788.0044]{a3}{a4} \ncput{\psdot[dotstyle=Bsquare,dotsize=\MarkerSize,linecolor=color788.0044]}
\ncline[linestyle=solid,linewidth=\LineWidth,linecolor=color789.0044]{a5}{a6} \ncput{\psdot[dotstyle=Bo,dotsize=\MarkerSize,linecolor=color789.0044]}
\ncline[linestyle=solid,linewidth=\LineWidth,linecolor=color790.0044]{a7}{a8} \ncput{\psdot[dotstyle=Bpentagon,dotsize=\MarkerSize,linecolor=color790.0044]}
\ncline[linestyle=solid,linewidth=\LineWidth,linecolor=color791.0044]{a9}{a10} \ncput{\psdot[dotstyle=Basterisk,dotsize=\MarkerSize,linecolor=color791.0044]}
\ncline[linestyle=solid,linewidth=\LineWidth,linecolor=color793.0039]{a11}{a12} \ncput{\psdot[dotstyle=B+,dotsize=\MarkerSize,linecolor=color793.0039]}
}%
}%
} 

\end{pspicture}%
\caption{The first derivative of the natural logarithm of sigmoidal-like utility functions  $\frac{\partial \log U_i(\gamma_i(P_i))}{\partial P_i}$.}
\label{fig:diff_log_sigmoid}
\end{figure}

\section{The Dual Problem}\label{sec:Dual}

The key to a distributed and a decentralized optimal solution of the primal problem in (\ref{eqn:opt_prob_fairness_mod}) is to convert it to the dual problem, similar to \cite{kelly98powercontrol} and \cite{Low99optimizationflow}. The optimization problem (\ref{eqn:opt_prob_fairness_mod}) can be divided into two simpler problems by using the dual problem.  We define the Lagrangian
\begin{equation}\label{eqn:lagrangian}
\begin{aligned}
L(\textbf{P},p) = & \sum_{i=1}^{M}{\log(U_i(\gamma_i(P_i)))-p(\sum_{i=1}^{M}P_i + z_i - P_T)}\\
                = &  \sum_{i=1}^{M}\Big({\log(U_i(\gamma_i(P_i)))-pP_i\Big)} + p \sum_{i=1}^{M} (P_T-z_i)\\
                = &  \sum_{i=1}^{M}L_i(P_i,p) + p \sum_{i=1}^{M} (P_T-z_i)\\
\end{aligned}
\end{equation}
where $z_i \geq 0$ is the slack variable and $p$ is Lagrange multiplier or the shadow price (i.e. the total price per unit power for all the $M$ channels). Therefore, the $i^{\text{th}}$ UE bid for power can be given by $w_i = p P_i$ and we have $\sum_{i=1}^{M}w_i = p \sum_{i=1}^{M}P_i$. The first term in equation (\ref{eqn:lagrangian}) is separable in $P_i$. So we have $\underset{\textbf{P}}\max \sum_{i=1}^{M}({\log(U_i(\gamma_i(P_i)))-pP_i)} = \sum_{i=1}^{M}\underset{{P_i}}\max({\log(U_i(\gamma_i(P_i)))-pP_i)}$.  
The dual problem objective function can be written as
\begin{equation}\label{eqn:dual_obj_fn}
\begin{aligned}
D(p) = & \underset{{\textbf{P}}}\max \:L(\textbf{P},p) \\
= &\sum_{i=1}^{M}\underset{{P_i}}\max\Big({\log(U_i(\gamma_i(P_i)))-pP_i\Big)} + p \sum_{i=1}^{M} (P_T-z_i)\\
= &\sum_{i=1}^{M}\underset{{P_i}}\max (L_i(P_i,p)) + p \sum_{i=1}^{M} (P_T-z_i).
\end{aligned}
\end{equation}
The dual problem is given by
\begin{equation}\label{eqn:dual_problem}
\begin{aligned}
& \underset{{p}}{\text{min}}
& & D(p) \\
& \text{subject to}
& & p \geq 0.
\end{aligned}
\end{equation}
So we have 
\begin{equation}\label{eqn:dual_max}
\frac{\partial D(p)}{\partial p} =  P_T -\sum_{i=1}^{M} (P_i -z_i) = 0
\end{equation}
substituting by $\sum_{i=1}^{M}w_i = p \sum_{i=1}^{M}P_i$ we have 
\begin{equation}\label{eqn:dual_new_obj}
p = \frac{\sum_{i=1}^{M}w_i}{P_T-\sum_{i=1}^{M} z_i} \cdot
\end{equation}
Now, we divide the primal problem (\ref{eqn:opt_prob_fairness_mod}) into two simpler optimization problems in the UEs and the BS. The $i^{\text{th}}$ UE optimization problem is given by: 
\begin{equation}\label{eqn:opt_prob_fairness_UE}
\begin{aligned}
& \underset{{P_i}}{\text{max}}
& & \log U_i(\gamma_i(P_i)) - p P_i \\
& \text{subject to}
& & p \geq 0\\
& & &  P_i \geq 0, \;\;\;\;\; \text{for} \; \; i = 1,2, ...,M \; \; \text{and} \; \; P_T \geq 0.
\end{aligned}
\end{equation}
The BS optimization problem is given by: 
\begin{equation}\label{eqn:opt_prob_fairness_eNodeB}
\begin{aligned}
& \underset{p}{\text{min}}
& & D(p) \\
& \text{subject to}
& & p \geq 0.\\
\end{aligned}
\end{equation}
The minimization of shadow price $p$ is achieved by the minimization of the slack variable $z_i \geq 0$ from equation (\ref{eqn:dual_new_obj}). Therefore, the maximum utilization of the available BS power is achieved by setting the slack variable $z_i = 0$. In this case, we replace the inequality in primal problem (\ref{eqn:opt_prob_fairness_mod}) constraint by  equality constraint and so we have $\sum_{i=1}^{M}w_i = p P_T$. Therefore, we have $p = \frac{\sum_{i=1}^{M}w_i}{P}$ where $w_i = p P_i$ is transmitted by the $i^{\text{th}}$ UE to the BS. The utility proportional fairness in the objective function of the optimization problem (\ref{eqn:opt_prob_fairness}) is guaranteed in the solution of the optimization problems (\ref{eqn:opt_prob_fairness_UE}) and (\ref{eqn:opt_prob_fairness_eNodeB}).

\section{Distributed Algorithm}\label{sec:Algorithm}

We can directly construct a distributed power allocation algorithm form the dual problem. The distributed power allocation algorithm is an iterative solution for allocating the network resources with bandwidth proportional fairness. Our algorithm allocates powers with utility proportional fairness, which is the objective of our new problem formulation. The algorithm is divided into an UE algorithm shown in Algorithm (\ref{alg:UE_first}) and an BS algorithm shown in Algorithm (\ref{alg:eNodeB_first}). For the Algorithm in (\ref{alg:UE_first}) and (\ref{alg:eNodeB_first}), each UE starts with an initial bid $w_i(1)$ which is transmitted to the BS. The BS calculates the difference between the received bid $w_i(n)$ and the previously received bid $w_i(n-1)$ and exits if it is less than a pre-specified threshold $\delta$. We set $w_i(0) = 0$. If the value is greater 
than the 
threshold $\delta$, BS calculates the shadow price $p(n) = \frac{\sum_{i=1}^{M}w_i(n)}{P}$ and sends that value to all the UEs. Each UE receives the shadow price to solve for the power $P_i$ that maximizes $\log U_i(\gamma_i(P_i)) - p(n)P_i$. That power is used to calculate the new bid $w_i(n)=p(n) P_{i}(n)$. Each UE sends the value of its new bid $w_i(n)$ to the BS. This process is repeated until $|w_i(n) -w_i(n-1)|$ is less than the pre-specified threshold $\delta$.


\begin{algorithm}
\caption{UE Algorithm}\label{alg:UE_first}
\begin{algorithmic}
\STATE {Send initial bid $w_i(1)$ to BS}
\LOOP
	\STATE {Receive shadow price $p(n)$ from BS}
	\IF {STOP from BS} %
		\STATE {Calculate allocated power $P_i ^{\text{opt}}=\frac{w_i(n)}{p(n)}$}
		\STATE {STOP}
	\ELSE
		\STATE {Solve $P_{i}(n) = \arg \underset{P_i}\max \Big(\log U_i(\gamma_i(P_i)) - p(n)P_i\Big)$}
		\STATE {Send new bid $w_i (n)= p(n) P_{i}(n)$ to BS}
	\ENDIF 
\ENDLOOP
\end{algorithmic}
\end{algorithm}


\begin{algorithm}
\caption{BS Algorithm}\label{alg:eNodeB_first}
\begin{algorithmic}
\LOOP
	\STATE {Receive bids $w_i(n)$ from UEs}
	\COMMENT{Let $w_i(0) = 0\:\:\forall i$}
			\IF {$|w_i(n) -w_i(n-1)|<\delta  \:\:\forall i$} %
	   		\STATE {Allocate powers, $P_{i}^{\text{opt}}=\frac{w_i(n)}{p(n)}$ to user $i$}  
	   		\STATE {STOP} 
		\ELSE
	\STATE {Calculate $p(n) = \frac{\sum_{i=1}^{M}w_i(n)}{P}$}
	\STATE {Send new shadow price $p(n)$ to all UEs}
	\ENDIF 
\ENDLOOP
\end{algorithmic}
\end{algorithm}

The solution $P_i$ of the optimization problem $P_{i}(n) = \arg \underset{P_i}\max \Big(\log U_i(\gamma_i(P_i)) - p(n)P_i\Big)$ in Algorithm (\ref{alg:UE_first}),  is the value of $P_i$ that solves equation $\frac{\partial \log U_i(\gamma_i(P_i))}{\partial P_i} = p(n)$. 


 \section{Convergence}\label{sec:conv_analy}
In this section, we present the convergence analysis of Algorithm (\ref{alg:UE_first}) and (\ref{alg:eNodeB_first}) for different values of $P_T$.
\begin{lem}\label{lem:slope_curve}
For sigmoidal-like utility function $U_i(\gamma_i(P_i))$,  the slope curvature function $\frac{\partial \log U_i(\gamma_i(P_i))}{\partial P_i}$ has an inflection point at $P_i = P_i^{s} \approx b_i$ and is convex for $P_i > P_i^{s}$.
\end{lem}
\begin{proof}
For the sigmoidal-like function $U_i(\gamma_i(P_i)) = c_i\Big(\frac{1}{1+e^{-a_i(P_i-b_i)}}-d_i\Big)$, let $S_i(P_i) = \frac{\partial \log U_i(\gamma_i(P_i))}{\partial P_i}$ be the slope curvature function. Then, we have that 
\begin{equation}
\begin{aligned}\label{eqn:diff_slope}
\frac{\partial S_i}{\partial P_i} &= \frac{-a_i^2d_ie^{-a_i(P_i-b_i)}}{c_i\Big(1-d_i(1+e^{-a_i(P_i-b_i)})\Big)^2} - \frac{a_i^2e^{-a_i(P_i-b_i)}}{\Big(1+e^{-a_i(P_i-b_i)}\Big)^2}\\
\text{and}\\
\frac{\partial^2 S_i}{\partial P_i^2}& = \frac{a_i^3d_ie^{-a_i(P_i-b_i)}(1-d_i(1-e^{-a_i(P_i-b_i)}))}{c_i\Big(1-d_i(1+e^{-a_i(P_i-b_i)})\Big)^3} + \frac{a_i^3e^{-a_i(P_i-b_i)}(1-e^{-a_i(P_i-b_i)})}{\Big(1+e^{-a_i(P_i-b_i)}\Big)^3}.\\
\end{aligned}
\end{equation}
We analyze the curvature of the slope of the natural logarithm of sigmoidal-like utility function. For the first derivative, we have $\frac{\partial S_i}{\partial P_i}<0 \:\:\:\forall\: P_i$. The first term $S^1_i$ of $\frac{\partial^2 S_i}{\partial P_i^2}$ in equation (\ref{eqn:diff_slope}) can be written as
\begin{equation}\label{eqn:slope_fn}
S^1_i = \frac{a_i^3e^{a_ib_i}(e^{a_ib_i}+e^{-a_i(P_i-b_i)})}{(e^{a_ib_i}-e^{-a_i(P_i-b_i)})^3}
\end{equation}
and we have
\begin{equation}\label{eqn:slope_fn_term1}
\lim_{P_i \rightarrow 0} S^1_i = \infty, \:\: \text{and} \:\:\lim_{P_i \rightarrow b_i} S^1_i = 0 \:\:\text{for} \:\:b_i 	\gg \frac{1}{a_i}.
\end{equation}
For second term $S^2_i$ of $\frac{\partial^2 S_i}{\partial P_i^2}$ in equation (\ref{eqn:diff_slope}), we have the following properties
\begin{equation}\label{eqn:slope_fn_term2}
S^2_i (b_i) = 0, \:\:S^2_i (P_i>b_i) > 0, \:\: \text{and} \:\:S^2_i (P_i<b_i) < 0.
\end{equation}
From equation (\ref{eqn:slope_fn_term1}) and (\ref{eqn:slope_fn_term2}), $S_i$ has an inflection point at $P_i = P_i^{s}  	\approx b_i$. In addition, we have the curvature of $S_i$ changes from a convex function close to origin to a concave function before the inflection point $P_i = P_i^{s}$ then to a convex function after the inflection point. 
\end{proof}
\begin{cor}\label{cor:sig_convergence}
If $\sum_{i=1}^{M}P_i^{\text{inf}} \ll P_T$ then Algorithm in (\ref{alg:UE_first}) and  (\ref{alg:eNodeB_first}) converges to the global optimal powers which correspond to the steady state shadow price $p_{ss}< \frac{a_{i_{\max}} d_{i_{\max}} }{1-d_{i_{\max}} }+\frac{a_{i_{\max} }}{2}$ where $i_{\max} = \arg \max_i b_i$.
\end{cor}
\begin{proof}
For the sigmoidal-like function $U_i(\gamma_i(P_i)) = c_i\Big(\frac{1}{1+e^{-a_i(P_i-b_i)}}-d_i\Big)$, the optimal solution is achieved by solving the optimization problem (\ref{eqn:opt_prob_fairness_mod}). In Algorithm (\ref{alg:UE_first}), an important step to reach to the optimal solution is to solve the optimization problem $P_{i}(n) = \arg \underset{P_i}\max \Big(\log U_i(\gamma_i(P_i)) - p(n)P_i\Big)$ for every UE. The solution of this problem can be written, using Lagrange multipliers method, in the form 
\begin{equation}\label{eqn:slope_equation}
\frac{\partial \log U_i(\gamma_i(P_i))}{\partial P_i}-p =  S_i(P_i) - p = 0.
\end{equation}
From equation (\ref{eqn:slope_fn_term1}) and (\ref{eqn:slope_fn_term2}) in Lemma \ref{lem:slope_curve}, we have the curvature of $S_i(P_i)$ is convex for $P_i  > P_i^s \approx b_i$. The Algorithm in (\ref{alg:UE_first}) and  (\ref{alg:eNodeB_first}) is guaranteed to converges to the global optimal solution when the slope $S_i(P_i)$ of all the utility functions' natural logarithm $\log U_i(\gamma_i(P_i))$ are in the convex region of the functions, similar to the analysis of logarithmic functions in \cite{kelly98powercontrol} and \cite{Low99optimizationflow}. Therefore, the natural logarithm of sigmoidal-like functions $\log U_i(\gamma_i(P_i))$ converge to the global optimal solution for $P_i > P_i^ s \approx b_i$. The inflection point of sigmoidal-like function $U_i(\gamma_i(P_i))$ is at $P_i^{\text{inf}} = b_i$. For $\sum_{i=1}^{M}P_i^{\text{inf}} \ll P_T$, Algorithm in (\ref{alg:UE_first}) and  (\ref{alg:eNodeB_first}) allocates powers $P_i>b_i$ for all users. Since $S_i(P_i)$ is convex for $P_i>P_i^s \
approx b_
i$ then the optimal solution 
can 
be achieved by  
Algorithm (\ref{alg:UE_first}) and  (\ref{alg:eNodeB_first}). We have from equation (\ref{eqn:slope_equation}) and as $S_i(P_i)$ is convex for $P_i  > P_i^s \approx b_i$, that $p_{ss}< S_i(P_i =\max b_i)$ where $S_i(P_i =\max b_i) = \frac{a_{i_{\max}} d_{i_{\max}} }{1-d_{i_{\max}} }+\frac{a_{i_{\max} }}{2}$ and $i_{\max} = \arg \max_i b_i$.
\end{proof}
\begin{cor}\label{cor:sig_fluctuate}
For $\sum_{i=1}^{M}P_i^{\text{inf}}>P_T$ and the global optimal shadow price $p_{ss} \approx \frac{a_id_i e^{\frac{a_ib_i}{2}}}{1-d_i(1+e^{\frac{a_ib_i}{2}})} + \frac{a_ie^{\frac{a_ib_i}{2}}}{(1+e^{\frac{a_ib_i}{2}})}$, then the solution by Algorithm in (\ref{alg:UE_first}) and  (\ref{alg:eNodeB_first}) fluctuates about the global optimal solution.
\end{cor}
\begin{proof}
It follows from lemma \ref{lem:slope_curve} that for  $\sum_{i=1}^{M}P_i^{\text{inf}}>P_T$  $\exists \:\: i$ such that the optimal powers $P_i^{\text{opt}} < b_i$. Therefore, if $p_{ss} \approx \frac{a_id_i e^{\frac{a_ib_i}{2}}}{1-d_i(1+e^{\frac{a_ib_i}{2}})} + \frac{a_ie^{\frac{a_ib_i}{2}}}{(1+e^{\frac{a_ib_i}{2}})}$ is the optimal shadow price for optimization problem (\ref{eqn:opt_prob_fairness_eNodeB}). Then, a small change in the shadow price $p(n)$ in the $n^{\text{th}}$ iteration can lead the power $P_i(n)$ (root of $S_i(P_i) - p(n) =0$) to fluctuate between the concave and convex curvature of the slope curve $S_i(P_i)$ for the $i^{\text{th}}$ user. Therefore, it causes fluctuation in the bid $w_i(n)$ sent to the BS and fluctuation in the shadow price $p(n)$ set by BS. Therefore, the iterative solution of Algorithm in (\ref{alg:UE_first}) and (\ref{alg:eNodeB_first}) fluctuates about the global optimal powers $P_i^{\text{opt}}$.
\end{proof}
\begin{thm}\label{thm:sig_not_conv}
Algorithm in (\ref{alg:UE_first}) and  (\ref{alg:eNodeB_first}) does not converge to the global optimal solution for all values of $P_T$.
\end{thm}
\begin{proof}
It follows from Corollary \ref{cor:sig_convergence} and \ref{cor:sig_fluctuate} that Algorithm in (\ref{alg:UE_first}) and  (\ref{alg:eNodeB_first}) does not converge to the global optimal solution for all values of $P_T$.
\end{proof}

\subsection{Fluctuation Example}\label{sec:Fluctuation_example}

We consider an example of six users using sigmoidal-like utility functions. The sigmoidal-like utility functions' parameters are $a=\{4, 3.5, 3, 2.5, 1.5, 1\}$ and  $b=\{5, 10, 15, 20, 25, 30\}$, respectively. We assume that the BS's maximum power is $P_T=100$, therefore,  $\sum_{i=1}^{6} P_i^{\text{inf}} = 105 > P_T = 100$. Hence, we can't guarantee convergence with Algorithm in (\ref{alg:UE_first}) and (\ref{alg:eNodeB_first}), as stated by Corollary \ref{cor:sig_fluctuate}. In Figure \ref{fig:osc_p_damped}, we show that the shadow price $p(n)$ fluctuates between a concave and convex curvature of the $\frac{\partial \log U_i(\gamma_i(P_i))}{\partial P_i}$ curve. The fluctuation in the shadow price $p(n)$ causes fluctuation in the allocated powers and hinders the convergence to the optimal powers. Therefore, the optimal power allocation is not 
achievable by 
Algorithm 
in (\ref{alg:UE_first}) and (\ref{alg:eNodeB_first}).
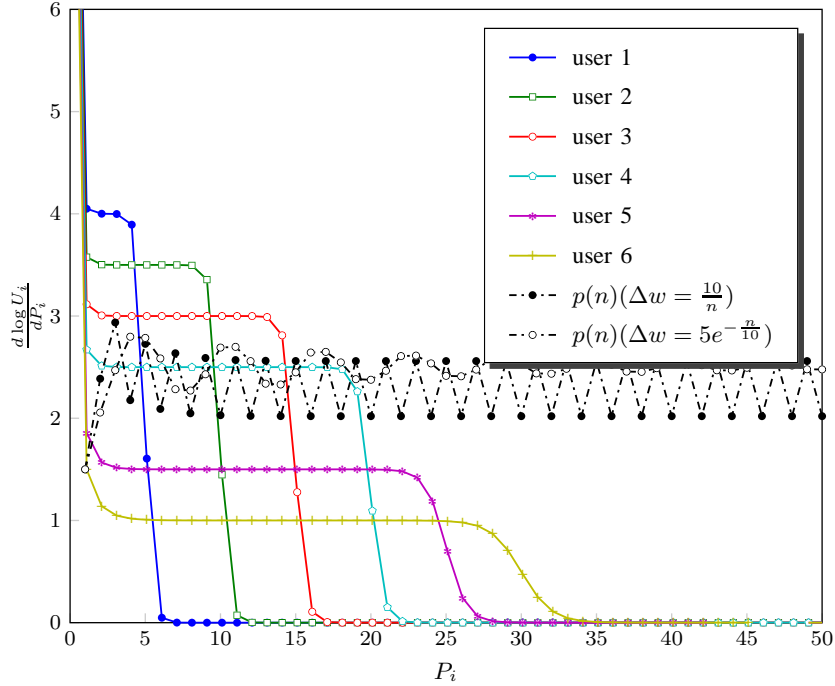
\begin{figure}[!t]
\centering
%
\psset{xunit=0.0250000\plotwidth,yunit=0.17\plotwidth}%
\begin{pspicture}(-4.262673,-0.666667)(50.921659,6.140351)%


\psline[linewidth=\AxesLineWidth,linecolor=GridColor](0.000000,0.000000)(0.000000,0.091288)
\psline[linewidth=\AxesLineWidth,linecolor=GridColor](5.000000,0.000000)(5.000000,0.091288)
\psline[linewidth=\AxesLineWidth,linecolor=GridColor](10.000000,0.000000)(10.000000,0.091288)
\psline[linewidth=\AxesLineWidth,linecolor=GridColor](15.000000,0.000000)(15.000000,0.091288)
\psline[linewidth=\AxesLineWidth,linecolor=GridColor](20.000000,0.000000)(20.000000,0.091288)
\psline[linewidth=\AxesLineWidth,linecolor=GridColor](25.000000,0.000000)(25.000000,0.091288)
\psline[linewidth=\AxesLineWidth,linecolor=GridColor](30.000000,0.000000)(30.000000,0.091288)
\psline[linewidth=\AxesLineWidth,linecolor=GridColor](35.000000,0.000000)(35.000000,0.091288)
\psline[linewidth=\AxesLineWidth,linecolor=GridColor](40.000000,0.000000)(40.000000,0.091288)
\psline[linewidth=\AxesLineWidth,linecolor=GridColor](45.000000,0.000000)(45.000000,0.091288)
\psline[linewidth=\AxesLineWidth,linecolor=GridColor](50.000000,0.000000)(50.000000,0.091288)
\psline[linewidth=\AxesLineWidth,linecolor=GridColor](0.000000,0.000000)(0.600000,0.000000)
\psline[linewidth=\AxesLineWidth,linecolor=GridColor](0.000000,1.000000)(0.600000,1.000000)
\psline[linewidth=\AxesLineWidth,linecolor=GridColor](0.000000,2.000000)(0.600000,2.000000)
\psline[linewidth=\AxesLineWidth,linecolor=GridColor](0.000000,3.000000)(0.600000,3.000000)
\psline[linewidth=\AxesLineWidth,linecolor=GridColor](0.000000,4.000000)(0.600000,4.000000)
\psline[linewidth=\AxesLineWidth,linecolor=GridColor](0.000000,5.000000)(0.600000,5.000000)
\psline[linewidth=\AxesLineWidth,linecolor=GridColor](0.000000,6.000000)(0.600000,6.000000)

{ \footnotesize 
\rput[t](0.000000,-0.091288){$0$}
\rput[t](5.000000,-0.091288){$5$}
\rput[t](10.000000,-0.091288){$10$}
\rput[t](15.000000,-0.091288){$15$}
\rput[t](20.000000,-0.091288){$20$}
\rput[t](25.000000,-0.091288){$25$}
\rput[t](30.000000,-0.091288){$30$}
\rput[t](35.000000,-0.091288){$35$}
\rput[t](40.000000,-0.091288){$40$}
\rput[t](45.000000,-0.091288){$45$}
\rput[t](50.000000,-0.091288){$50$}
\rput[r](-0.600000,0.000000){$0$}
\rput[r](-0.600000,1.000000){$1$}
\rput[r](-0.600000,2.000000){$2$}
\rput[r](-0.600000,3.000000){$3$}
\rput[r](-0.600000,4.000000){$4$}
\rput[r](-0.600000,5.000000){$5$}
\rput[r](-0.600000,6.000000){$6$}
} 

\psframe[linewidth=\AxesLineWidth,dimen=middle](0.000000,0.000000)(50.000000,6.000000)

{ \small 
\rput[b](25.000000,-0.666667){
\begin{tabular}{c}
$P_i$\\
\end{tabular}
}

\rput[t]{90}(-4.262673,3.000000){
\begin{tabular}{c}
$\frac{d \log U_i}{d P_i}$\\
\end{tabular}
}
} 

\newrgbcolor{color787.0049}{0  0  1}
\psline[plotstyle=line,linejoin=1,showpoints=false,dotstyle=*,dotsize=\MarkerSize,linestyle=solid,linewidth=\LineWidth,linecolor=color787.0049]
(0.858726,6.000000)(1.100000,4.049719)
\psline[plotstyle=line,linejoin=1,showpoints=false,dotstyle=*,dotsize=\MarkerSize,linestyle=solid,linewidth=\LineWidth,linecolor=color787.0049]
(49.100000,0.000000)(50.000000,0.000000)
\psline[plotstyle=line,linejoin=1,showpoints=true,dotstyle=*,dotsize=\MarkerSize,linestyle=solid,linewidth=\LineWidth,linecolor=color787.0049]
(1.100000,4.049719)(2.100000,4.000863)(3.100000,3.998016)(4.100000,3.893612)(5.100000,1.605249)
(6.100000,0.048514)(7.100000,0.000899)(8.100000,0.000016)(9.100000,0.000000)(10.100000,0.000000)
(11.100000,0.000000)(12.100000,0.000000)(13.100000,0.000000)(14.100000,0.000000)(15.100000,0.000000)
(16.100000,0.000000)(17.100000,0.000000)(18.100000,0.000000)(19.100000,0.000000)(20.100000,0.000000)
(21.100000,0.000000)(22.100000,0.000000)(23.100000,0.000000)(24.100000,0.000000)(25.100000,0.000000)
(26.100000,0.000000)(27.100000,0.000000)(28.100000,0.000000)(29.100000,0.000000)(30.100000,0.000000)
(31.100000,0.000000)(32.100000,0.000000)(33.100000,0.000000)(34.100000,0.000000)(35.100000,0.000000)
(36.100000,0.000000)(37.100000,0.000000)(38.100000,0.000000)(39.100000,0.000000)(40.100000,0.000000)
(41.100000,0.000000)(42.100000,0.000000)(43.100000,0.000000)(44.100000,0.000000)(45.100000,0.000000)
(46.100000,0.000000)(47.100000,0.000000)(48.100000,0.000000)(49.100000,0.000000)

\newrgbcolor{color788.0044}{0         0.5           0}
\psline[plotstyle=line,linejoin=1,showpoints=false,dotstyle=Bsquare,dotsize=\MarkerSize,linestyle=solid,linewidth=\LineWidth,linecolor=color788.0044]
(0.807109,6.000000)(1.100000,3.576098)
\psline[plotstyle=line,linejoin=1,showpoints=false,dotstyle=Bsquare,dotsize=\MarkerSize,linestyle=solid,linewidth=\LineWidth,linecolor=color788.0044]
(49.100000,0.000000)(50.000000,0.000000)
\psline[plotstyle=line,linejoin=1,showpoints=true,dotstyle=Bsquare,dotsize=\MarkerSize,linestyle=solid,linewidth=\LineWidth,linecolor=color788.0044]
(1.100000,3.576098)(2.100000,3.502251)(3.100000,3.500068)(4.100000,3.500002)(5.100000,3.500000)
(6.100000,3.499996)(7.100000,3.499863)(8.100000,3.495477)(9.100000,3.356181)(10.100000,1.446838)
(11.100000,0.072927)(12.100000,0.002248)(13.100000,0.000068)(14.100000,0.000002)(15.100000,0.000000)
(16.100000,0.000000)(17.100000,0.000000)(18.100000,0.000000)(19.100000,0.000000)(20.100000,0.000000)
(21.100000,0.000000)(22.100000,0.000000)(23.100000,0.000000)(24.100000,0.000000)(25.100000,0.000000)
(26.100000,0.000000)(27.100000,0.000000)(28.100000,0.000000)(29.100000,0.000000)(30.100000,0.000000)
(31.100000,0.000000)(32.100000,0.000000)(33.100000,0.000000)(34.100000,0.000000)(35.100000,0.000000)
(36.100000,0.000000)(37.100000,0.000000)(38.100000,0.000000)(39.100000,0.000000)(40.100000,0.000000)
(41.100000,0.000000)(42.100000,0.000000)(43.100000,0.000000)(44.100000,0.000000)(45.100000,0.000000)
(46.100000,0.000000)(47.100000,0.000000)(48.100000,0.000000)(49.100000,0.000000)

\newrgbcolor{color789.0044}{1  0  0}
\psline[plotstyle=line,linejoin=1,showpoints=false,dotstyle=Bo,dotsize=\MarkerSize,linestyle=solid,linewidth=\LineWidth,linecolor=color789.0044]
(0.758970,6.000000)(1.100000,3.114887)
\psline[plotstyle=line,linejoin=1,showpoints=false,dotstyle=Bo,dotsize=\MarkerSize,linestyle=solid,linewidth=\LineWidth,linecolor=color789.0044]
(49.100000,0.000000)(50.000000,0.000000)
\psline[plotstyle=line,linejoin=1,showpoints=true,dotstyle=Bo,dotsize=\MarkerSize,linestyle=solid,linewidth=\LineWidth,linecolor=color789.0044]
(1.100000,3.114887)(2.100000,3.005519)(3.100000,3.000274)(4.100000,3.000014)(5.100000,3.000001)
(6.100000,3.000000)(7.100000,3.000000)(8.100000,3.000000)(9.100000,3.000000)(10.100000,2.999999)
(11.100000,2.999975)(12.100000,2.999500)(13.100000,2.989996)(14.100000,2.811080)(15.100000,1.276672)
(16.100000,0.106714)(17.100000,0.005499)(18.100000,0.000274)(19.100000,0.000014)(20.100000,0.000001)
(21.100000,0.000000)(22.100000,0.000000)(23.100000,0.000000)(24.100000,0.000000)(25.100000,0.000000)
(26.100000,0.000000)(27.100000,0.000000)(28.100000,0.000000)(29.100000,0.000000)(30.100000,0.000000)
(31.100000,0.000000)(32.100000,0.000000)(33.100000,0.000000)(34.100000,0.000000)(35.100000,0.000000)
(36.100000,0.000000)(37.100000,0.000000)(38.100000,0.000000)(39.100000,0.000000)(40.100000,0.000000)
(41.100000,0.000000)(42.100000,0.000000)(43.100000,0.000000)(44.100000,0.000000)(45.100000,0.000000)
(46.100000,0.000000)(47.100000,0.000000)(48.100000,0.000000)(49.100000,0.000000)

\newrgbcolor{color790.0044}{0        0.75        0.75}
\psline[plotstyle=line,linejoin=1,showpoints=false,dotstyle=Bpentagon,dotsize=\MarkerSize,linestyle=solid,linewidth=\LineWidth,linecolor=color790.0044]
(0.714280,6.000000)(1.100000,2.670734)
\psline[plotstyle=line,linejoin=1,showpoints=false,dotstyle=Bpentagon,dotsize=\MarkerSize,linestyle=solid,linewidth=\LineWidth,linecolor=color790.0044]
(49.100000,0.000000)(50.000000,0.000000)
\psline[plotstyle=line,linejoin=1,showpoints=true,dotstyle=Bpentagon,dotsize=\MarkerSize,linestyle=solid,linewidth=\LineWidth,linecolor=color790.0044]
(1.100000,2.670734)(2.100000,2.513188)(3.100000,2.501077)(4.100000,2.500088)(5.100000,2.500007)
(6.100000,2.500001)(7.100000,2.500000)(8.100000,2.500000)(9.100000,2.500000)(10.100000,2.500000)
(11.100000,2.500000)(12.100000,2.500000)(13.100000,2.500000)(14.100000,2.499999)(15.100000,2.499988)
(16.100000,2.499854)(17.100000,2.498226)(18.100000,2.478556)(19.100000,2.261626)(20.100000,1.094559)
(21.100000,0.150217)(22.100000,0.013050)(23.100000,0.001076)(24.100000,0.000088)(25.100000,0.000007)
(26.100000,0.000001)(27.100000,0.000000)(28.100000,0.000000)(29.100000,0.000000)(30.100000,0.000000)
(31.100000,0.000000)(32.100000,0.000000)(33.100000,0.000000)(34.100000,0.000000)(35.100000,0.000000)
(36.100000,0.000000)(37.100000,0.000000)(38.100000,0.000000)(39.100000,0.000000)(40.100000,0.000000)
(41.100000,0.000000)(42.100000,0.000000)(43.100000,0.000000)(44.100000,0.000000)(45.100000,0.000000)
(46.100000,0.000000)(47.100000,0.000000)(48.100000,0.000000)(49.100000,0.000000)

\newrgbcolor{color791.0044}{0.75           0        0.75}
\psline[plotstyle=line,linejoin=1,showpoints=false,dotstyle=Basterisk,dotsize=\MarkerSize,linestyle=solid,linewidth=\LineWidth,linecolor=color791.0044]
(0.635081,6.000000)(1.100000,1.856550)
\psline[plotstyle=line,linejoin=1,showpoints=false,dotstyle=Basterisk,dotsize=\MarkerSize,linestyle=solid,linewidth=\LineWidth,linecolor=color791.0044]
(49.100000,0.000000)(50.000000,0.000000)
\psline[plotstyle=line,linejoin=1,showpoints=true,dotstyle=Basterisk,dotsize=\MarkerSize,linestyle=solid,linewidth=\LineWidth,linecolor=color791.0044]
(1.100000,1.856550)(2.100000,1.567156)(3.100000,1.514481)(4.100000,1.503207)(5.100000,1.500714)
(6.100000,1.500159)(7.100000,1.500036)(8.100000,1.500008)(9.100000,1.500002)(10.100000,1.500000)
(11.100000,1.500000)(12.100000,1.500000)(13.100000,1.500000)(14.100000,1.500000)(15.100000,1.499999)
(16.100000,1.499998)(17.100000,1.499989)(18.100000,1.499952)(19.100000,1.499785)(20.100000,1.499037)
(21.100000,1.495693)(22.100000,1.480886)(23.100000,1.417978)(24.100000,1.191194)(25.100000,0.693855)
(26.100000,0.241663)(27.100000,0.061637)(28.100000,0.014207)(29.100000,0.003193)(30.100000,0.000714)
(31.100000,0.000159)(32.100000,0.000036)(33.100000,0.000008)(34.100000,0.000002)(35.100000,0.000000)
(36.100000,0.000000)(37.100000,0.000000)(38.100000,0.000000)(39.100000,0.000000)(40.100000,0.000000)
(41.100000,0.000000)(42.100000,0.000000)

\newrgbcolor{color793.0039}{0.75        0.75           0}
\psline[plotstyle=line,linejoin=1,showpoints=false,dotstyle=B+,dotsize=\MarkerSize,linestyle=solid,linewidth=\LineWidth,linecolor=color793.0039]
(0.600405,6.000000)(1.100000,1.498961)
\psline[plotstyle=line,linejoin=1,showpoints=false,dotstyle=B+,dotsize=\MarkerSize,linestyle=solid,linewidth=\LineWidth,linecolor=color793.0039]
(49.100000,0.000000)(50.000000,0.000000)
\psline[plotstyle=line,linejoin=1,showpoints=true,dotstyle=B+,dotsize=\MarkerSize,linestyle=solid,linewidth=\LineWidth,linecolor=color793.0039]
(1.100000,1.498961)(2.100000,1.139545)(3.100000,1.047174)(4.100000,1.016852)(5.100000,1.006134)
(6.100000,1.002248)(7.100000,1.000826)(8.100000,1.000304)(9.100000,1.000112)(10.100000,1.000041)
(11.100000,1.000015)(12.100000,1.000006)(13.100000,1.000002)(14.100000,1.000001)(15.100000,1.000000)
(16.100000,0.999999)(17.100000,0.999998)(18.100000,0.999993)(19.100000,0.999982)(20.100000,0.999950)
(21.100000,0.999864)(22.100000,0.999629)(23.100000,0.998993)(24.100000,0.997268)(25.100000,0.992608)
(26.100000,0.980160)(27.100000,0.947846)(28.100000,0.869892)(29.100000,0.710950)(30.100000,0.475021)
(31.100000,0.249740)(32.100000,0.109097)(33.100000,0.043107)(34.100000,0.016302)(35.100000,0.006060)
(36.100000,0.002238)(37.100000,0.000824)(38.100000,0.000303)(39.100000,0.000112)(40.100000,0.000041)
(41.100000,0.000015)(42.100000,0.000006)(43.100000,0.000002)(44.100000,0.000001)(45.100000,0.000000)

\newrgbcolor{color794.004}{0  0  0}
\psline[plotstyle=line,linejoin=1,showpoints=true,dotstyle=*,dotsize=\MarkerSize,linestyle=dashed,dash=3pt 2pt 1pt 2pt,linewidth=\LineWidth,linecolor=color794.004]
(1.000000,1.500000)(2.000000,2.386737)(3.000000,2.936588)(4.000000,2.176839)(5.000000,2.727586)
(6.000000,2.090923)(7.000000,2.635990)(8.000000,2.048204)(9.000000,2.589860)(10.000000,2.029439)
(11.000000,2.569505)(12.000000,2.022790)(13.000000,2.562281)(14.000000,2.020861)(15.000000,2.560184)
(16.000000,2.020358)(17.000000,2.559637)(18.000000,2.020231)(19.000000,2.559500)(20.000000,2.020200)
(21.000000,2.559465)(22.000000,2.020192)(23.000000,2.559457)(24.000000,2.020190)(25.000000,2.559455)
(26.000000,2.020190)(27.000000,2.559454)(28.000000,2.020189)(29.000000,2.559454)(30.000000,2.020189)
(31.000000,2.559454)(32.000000,2.020189)(33.000000,2.559454)(34.000000,2.020189)(35.000000,2.559454)
(36.000000,2.020189)(37.000000,2.559454)(38.000000,2.020189)(39.000000,2.559454)(40.000000,2.020189)
(41.000000,2.559454)(42.000000,2.020189)(43.000000,2.559454)(44.000000,2.020189)(45.000000,2.559454)
(46.000000,2.020189)(47.000000,2.559454)(48.000000,2.020189)(49.000000,2.559454)(50.000000,2.020189)

\newrgbcolor{color795.004}{0  0  0}
\psline[plotstyle=line,linejoin=1,showpoints=true,dotstyle=Bo,dotsize=\MarkerSize,linestyle=dashed,dash=3pt 2pt 1pt 2pt,linewidth=\LineWidth,linecolor=color795.004]
(1.000000,1.500000)(2.000000,2.055934)(3.000000,2.469890)(4.000000,2.797643)(5.000000,2.785880)
(6.000000,2.584971)(7.000000,2.284036)(8.000000,2.270948)(9.000000,2.430081)(10.000000,2.691192)
(11.000000,2.698550)(12.000000,2.558551)(13.000000,2.338034)(14.000000,2.330082)(15.000000,2.449132)
(16.000000,2.642869)(17.000000,2.649252)(18.000000,2.546174)(19.000000,2.382371)(20.000000,2.376488)
(21.000000,2.464711)(22.000000,2.607959)(23.000000,2.612843)(24.000000,2.536617)(25.000000,2.414961)
(26.000000,2.410575)(27.000000,2.475931)(28.000000,2.581875)(29.000000,2.585549)(30.000000,2.529135)
(31.000000,2.438845)(32.000000,2.435587)(33.000000,2.484008)(34.000000,2.562393)(35.000000,2.565140)
(36.000000,2.523373)(37.000000,2.456398)(38.000000,2.453983)(39.000000,2.489858)(40.000000,2.547872)
(41.000000,2.549918)(42.000000,2.518989)(43.000000,2.469326)(44.000000,2.467537)(45.000000,2.494118)
(46.000000,2.537065)(47.000000,2.538586)(48.000000,2.515680)(49.000000,2.478863)(50.000000,2.477539)

{ \small 
\rput[tr](48.800000,5.817423){%
\psshadowbox[framesep=0pt,linewidth=\AxesLineWidth]{\psframebox*{\begin{tabular}{l}
\Rnode{a1}{\hspace*{0.0ex}} \hspace*{0.4cm} \Rnode{a2}{~~user 1} \\
\Rnode{a3}{\hspace*{0.0ex}} \hspace*{0.4cm} \Rnode{a4}{~~user 2} \\
\Rnode{a5}{\hspace*{0.0ex}} \hspace*{0.4cm} \Rnode{a6}{~~user 3} \\
\Rnode{a7}{\hspace*{0.0ex}} \hspace*{0.4cm} \Rnode{a8}{~~user 4} \\
\Rnode{a9}{\hspace*{0.0ex}} \hspace*{0.4cm} \Rnode{a10}{~~user 5} \\
\Rnode{a11}{\hspace*{0.0ex}} \hspace*{0.4cm} \Rnode{a12}{~~user 6} \\
\Rnode{a13}{\hspace*{0.0ex}} \hspace*{0.4cm} \Rnode{a14}{~~$p(n)(\Delta w = {\frac{10}{n}})$} \\
\Rnode{a15}{\hspace*{0.0ex}} \hspace*{0.4cm} \Rnode{a16}{~~$p(n)(\Delta w = 5 e^{-\frac{n}{10}})$} \\
\end{tabular}}
\ncline[linestyle=solid,linewidth=\LineWidth,linecolor=color787.0049]{a1}{a2} \ncput{\psdot[dotstyle=*,dotsize=\MarkerSize,linecolor=color787.0049]}
\ncline[linestyle=solid,linewidth=\LineWidth,linecolor=color788.0044]{a3}{a4} \ncput{\psdot[dotstyle=Bsquare,dotsize=\MarkerSize,linecolor=color788.0044]}
\ncline[linestyle=solid,linewidth=\LineWidth,linecolor=color789.0044]{a5}{a6} \ncput{\psdot[dotstyle=Bo,dotsize=\MarkerSize,linecolor=color789.0044]}
\ncline[linestyle=solid,linewidth=\LineWidth,linecolor=color790.0044]{a7}{a8} \ncput{\psdot[dotstyle=Bpentagon,dotsize=\MarkerSize,linecolor=color790.0044]}
\ncline[linestyle=solid,linewidth=\LineWidth,linecolor=color791.0044]{a9}{a10} \ncput{\psdot[dotstyle=Basterisk,dotsize=\MarkerSize,linecolor=color791.0044]}
\ncline[linestyle=solid,linewidth=\LineWidth,linecolor=color793.0039]{a11}{a12} \ncput{\psdot[dotstyle=B+,dotsize=\MarkerSize,linecolor=color793.0039]}
\ncline[linestyle=dashed,dash=3pt 2pt 1pt 2pt,linewidth=\LineWidth,linecolor=color794.004]{a13}{a14} \ncput{\psdot[dotstyle=*,dotsize=\MarkerSize,linecolor=color794.004]}
\ncline[linestyle=dashed,dash=3pt 2pt 1pt 2pt,linewidth=\LineWidth,linecolor=color795.004]{a15}{a16} \ncput{\psdot[dotstyle=Bo,dotsize=\MarkerSize,linecolor=color795.004]}
}%
}%
} 

\end{pspicture}%
\caption{The $\frac{\partial \log U_i(\gamma_i(P_i))}{\partial P_i} $ curve of fluctuation example in Section \ref{sec:Fluctuation_example} and the shadow price $p(n)$ from Algorithm in (\ref{alg:UE_first}) and (\ref{alg:eNodeB_first})
for $P_T=100$ (i.e. $\sum P_i^{\text{inf}}> P_T$). When BS has scarce power, Algorithm in (\ref{alg:UE_first}) and (\ref{alg:eNodeB_first}) don't guarantee convergence for shadow price and thus optimal power allocation is not   if we rely on these algorithms.}
\label{fig:osc_p_damped}
\end{figure}


 \section{Robust Distributed Algorithm}\label{sec:OuP_Algorithm}
In this section, we present a modified version of distributed algorithm in Section \ref{sec:Algorithm} to avoid the drawback discussed in section \ref{sec:conv_analy}. The modified algorithm is robust and it guarantees convergence for all values of the BS maximum power $P_T$. Our algorithm allocates powers that coincide with the Algorithm in (\ref{alg:UE_first}) and (\ref{alg:eNodeB_first}) for $\sum{P_i^{\text{inf}}> P_T}$. For $\sum{P_i^{\text{inf}} \ll P_T}$, our algorithm avoids fluctuations in the non-convergent region, as discussed in the previous section. This is achieved by adding a convergence measure $\Delta w(n)$ that senses the fluctuation in the bids $w_i$. In the case of fluctuation, our algorithm decreases the step size between the current and the previous bid $w_i(n) -w_i(n-1)$ for every user $i$ using \textit{fluctuation decay function}. The fluctuation decay function could be in the following forms:
\begin{itemize}
\item \textit{Exponential function}: It takes the form $\Delta w(n) = l_1 e^{-\frac{n}{l_2}}$.
\item \textit{Rational function}: It takes the form $\Delta w(n) = \frac{l_3}{n}$.
\end{itemize}
where $l_1, l_2, l_3$ can be adjusted to change the power of decay of the bids $w_i$. The new algorithm with the fluctuation decay function is in Algorithm (\ref{alg:OuP_UE}) and (\ref{alg:eNodeB_first}). 

\begin{rem} 
The fluctuation decay function can be included in Algorithm (\ref{alg:UE_first}) of the UE or Algorithm (\ref{alg:eNodeB_first}) of the BS.
\end{rem}
In our model, we add the decay part in Algorithm (\ref{alg:UE_first}) of the UE. Thus, the modified UE algorithm with the decay part becomes Algorithm (\ref{alg:OuP_UE}).
\begin{algorithm}
\caption{UE Algorithm}\label{alg:OuP_UE}
\begin{algorithmic}
\STATE {Send initial bid $w_i(1)$ to BS}
\LOOP
	\STATE {Receive shadow price $p(n)$ from BS}
	\IF {STOP from BS} %

	\STATE {Calculate allocated power $P_i ^{\text{opt}}=\frac{w_i(n)}{p(n)}$}
			\ELSE
	\STATE {Calculate new bid $w_i (n)= p(n) P_{i}(n)$}
	\IF {$|w_i(n) -w_i(n-1)| >\Delta w(n)$} %
	   	\STATE {$w_i(n) =w_i(n-1) + \text{sign}(w_i(n) -w_i(n-1))\Delta w(n)$}  
	   	\COMMENT {$\Delta w = l_1 e^{-\frac{n}{l_2}}$ or $\Delta w = \frac{l_3}{n}$}
	\ENDIF
	\STATE {Send new bid $w_i (n)$ to BS}
		\ENDIF 
\ENDLOOP
\end{algorithmic}
\end{algorithm}
\begin{figure}[tb]
\centering
%
\psset{xunit=0.020000\plotwidth,yunit=0.019718\plotwidth}%
\begin{pspicture}(-5.069124,-4.444444)(50.921659,40.935673)%


\psline[linewidth=\AxesLineWidth,linecolor=GridColor](0.000000,0.000000)(0.000000,0.608589)
\psline[linewidth=\AxesLineWidth,linecolor=GridColor](5.000000,0.000000)(5.000000,0.608589)
\psline[linewidth=\AxesLineWidth,linecolor=GridColor](10.000000,0.000000)(10.000000,0.608589)
\psline[linewidth=\AxesLineWidth,linecolor=GridColor](15.000000,0.000000)(15.000000,0.608589)
\psline[linewidth=\AxesLineWidth,linecolor=GridColor](20.000000,0.000000)(20.000000,0.608589)
\psline[linewidth=\AxesLineWidth,linecolor=GridColor](25.000000,0.000000)(25.000000,0.608589)
\psline[linewidth=\AxesLineWidth,linecolor=GridColor](30.000000,0.000000)(30.000000,0.608589)
\psline[linewidth=\AxesLineWidth,linecolor=GridColor](35.000000,0.000000)(35.000000,0.608589)
\psline[linewidth=\AxesLineWidth,linecolor=GridColor](40.000000,0.000000)(40.000000,0.608589)
\psline[linewidth=\AxesLineWidth,linecolor=GridColor](45.000000,0.000000)(45.000000,0.608589)
\psline[linewidth=\AxesLineWidth,linecolor=GridColor](50.000000,0.000000)(50.000000,0.608589)
\psline[linewidth=\AxesLineWidth,linecolor=GridColor](0.000000,0.000000)(0.600000,0.000000)
\psline[linewidth=\AxesLineWidth,linecolor=GridColor](0.000000,5.000000)(0.600000,5.000000)
\psline[linewidth=\AxesLineWidth,linecolor=GridColor](0.000000,10.000000)(0.600000,10.000000)
\psline[linewidth=\AxesLineWidth,linecolor=GridColor](0.000000,15.000000)(0.600000,15.000000)
\psline[linewidth=\AxesLineWidth,linecolor=GridColor](0.000000,20.000000)(0.600000,20.000000)
\psline[linewidth=\AxesLineWidth,linecolor=GridColor](0.000000,25.000000)(0.600000,25.000000)
\psline[linewidth=\AxesLineWidth,linecolor=GridColor](0.000000,30.000000)(0.600000,30.000000)
\psline[linewidth=\AxesLineWidth,linecolor=GridColor](0.000000,35.000000)(0.600000,35.000000)
\psline[linewidth=\AxesLineWidth,linecolor=GridColor](0.000000,40.000000)(0.600000,40.000000)

{ \footnotesize 
\rput[t](0.000000,-0.608589){$0$}
\rput[t](5.000000,-0.608589){$5$}
\rput[t](10.000000,-0.608589){$10$}
\rput[t](15.000000,-0.608589){$15$}
\rput[t](20.000000,-0.608589){$20$}
\rput[t](25.000000,-0.608589){$25$}
\rput[t](30.000000,-0.608589){$30$}
\rput[t](35.000000,-0.608589){$35$}
\rput[t](40.000000,-0.608589){$40$}
\rput[t](45.000000,-0.608589){$45$}
\rput[t](50.000000,-0.608589){$50$}
\rput[r](-0.600000,0.000000){$0$}
\rput[r](-0.600000,5.000000){$5$}
\rput[r](-0.600000,10.000000){$10$}
\rput[r](-0.600000,15.000000){$15$}
\rput[r](-0.600000,20.000000){$20$}
\rput[r](-0.600000,25.000000){$25$}
\rput[r](-0.600000,30.000000){$30$}
\rput[r](-0.600000,35.000000){$35$}
\rput[r](-0.600000,40.000000){$40$}
} 

\psframe[linewidth=\AxesLineWidth,dimen=middle](0.000000,0.000000)(50.000000,40.000000)

{ \small 
\rput[b](25.000000,-4.444444){
\begin{tabular}{c}
$n$\\
\end{tabular}
}

\rput[t]{90}(-5.069124,20.000000){
\begin{tabular}{c}
$P_i$\\
\end{tabular}
}
} 

\newrgbcolor{color227.0109}{0  0  1}
\psline[plotstyle=line,linejoin=1,showpoints=true,dotstyle=*,dotsize=\MarkerSize,linestyle=solid,linewidth=\LineWidth,linecolor=color227.0109]
(1.000000,5.127706)(2.000000,4.902083)(3.000000,4.746059)(4.000000,4.955674)(5.000000,4.809375)
(6.000000,4.977254)(7.000000,4.835293)(8.000000,4.987947)(9.000000,4.848021)(10.000000,4.992640)
(11.000000,4.853577)(12.000000,4.994302)(13.000000,4.855540)(14.000000,4.994784)(15.000000,4.856109)
(16.000000,4.994910)(17.000000,4.856257)(18.000000,4.994942)(19.000000,4.856295)(20.000000,4.994950)
(21.000000,4.856304)(22.000000,4.994952)(23.000000,4.856306)(24.000000,4.994952)(25.000000,4.856307)
(26.000000,4.994952)(27.000000,4.856307)(28.000000,4.994952)(29.000000,4.856307)(30.000000,4.994952)
(31.000000,4.856307)(32.000000,4.994952)(33.000000,4.856307)(34.000000,4.994952)(35.000000,4.856307)
(36.000000,4.994952)(37.000000,4.856307)(38.000000,4.994952)(39.000000,4.856307)(40.000000,4.994952)
(41.000000,4.856307)(42.000000,4.994952)(43.000000,4.856307)(44.000000,4.994952)(45.000000,4.856307)
(46.000000,4.994952)(47.000000,4.856307)(48.000000,4.994952)(49.000000,4.856307)(50.000000,4.994952)

\newrgbcolor{color228.0104}{0         0.5           0}
\psline[plotstyle=line,linejoin=1,showpoints=true,dotstyle=Bsquare,dotsize=\MarkerSize,linestyle=solid,linewidth=\LineWidth,linecolor=color228.0104]
(1.000000,10.082195)(2.000000,9.782105)(3.000000,9.528288)(4.000000,9.857757)(5.000000,9.639528)
(6.000000,9.887237)(7.000000,9.681306)(8.000000,9.901668)(9.000000,9.701211)(10.000000,9.907967)
(11.000000,9.709785)(12.000000,9.910194)(13.000000,9.712799)(14.000000,9.910839)(15.000000,9.713671)
(16.000000,9.911007)(17.000000,9.713898)(18.000000,9.911050)(19.000000,9.713956)(20.000000,9.911060)
(21.000000,9.713970)(22.000000,9.911063)(23.000000,9.713974)(24.000000,9.911063)(25.000000,9.713974)
(26.000000,9.911064)(27.000000,9.713975)(28.000000,9.911064)(29.000000,9.713975)(30.000000,9.911064)
(31.000000,9.713975)(32.000000,9.911064)(33.000000,9.713975)(34.000000,9.911064)(35.000000,9.713975)
(36.000000,9.911064)(37.000000,9.713975)(38.000000,9.911064)(39.000000,9.713975)(40.000000,9.911064)
(41.000000,9.713975)(42.000000,9.911064)(43.000000,9.713975)(44.000000,9.911064)(45.000000,9.713975)
(46.000000,9.911064)(47.000000,9.713975)(48.000000,9.911064)(49.000000,9.713975)(50.000000,9.911064)

\newrgbcolor{color229.0104}{1  0  0}
\psline[plotstyle=line,linejoin=1,showpoints=true,dotstyle=Bo,dotsize=\MarkerSize,linestyle=solid,linewidth=\LineWidth,linecolor=color229.0104]
(1.000000,15.000000)(2.000000,14.547037)(3.000000,13.721551)(4.000000,14.675841)(5.000000,14.232050)
(6.000000,14.722356)(7.000000,14.340056)(8.000000,14.744544)(9.000000,14.385713)(10.000000,14.754120)
(11.000000,14.404489)(12.000000,14.757490)(13.000000,14.410975)(14.000000,14.758465)(15.000000,14.412840)
(16.000000,14.758719)(17.000000,14.413326)(18.000000,14.758783)(19.000000,14.413448)(20.000000,14.758799)
(21.000000,14.413479)(22.000000,14.758803)(23.000000,14.413486)(24.000000,14.758804)(25.000000,14.413488)
(26.000000,14.758804)(27.000000,14.413488)(28.000000,14.758804)(29.000000,14.413489)(30.000000,14.758804)
(31.000000,14.413489)(32.000000,14.758804)(33.000000,14.413489)(34.000000,14.758804)(35.000000,14.413489)
(36.000000,14.758804)(37.000000,14.413489)(38.000000,14.758804)(39.000000,14.413489)(40.000000,14.758804)
(41.000000,14.413489)(42.000000,14.758804)(43.000000,14.413489)(44.000000,14.758804)(45.000000,14.413489)
(46.000000,14.758804)(47.000000,14.413489)(48.000000,14.758804)(49.000000,14.413489)(50.000000,14.758804)

\newrgbcolor{color230.0104}{0        0.75        0.75}
\psline[plotstyle=line,linejoin=1,showpoints=true,dotstyle=Bpentagon,dotsize=\MarkerSize,linestyle=solid,linewidth=\LineWidth,linecolor=color230.0104]
(1.000000,19.837814)(2.000000,18.780811)(3.000000,0.762406)(4.000000,19.237008)(5.000000,0.993457)
(6.000000,19.347417)(7.000000,1.185774)(8.000000,19.395405)(9.000000,1.344442)(10.000000,19.415364)
(11.000000,1.444028)(12.000000,19.422290)(13.000000,1.486800)(14.000000,19.424284)(15.000000,1.500168)
(16.000000,19.424804)(17.000000,1.503736)(18.000000,19.424935)(19.000000,1.504639)(20.000000,19.424967)
(21.000000,1.504864)(22.000000,19.424975)(23.000000,1.504920)(24.000000,19.424977)(25.000000,1.504934)
(26.000000,19.424978)(27.000000,1.504938)(28.000000,19.424978)(29.000000,1.504939)(30.000000,19.424978)
(31.000000,1.504939)(32.000000,19.424978)(33.000000,1.504939)(34.000000,19.424978)(35.000000,1.504939)
(36.000000,19.424978)(37.000000,1.504939)(38.000000,19.424978)(39.000000,1.504939)(40.000000,19.424978)
(41.000000,1.504939)(42.000000,19.424978)(43.000000,1.504939)(44.000000,19.424978)(45.000000,1.504939)
(46.000000,19.424978)(47.000000,1.504939)(48.000000,19.424978)(49.000000,1.504939)(50.000000,19.424978)

\newrgbcolor{color231.0104}{0.75           0        0.75}
\psline[plotstyle=line,linejoin=1,showpoints=true,dotstyle=Basterisk,dotsize=\MarkerSize,linestyle=solid,linewidth=\LineWidth,linecolor=color231.0104]
(1.000000,12.500000)(2.000000,0.660089)(3.000000,0.476652)(4.000000,0.778797)(5.000000,0.532245)
(6.000000,0.842450)(7.000000,0.561170)(8.000000,0.878714)(9.000000,0.577036)(10.000000,0.895798)
(11.000000,0.584345)(12.000000,0.902036)(13.000000,0.586986)(14.000000,0.903863)(15.000000,0.587758)
(16.000000,0.904342)(17.000000,0.587959)(18.000000,0.904462)(19.000000,0.588010)(20.000000,0.904492)
(21.000000,0.588023)(22.000000,0.904500)(23.000000,0.588026)(24.000000,0.904502)(25.000000,0.588027)
(26.000000,0.904502)(27.000000,0.588027)(28.000000,0.904502)(29.000000,0.588027)(30.000000,0.904502)
(31.000000,0.588027)(32.000000,0.904502)(33.000000,0.588027)(34.000000,0.904502)(35.000000,0.588027)
(36.000000,0.904502)(37.000000,0.588027)(38.000000,0.904502)(39.000000,0.588027)(40.000000,0.904502)
(41.000000,0.588027)(42.000000,0.904502)(43.000000,0.588027)(44.000000,0.904502)(45.000000,0.588027)
(46.000000,0.904502)(47.000000,0.588027)(48.000000,0.904502)(49.000000,0.588027)(50.000000,0.904502)

\newrgbcolor{color232.0104}{0.75        0.75           0}
\psline[plotstyle=line,linejoin=1,showpoints=true,dotstyle=B+,dotsize=\MarkerSize,linestyle=solid,linewidth=\LineWidth,linecolor=color232.0104]
(1.000000,1.098612)(2.000000,0.542974)(3.000000,0.416321)(4.000000,0.615042)(5.000000,0.456692)
(6.000000,0.650581)(7.000000,0.477011)(8.000000,0.669885)(9.000000,0.487958)(10.000000,0.678745)
(11.000000,0.492953)(12.000000,0.681944)(13.000000,0.494751)(14.000000,0.682877)(15.000000,0.495275)
(16.000000,0.683121)(17.000000,0.495412)(18.000000,0.683183)(19.000000,0.495447)(20.000000,0.683198)
(21.000000,0.495455)(22.000000,0.683202)(23.000000,0.495458)(24.000000,0.683203)(25.000000,0.495458)
(26.000000,0.683203)(27.000000,0.495458)(28.000000,0.683203)(29.000000,0.495458)(30.000000,0.683203)
(31.000000,0.495458)(32.000000,0.683203)(33.000000,0.495458)(34.000000,0.683203)(35.000000,0.495458)
(36.000000,0.683203)(37.000000,0.495458)(38.000000,0.683203)(39.000000,0.495458)(40.000000,0.683203)
(41.000000,0.495458)(42.000000,0.683203)(43.000000,0.495458)(44.000000,0.683203)(45.000000,0.495458)
(46.000000,0.683203)(47.000000,0.495458)(48.000000,0.683203)(49.000000,0.495458)(50.000000,0.683203)

{ \small 
\rput[tr](48.800000,38.800000){%
\psshadowbox[framesep=0pt,linewidth=\AxesLineWidth]{\psframebox*{\begin{tabular}{l}
\Rnode{a1}{\hspace*{0.0ex}} \hspace*{0.4cm} \Rnode{a2}{~~user 1} \\
\Rnode{a3}{\hspace*{0.0ex}} \hspace*{0.4cm} \Rnode{a4}{~~user 2} \\
\Rnode{a5}{\hspace*{0.0ex}} \hspace*{0.4cm} \Rnode{a6}{~~user 3} \\
\Rnode{a7}{\hspace*{0.0ex}} \hspace*{0.4cm} \Rnode{a8}{~~user 4} \\
\Rnode{a9}{\hspace*{0.0ex}} \hspace*{0.4cm} \Rnode{a10}{~~user 5} \\
\Rnode{a11}{\hspace*{0.0ex}} \hspace*{0.4cm} \Rnode{a12}{~~user 6} \\
\end{tabular}}
\ncline[linestyle=solid,linewidth=\LineWidth,linecolor=color227.0109]{a1}{a2} \ncput{\psdot[dotstyle=*,dotsize=\MarkerSize,linecolor=color227.0109]}
\ncline[linestyle=solid,linewidth=\LineWidth,linecolor=color228.0104]{a3}{a4} \ncput{\psdot[dotstyle=Bsquare,dotsize=\MarkerSize,linecolor=color228.0104]}
\ncline[linestyle=solid,linewidth=\LineWidth,linecolor=color229.0104]{a5}{a6} \ncput{\psdot[dotstyle=Bo,dotsize=\MarkerSize,linecolor=color229.0104]}
\ncline[linestyle=solid,linewidth=\LineWidth,linecolor=color230.0104]{a7}{a8} \ncput{\psdot[dotstyle=Bpentagon,dotsize=\MarkerSize,linecolor=color230.0104]}
\ncline[linestyle=solid,linewidth=\LineWidth,linecolor=color231.0104]{a9}{a10} \ncput{\psdot[dotstyle=Basterisk,dotsize=\MarkerSize,linecolor=color231.0104]}
\ncline[linestyle=solid,linewidth=\LineWidth,linecolor=color232.0104]{a11}{a12} \ncput{\psdot[dotstyle=B+,dotsize=\MarkerSize,linecolor=color232.0104]}
}%
}%
} 

\end{pspicture}%
\caption{The convergence of powers $P_i(n)$ of Algorithm in (\ref{alg:UE_first}) and (\ref{alg:eNodeB_first}) with number of iterations $n$ for different users and $P_T = 45$. It can be observed that powers don't converge and fluctuate around optimal powers. Thus the power allocation algorithm is not optimal even though the power allocation optimization problem has optimal solution.}
\label{fig:sim:powers_iter}
\end{figure}
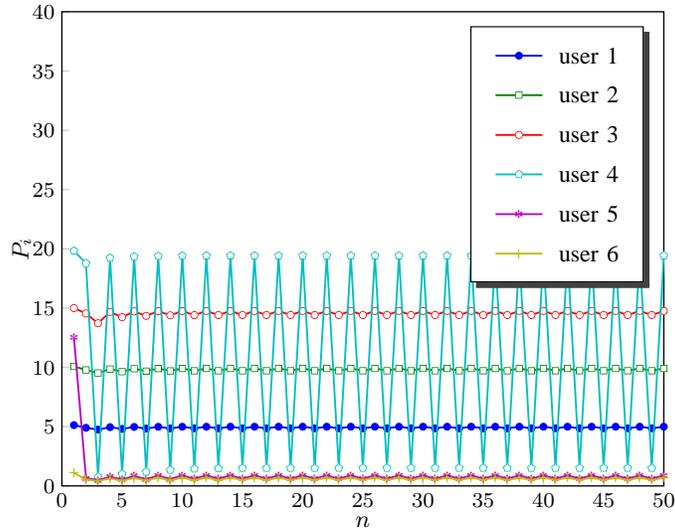
\begin{figure}[tb]
\centering
%
\psset{xunit=0.020000\plotwidth,yunit=0.015\plotwidth}%
\begin{pspicture}(-5.069124,-7.222222)(50.921659,65.000000)%


\psline[linewidth=\AxesLineWidth,linecolor=GridColor](0.000000,0.000000)(0.000000,0.988957)
\psline[linewidth=\AxesLineWidth,linecolor=GridColor](5.000000,0.000000)(5.000000,0.988957)
\psline[linewidth=\AxesLineWidth,linecolor=GridColor](10.000000,0.000000)(10.000000,0.988957)
\psline[linewidth=\AxesLineWidth,linecolor=GridColor](15.000000,0.000000)(15.000000,0.988957)
\psline[linewidth=\AxesLineWidth,linecolor=GridColor](20.000000,0.000000)(20.000000,0.988957)
\psline[linewidth=\AxesLineWidth,linecolor=GridColor](25.000000,0.000000)(25.000000,0.988957)
\psline[linewidth=\AxesLineWidth,linecolor=GridColor](30.000000,0.000000)(30.000000,0.988957)
\psline[linewidth=\AxesLineWidth,linecolor=GridColor](35.000000,0.000000)(35.000000,0.988957)
\psline[linewidth=\AxesLineWidth,linecolor=GridColor](40.000000,0.000000)(40.000000,0.988957)
\psline[linewidth=\AxesLineWidth,linecolor=GridColor](45.000000,0.000000)(45.000000,0.988957)
\psline[linewidth=\AxesLineWidth,linecolor=GridColor](50.000000,0.000000)(50.000000,0.988957)
\psline[linewidth=\AxesLineWidth,linecolor=GridColor](0.000000,0.000000)(0.600000,0.000000)
\psline[linewidth=\AxesLineWidth,linecolor=GridColor](0.000000,10.000000)(0.600000,10.000000)
\psline[linewidth=\AxesLineWidth,linecolor=GridColor](0.000000,20.000000)(0.600000,20.000000)
\psline[linewidth=\AxesLineWidth,linecolor=GridColor](0.000000,30.000000)(0.600000,30.000000)
\psline[linewidth=\AxesLineWidth,linecolor=GridColor](0.000000,40.000000)(0.600000,40.000000)
\psline[linewidth=\AxesLineWidth,linecolor=GridColor](0.000000,50.000000)(0.600000,50.000000)
\psline[linewidth=\AxesLineWidth,linecolor=GridColor](0.000000,60.000000)(0.600000,60.000000)

{ \footnotesize 
\rput[t](0.000000,-0.988957){$0$}
\rput[t](5.000000,-0.988957){$5$}
\rput[t](10.000000,-0.988957){$10$}
\rput[t](15.000000,-0.988957){$15$}
\rput[t](20.000000,-0.988957){$20$}
\rput[t](25.000000,-0.988957){$25$}
\rput[t](30.000000,-0.988957){$30$}
\rput[t](35.000000,-0.988957){$35$}
\rput[t](40.000000,-0.988957){$40$}
\rput[t](45.000000,-0.988957){$45$}
\rput[t](50.000000,-0.988957){$50$}
\rput[r](-0.600000,0.000000){$0$}
\rput[r](-0.600000,10.000000){$10$}
\rput[r](-0.600000,20.000000){$20$}
\rput[r](-0.600000,30.000000){$30$}
\rput[r](-0.600000,40.000000){$40$}
\rput[r](-0.600000,50.000000){$50$}
\rput[r](-0.600000,60.000000){$60$}
} 

\psframe[linewidth=\AxesLineWidth,dimen=middle](0.000000,0.000000)(50.000000,65.000000)

{ \small 
\rput[b](25.000000,-7.222222){
\begin{tabular}{c}
$n$\\
\end{tabular}
}

\rput[t]{90}(-5.069124,32.500000){
\begin{tabular}{c}
$w_i$\\
\end{tabular}
}
} 

\newrgbcolor{color366.0116}{0  0  1}
\psline[plotstyle=line,linejoin=1,showpoints=true,dotstyle=*,dotsize=\MarkerSize,linestyle=solid,linewidth=\LineWidth,linecolor=color366.0116]
(1.000000,7.691560)(2.000000,11.699984)(3.000000,13.937218)(4.000000,10.787707)(5.000000,13.117984)
(6.000000,10.407054)(7.000000,12.745782)(8.000000,10.216333)(9.000000,12.555697)(10.000000,10.132258)
(11.000000,12.471290)(12.000000,10.102424)(13.000000,12.441257)(14.000000,10.093767)(15.000000,12.432535)
(16.000000,10.091508)(17.000000,12.430258)(18.000000,10.090939)(19.000000,12.429684)(20.000000,10.090797)
(21.000000,12.429542)(22.000000,10.090762)(23.000000,12.429506)(24.000000,10.090753)(25.000000,12.429497)
(26.000000,10.090751)(27.000000,12.429495)(28.000000,10.090750)(29.000000,12.429494)(30.000000,10.090750)
(31.000000,12.429494)(32.000000,10.090750)(33.000000,12.429494)(34.000000,10.090750)(35.000000,12.429494)
(36.000000,10.090750)(37.000000,12.429494)(38.000000,10.090750)(39.000000,12.429494)(40.000000,10.090750)
(41.000000,12.429494)(42.000000,10.090750)(43.000000,12.429494)(44.000000,10.090750)(45.000000,12.429494)
(46.000000,10.090750)(47.000000,12.429494)(48.000000,10.090750)(49.000000,12.429494)(50.000000,10.090750)

\newrgbcolor{color367.0111}{0         0.5           0}
\psline[plotstyle=line,linejoin=1,showpoints=true,dotstyle=Bsquare,dotsize=\MarkerSize,linestyle=solid,linewidth=\LineWidth,linecolor=color367.0111]
(1.000000,15.123292)(2.000000,23.347315)(3.000000,27.980654)(4.000000,21.458753)(5.000000,26.292643)
(6.000000,20.673451)(7.000000,25.519824)(8.000000,20.280638)(9.000000,25.124780)(10.000000,20.107615)
(11.000000,24.949341)(12.000000,20.046240)(13.000000,24.886920)(14.000000,20.028431)(15.000000,24.868790)
(16.000000,20.023784)(17.000000,24.864057)(18.000000,20.022614)(19.000000,24.862866)(20.000000,20.022323)
(21.000000,24.862569)(22.000000,20.022250)(23.000000,24.862495)(24.000000,20.022232)(25.000000,24.862477)
(26.000000,20.022228)(27.000000,24.862472)(28.000000,20.022226)(29.000000,24.862471)(30.000000,20.022226)
(31.000000,24.862471)(32.000000,20.022226)(33.000000,24.862471)(34.000000,20.022226)(35.000000,24.862471)
(36.000000,20.022226)(37.000000,24.862471)(38.000000,20.022226)(39.000000,24.862471)(40.000000,20.022226)
(41.000000,24.862471)(42.000000,20.022226)(43.000000,24.862471)(44.000000,20.022226)(45.000000,24.862471)
(46.000000,20.022226)(47.000000,24.862471)(48.000000,20.022226)(49.000000,24.862471)(50.000000,20.022226)

\newrgbcolor{color368.0111}{1  0  0}
\psline[plotstyle=line,linejoin=1,showpoints=true,dotstyle=Bo,dotsize=\MarkerSize,linestyle=solid,linewidth=\LineWidth,linecolor=color368.0111]
(1.000000,22.500000)(2.000000,34.719955)(3.000000,40.294539)(4.000000,31.946948)(5.000000,38.819143)
(6.000000,30.783314)(7.000000,37.800241)(8.000000,30.199836)(9.000000,37.256984)(10.000000,29.942587)
(11.000000,37.012406)(12.000000,29.851301)(13.000000,36.924965)(14.000000,29.824811)(15.000000,36.899530)
(16.000000,29.817898)(17.000000,36.892887)(18.000000,29.816158)(19.000000,36.891215)(20.000000,29.815724)
(21.000000,36.890798)(22.000000,29.815616)(23.000000,36.890694)(24.000000,29.815589)(25.000000,36.890669)
(26.000000,29.815582)(27.000000,36.890662)(28.000000,29.815581)(29.000000,36.890661)(30.000000,29.815580)
(31.000000,36.890660)(32.000000,29.815580)(33.000000,36.890660)(34.000000,29.815580)(35.000000,36.890660)
(36.000000,29.815580)(37.000000,36.890660)(38.000000,29.815580)(39.000000,36.890660)(40.000000,29.815580)
(41.000000,36.890660)(42.000000,29.815580)(43.000000,36.890660)(44.000000,29.815580)(45.000000,36.890660)
(46.000000,29.815580)(47.000000,36.890660)(48.000000,29.815580)(49.000000,36.890660)(50.000000,29.815580)

\newrgbcolor{color369.0111}{0        0.75        0.75}
\psline[plotstyle=line,linejoin=1,showpoints=true,dotstyle=Bpentagon,dotsize=\MarkerSize,linestyle=solid,linewidth=\LineWidth,linecolor=color369.0111]
(1.000000,29.756721)(2.000000,44.824862)(3.000000,2.238871)(4.000000,41.875876)(5.000000,2.709741)
(6.000000,40.453959)(7.000000,3.125687)(8.000000,39.725749)(9.000000,3.481917)(10.000000,39.402299)
(11.000000,3.710437)(12.000000,39.287212)(13.000000,3.809598)(14.000000,39.253786)(15.000000,3.840707)
(16.000000,39.245060)(17.000000,3.849018)(18.000000,39.242864)(19.000000,3.851123)(20.000000,39.242316)
(21.000000,3.851648)(22.000000,39.242180)(23.000000,3.851779)(24.000000,39.242146)(25.000000,3.851811)
(26.000000,39.242138)(27.000000,3.851819)(28.000000,39.242136)(29.000000,3.851821)(30.000000,39.242135)
(31.000000,3.851822)(32.000000,39.242135)(33.000000,3.851822)(34.000000,39.242135)(35.000000,3.851822)
(36.000000,39.242135)(37.000000,3.851822)(38.000000,39.242135)(39.000000,3.851822)(40.000000,39.242135)
(41.000000,3.851822)(42.000000,39.242135)(43.000000,3.851822)(44.000000,39.242135)(45.000000,3.851822)
(46.000000,39.242135)(47.000000,3.851822)(48.000000,39.242135)(49.000000,3.851822)(50.000000,39.242135)

\newrgbcolor{color370.0111}{0.75           0        0.75}
\psline[plotstyle=line,linejoin=1,showpoints=true,dotstyle=Basterisk,dotsize=\MarkerSize,linestyle=solid,linewidth=\LineWidth,linecolor=color370.0111]
(1.000000,18.750000)(2.000000,1.575460)(3.000000,1.399730)(4.000000,1.695316)(5.000000,1.451744)
(6.000000,1.761498)(7.000000,1.479237)(8.000000,1.799786)(9.000000,1.494443)(10.000000,1.817967)
(11.000000,1.501477)(12.000000,1.824628)(13.000000,1.504024)(14.000000,1.826583)(15.000000,1.504768)
(16.000000,1.827094)(17.000000,1.504962)(18.000000,1.827223)(19.000000,1.505011)(20.000000,1.827255)
(21.000000,1.505024)(22.000000,1.827263)(23.000000,1.505027)(24.000000,1.827265)(25.000000,1.505027)
(26.000000,1.827266)(27.000000,1.505028)(28.000000,1.827266)(29.000000,1.505028)(30.000000,1.827266)
(31.000000,1.505028)(32.000000,1.827266)(33.000000,1.505028)(34.000000,1.827266)(35.000000,1.505028)
(36.000000,1.827266)(37.000000,1.505028)(38.000000,1.827266)(39.000000,1.505028)(40.000000,1.827266)
(41.000000,1.505028)(42.000000,1.827266)(43.000000,1.505028)(44.000000,1.827266)(45.000000,1.505028)
(46.000000,1.827266)(47.000000,1.505028)(48.000000,1.827266)(49.000000,1.505028)(50.000000,1.827266)

\newrgbcolor{color371.0111}{0.75        0.75           0}
\psline[plotstyle=line,linejoin=1,showpoints=true,dotstyle=B+,dotsize=\MarkerSize,linestyle=solid,linewidth=\LineWidth,linecolor=color371.0111]
(1.000000,1.647918)(2.000000,1.295935)(3.000000,1.222562)(4.000000,1.338847)(5.000000,1.245666)
(6.000000,1.360316)(7.000000,1.257395)(8.000000,1.372061)(9.000000,1.263743)(10.000000,1.377472)
(11.000000,1.266645)(12.000000,1.379429)(13.000000,1.267691)(14.000000,1.380000)(15.000000,1.267996)
(16.000000,1.380149)(17.000000,1.268076)(18.000000,1.380187)(19.000000,1.268096)(20.000000,1.380196)
(21.000000,1.268101)(22.000000,1.380199)(23.000000,1.268102)(24.000000,1.380199)(25.000000,1.268102)
(26.000000,1.380199)(27.000000,1.268102)(28.000000,1.380199)(29.000000,1.268102)(30.000000,1.380199)
(31.000000,1.268102)(32.000000,1.380199)(33.000000,1.268102)(34.000000,1.380199)(35.000000,1.268102)
(36.000000,1.380199)(37.000000,1.268102)(38.000000,1.380199)(39.000000,1.268102)(40.000000,1.380199)
(41.000000,1.268102)(42.000000,1.380199)(43.000000,1.268102)(44.000000,1.380199)(45.000000,1.268102)
(46.000000,1.380199)(47.000000,1.268102)(48.000000,1.380199)(49.000000,1.268102)(50.000000,1.380199)

{ \small 
\rput[tr](48.800000,63.022086){%
\psshadowbox[framesep=0pt,linewidth=\AxesLineWidth]{\psframebox*{\begin{tabular}{l}
\Rnode{a1}{\hspace*{0.0ex}} \hspace*{0.4cm} \Rnode{a2}{~~user 1} \\
\Rnode{a3}{\hspace*{0.0ex}} \hspace*{0.4cm} \Rnode{a4}{~~user 2} \\
\Rnode{a5}{\hspace*{0.0ex}} \hspace*{0.4cm} \Rnode{a6}{~~user 3} \\
\Rnode{a7}{\hspace*{0.0ex}} \hspace*{0.4cm} \Rnode{a8}{~~user 4} \\
\Rnode{a9}{\hspace*{0.0ex}} \hspace*{0.4cm} \Rnode{a10}{~~user 5} \\
\Rnode{a11}{\hspace*{0.0ex}} \hspace*{0.4cm} \Rnode{a12}{~~user 6} \\
\end{tabular}}
\ncline[linestyle=solid,linewidth=\LineWidth,linecolor=color366.0116]{a1}{a2} \ncput{\psdot[dotstyle=*,dotsize=\MarkerSize,linecolor=color366.0116]}
\ncline[linestyle=solid,linewidth=\LineWidth,linecolor=color367.0111]{a3}{a4} \ncput{\psdot[dotstyle=Bsquare,dotsize=\MarkerSize,linecolor=color367.0111]}
\ncline[linestyle=solid,linewidth=\LineWidth,linecolor=color368.0111]{a5}{a6} \ncput{\psdot[dotstyle=Bo,dotsize=\MarkerSize,linecolor=color368.0111]}
\ncline[linestyle=solid,linewidth=\LineWidth,linecolor=color369.0111]{a7}{a8} \ncput{\psdot[dotstyle=Bpentagon,dotsize=\MarkerSize,linecolor=color369.0111]}
\ncline[linestyle=solid,linewidth=\LineWidth,linecolor=color370.0111]{a9}{a10} \ncput{\psdot[dotstyle=Basterisk,dotsize=\MarkerSize,linecolor=color370.0111]}
\ncline[linestyle=solid,linewidth=\LineWidth,linecolor=color371.0111]{a11}{a12} \ncput{\psdot[dotstyle=B+,dotsize=\MarkerSize,linecolor=color371.0111]}
}%
}%
} 

\end{pspicture}%

\caption{The convergence of bids $w_i(n)$ of Algorithm in (\ref{alg:UE_first}) and (\ref{alg:eNodeB_first}) with number of iterations $n$ for different users and $P_T= 45$. For these algorithms the process of bidding doesn't converge for all users and bids fluctuate around optimal bid values.}
\label{fig:sim:bids_iter}
\end{figure}
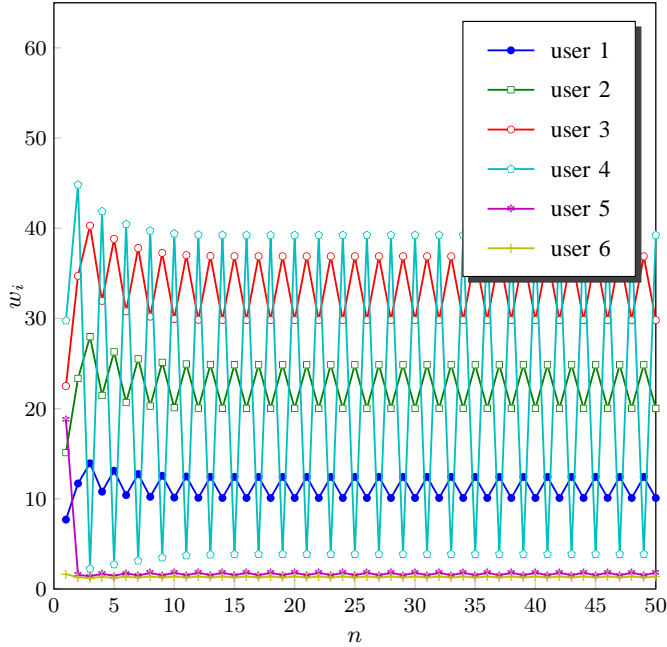
\begin{figure}[tb]
\centering
%
\psset{xunit=0.020000\plotwidth,yunit=0.030335\plotwidth}%
\begin{pspicture}(-5.069124,-2.888889)(50.921659,26.076023)%


\psline[linewidth=\AxesLineWidth,linecolor=GridColor](0.000000,0.000000)(0.000000,0.395583)
\psline[linewidth=\AxesLineWidth,linecolor=GridColor](5.000000,0.000000)(5.000000,0.395583)
\psline[linewidth=\AxesLineWidth,linecolor=GridColor](10.000000,0.000000)(10.000000,0.395583)
\psline[linewidth=\AxesLineWidth,linecolor=GridColor](15.000000,0.000000)(15.000000,0.395583)
\psline[linewidth=\AxesLineWidth,linecolor=GridColor](20.000000,0.000000)(20.000000,0.395583)
\psline[linewidth=\AxesLineWidth,linecolor=GridColor](25.000000,0.000000)(25.000000,0.395583)
\psline[linewidth=\AxesLineWidth,linecolor=GridColor](30.000000,0.000000)(30.000000,0.395583)
\psline[linewidth=\AxesLineWidth,linecolor=GridColor](35.000000,0.000000)(35.000000,0.395583)
\psline[linewidth=\AxesLineWidth,linecolor=GridColor](40.000000,0.000000)(40.000000,0.395583)
\psline[linewidth=\AxesLineWidth,linecolor=GridColor](45.000000,0.000000)(45.000000,0.395583)
\psline[linewidth=\AxesLineWidth,linecolor=GridColor](50.000000,0.000000)(50.000000,0.395583)
\psline[linewidth=\AxesLineWidth,linecolor=GridColor](0.000000,0.000000)(0.600000,0.000000)
\psline[linewidth=\AxesLineWidth,linecolor=GridColor](0.000000,5.000000)(0.600000,5.000000)
\psline[linewidth=\AxesLineWidth,linecolor=GridColor](0.000000,10.000000)(0.600000,10.000000)
\psline[linewidth=\AxesLineWidth,linecolor=GridColor](0.000000,15.000000)(0.600000,15.000000)
\psline[linewidth=\AxesLineWidth,linecolor=GridColor](0.000000,20.000000)(0.600000,20.000000)
\psline[linewidth=\AxesLineWidth,linecolor=GridColor](0.000000,25.000000)(0.600000,25.000000)

{ \footnotesize 
\rput[t](0.000000,-0.395583){$0$}
\rput[t](5.000000,-0.395583){$5$}
\rput[t](10.000000,-0.395583){$10$}
\rput[t](15.000000,-0.395583){$15$}
\rput[t](20.000000,-0.395583){$20$}
\rput[t](25.000000,-0.395583){$25$}
\rput[t](30.000000,-0.395583){$30$}
\rput[t](35.000000,-0.395583){$35$}
\rput[t](40.000000,-0.395583){$40$}
\rput[t](45.000000,-0.395583){$45$}
\rput[t](50.000000,-0.395583){$50$}
\rput[r](-0.600000,0.000000){$0$}
\rput[r](-0.600000,5.000000){$5$}
\rput[r](-0.600000,10.000000){$10$}
\rput[r](-0.600000,15.000000){$15$}
\rput[r](-0.600000,20.000000){$20$}
\rput[r](-0.600000,25.000000){$25$}
} 

\psframe[linewidth=\AxesLineWidth,dimen=middle](0.000000,0.000000)(50.000000,26.000000)

{ \small 
\rput[b](25.000000,-2.888889){
\begin{tabular}{c}
$n$\\
\end{tabular}
}

\rput[t]{90}(-5.069124,13.000000){
\begin{tabular}{c}
$P_i$\\
\end{tabular}
}
} 

\newrgbcolor{color420.0122}{0  0  1}
\psline[plotstyle=line,linejoin=1,showpoints=true,dotstyle=*,dotsize=\MarkerSize,linestyle=solid,linewidth=\LineWidth,linecolor=color420.0122]
(1.000000,5.127706)(2.000000,4.986013)(3.000000,4.880292)(4.000000,4.788877)(5.000000,4.792364)
(6.000000,4.849359)(7.000000,4.928508)(8.000000,4.931844)(9.000000,4.890775)(10.000000,4.819783)
(11.000000,4.817691)(12.000000,4.856552)(13.000000,4.914673)(14.000000,4.916718)(15.000000,4.885770)
(16.000000,4.833377)(17.000000,4.831595)(18.000000,4.859902)(19.000000,4.903216)(20.000000,4.904742)
(21.000000,4.881661)(22.000000,4.843051)(23.000000,4.841704)(24.000000,4.862480)(25.000000,4.894732)
(26.000000,4.895877)(27.000000,4.878692)(28.000000,4.850205)(29.000000,4.849201)(30.000000,4.864493)
(31.000000,4.888475)(32.000000,4.889331)(33.000000,4.876549)(34.000000,4.855510)(35.000000,4.854764)
(36.000000,4.866041)(37.000000,4.883856)(38.000000,4.884492)(39.000000,4.874994)(40.000000,4.859443)
(41.000000,4.858890)(42.000000,4.867217)(43.000000,4.880441)(44.000000,4.880914)(45.000000,4.873861)
(46.000000,4.862359)(47.000000,4.861949)(48.000000,4.868103)(49.000000,4.877914)(50.000000,4.878266)

\newrgbcolor{color421.0117}{0         0.5           0}
\psline[plotstyle=line,linejoin=1,showpoints=true,dotstyle=Bsquare,dotsize=\MarkerSize,linestyle=solid,linewidth=\LineWidth,linecolor=color421.0117]
(1.000000,10.082195)(2.000000,9.899066)(3.000000,9.750141)(4.000000,9.605117)(5.000000,9.611066)
(6.000000,9.703282)(7.000000,9.819884)(8.000000,9.824585)(9.000000,9.765617)(10.000000,9.656520)
(11.000000,9.653129)(12.000000,9.714350)(13.000000,9.800229)(14.000000,9.803151)(15.000000,9.758252)
(16.000000,9.678277)(17.000000,9.675452)(18.000000,9.719467)(19.000000,9.783746)(20.000000,9.785953)
(21.000000,9.752173)(22.000000,9.693483)(23.000000,9.691379)(24.000000,9.723390)(25.000000,9.771409)
(26.000000,9.773081)(27.000000,9.747762)(28.000000,9.704589)(29.000000,9.703037)(30.000000,9.726444)
(31.000000,9.762238)(32.000000,9.763496)(33.000000,9.744569)(34.000000,9.712753)(35.000000,9.711608)
(36.000000,9.728786)(37.000000,9.755423)(38.000000,9.756365)(39.000000,9.742247)(40.000000,9.718767)
(41.000000,9.717923)(42.000000,9.730563)(43.000000,9.750362)(44.000000,9.751065)(45.000000,9.740551)
(46.000000,9.723207)(47.000000,9.722584)(48.000000,9.731901)(49.000000,9.746605)(50.000000,9.747128)

\newrgbcolor{color422.0117}{1  0  0}
\psline[plotstyle=line,linejoin=1,showpoints=true,dotstyle=Bo,dotsize=\MarkerSize,linestyle=solid,linewidth=\LineWidth,linecolor=color422.0117]
(1.000000,13.008196)(2.000000,13.891822)(3.000000,14.487052)(4.000000,14.124501)(5.000000,14.144739)
(6.000000,14.390293)(7.000000,14.613310)(8.000000,14.621264)(9.000000,14.516605)(10.000000,14.278326)
(11.000000,14.269378)(12.000000,14.414289)(13.000000,14.579383)(14.000000,14.584499)(15.000000,14.502669)
(16.000000,14.332828)(17.000000,14.326012)(18.000000,14.425123)(19.000000,14.550012)(20.000000,14.553997)
(21.000000,14.490992)(22.000000,14.368348)(23.000000,14.363545)(24.000000,14.433323)(25.000000,14.527413)
(26.000000,14.530509)(27.000000,14.482417)(28.000000,14.393170)(29.000000,14.389754)(30.000000,14.439646)
(31.000000,14.510239)(32.000000,14.512614)(33.000000,14.476154)(34.000000,14.410875)(35.000000,14.408419)
(36.000000,14.444460)(37.000000,14.497255)(38.000000,14.499061)(39.000000,14.471569)(40.000000,14.423651)
(41.000000,14.421871)(42.000000,14.448092)(43.000000,14.487482)(44.000000,14.488846)(45.000000,14.468204)
(46.000000,14.432942)(47.000000,14.431644)(48.000000,14.450816)(49.000000,14.480152)(50.000000,14.481176)

\newrgbcolor{color423.0117}{0        0.75        0.75}
\psline[plotstyle=line,linejoin=1,showpoints=true,dotstyle=Bpentagon,dotsize=\MarkerSize,linestyle=solid,linewidth=\LineWidth,linecolor=color423.0117]
(1.000000,13.008196)(2.000000,13.891822)(3.000000,15.048344)(4.000000,10.358879)(5.000000,7.607086)
(6.000000,5.332458)(7.000000,9.120316)(8.000000,12.124600)(9.000000,13.954518)(10.000000,10.346829)
(11.000000,8.180618)(12.000000,6.483237)(13.000000,9.327559)(14.000000,11.490583)(15.000000,12.860747)
(16.000000,10.217828)(17.000000,8.579866)(18.000000,7.330424)(19.000000,9.457779)(20.000000,11.029192)
(21.000000,12.054197)(22.000000,10.115729)(23.000000,8.884972)(24.000000,7.964583)(25.000000,9.552181)
(26.000000,10.700128)(27.000000,11.464722)(28.000000,10.039176)(29.000000,9.117673)(30.000000,8.438810)
(31.000000,9.621507)(32.000000,10.463322)(33.000000,11.032506)(34.000000,9.982074)(35.000000,9.293941)
(36.000000,8.792701)(37.000000,9.672553)(38.000000,10.291567)(39.000000,10.714693)(40.000000,9.939557)
(41.000000,9.426721)(42.000000,9.056331)(43.000000,9.710196)(44.000000,10.166280)(45.000000,10.480525)
(46.000000,9.907938)(47.000000,9.526322)(48.000000,9.252452)(49.000000,9.737988)(50.000000,10.074511)

\newrgbcolor{color424.0117}{0.75           0        0.75}
\psline[plotstyle=line,linejoin=1,showpoints=true,dotstyle=Basterisk,dotsize=\MarkerSize,linestyle=solid,linewidth=\LineWidth,linecolor=color424.0117]
(1.000000,12.500000)(2.000000,4.718841)(3.000000,0.623164)(4.000000,0.512152)(5.000000,0.515414)
(6.000000,0.578774)(7.000000,0.712830)(8.000000,0.720221)(9.000000,0.640272)(10.000000,0.543353)
(11.000000,0.541068)(12.000000,0.588360)(13.000000,0.684005)(14.000000,0.688089)(15.000000,0.631961)
(16.000000,0.558882)(17.000000,0.556777)(18.000000,0.592968)(19.000000,0.662159)(20.000000,0.664971)
(21.000000,0.625334)(22.000000,0.570699)(23.000000,0.569014)(24.000000,0.596579)(25.000000,0.647038)
(26.000000,0.649030)(27.000000,0.620653)(28.000000,0.579880)(29.000000,0.578568)(30.000000,0.599439)
(31.000000,0.636419)(32.000000,0.637846)(33.000000,0.617330)(34.000000,0.586945)(35.000000,0.585938)
(36.000000,0.601661)(37.000000,0.628851)(38.000000,0.629881)(39.000000,0.614946)(40.000000,0.592331)
(41.000000,0.591566)(42.000000,0.603364)(43.000000,0.623400)(44.000000,0.624148)(45.000000,0.613223)
(46.000000,0.596409)(47.000000,0.595831)(48.000000,0.604657)(49.000000,0.619442)(50.000000,0.619989)

\newrgbcolor{color425.0117}{0.75        0.75           0}
\psline[plotstyle=line,linejoin=1,showpoints=true,dotstyle=B+,dotsize=\MarkerSize,linestyle=solid,linewidth=\LineWidth,linecolor=color425.0117]
(1.000000,1.098612)(2.000000,0.666305)(3.000000,0.518986)(4.000000,0.442301)(5.000000,0.444652)
(6.000000,0.489148)(7.000000,0.575936)(8.000000,0.580434)(9.000000,0.530193)(10.000000,0.464551)
(11.000000,0.462940)(12.000000,0.495685)(13.000000,0.558109)(14.000000,0.560663)(15.000000,0.524769)
(16.000000,0.475421)(17.000000,0.473955)(18.000000,0.498808)(19.000000,0.544296)(20.000000,0.546089)
(21.000000,0.520417)(22.000000,0.483602)(23.000000,0.482440)(24.000000,0.501248)(25.000000,0.534581)
(26.000000,0.535868)(27.000000,0.517327)(28.000000,0.489905)(29.000000,0.489007)(30.000000,0.503175)
(31.000000,0.527684)(32.000000,0.528614)(33.000000,0.515127)(34.000000,0.494723)(35.000000,0.494038)
(36.000000,0.504669)(37.000000,0.522730)(38.000000,0.523406)(39.000000,0.513545)(40.000000,0.498377)
(41.000000,0.497859)(42.000000,0.505813)(43.000000,0.519141)(44.000000,0.519635)(45.000000,0.512399)
(46.000000,0.501133)(47.000000,0.500743)(48.000000,0.506679)(49.000000,0.516526)(50.000000,0.516888)

{ \small 
\rput[tr](48.800000,25.208834){%
\psshadowbox[framesep=0pt,linewidth=\AxesLineWidth]{\psframebox*{\begin{tabular}{l}
\Rnode{a1}{\hspace*{0.0ex}} \hspace*{0.4cm} \Rnode{a2}{~~user 1} \\
\Rnode{a3}{\hspace*{0.0ex}} \hspace*{0.4cm} \Rnode{a4}{~~user 2} \\
\Rnode{a5}{\hspace*{0.0ex}} \hspace*{0.4cm} \Rnode{a6}{~~user 3} \\
\Rnode{a7}{\hspace*{0.0ex}} \hspace*{0.4cm} \Rnode{a8}{~~user 4} \\
\Rnode{a9}{\hspace*{0.0ex}} \hspace*{0.4cm} \Rnode{a10}{~~user 5} \\
\Rnode{a11}{\hspace*{0.0ex}} \hspace*{0.4cm} \Rnode{a12}{~~user 6} \\
\end{tabular}}
\ncline[linestyle=solid,linewidth=\LineWidth,linecolor=color420.0122]{a1}{a2} \ncput{\psdot[dotstyle=*,dotsize=\MarkerSize,linecolor=color420.0122]}
\ncline[linestyle=solid,linewidth=\LineWidth,linecolor=color421.0117]{a3}{a4} \ncput{\psdot[dotstyle=Bsquare,dotsize=\MarkerSize,linecolor=color421.0117]}
\ncline[linestyle=solid,linewidth=\LineWidth,linecolor=color422.0117]{a5}{a6} \ncput{\psdot[dotstyle=Bo,dotsize=\MarkerSize,linecolor=color422.0117]}
\ncline[linestyle=solid,linewidth=\LineWidth,linecolor=color423.0117]{a7}{a8} \ncput{\psdot[dotstyle=Bpentagon,dotsize=\MarkerSize,linecolor=color423.0117]}
\ncline[linestyle=solid,linewidth=\LineWidth,linecolor=color424.0117]{a9}{a10} \ncput{\psdot[dotstyle=Basterisk,dotsize=\MarkerSize,linecolor=color424.0117]}
\ncline[linestyle=solid,linewidth=\LineWidth,linecolor=color425.0117]{a11}{a12} \ncput{\psdot[dotstyle=B+,dotsize=\MarkerSize,linecolor=color425.0117]}
}%
}%
} 

\end{pspicture}%

\caption{The convergence of powers $P_i(n)$ of Algorithm in (\ref{alg:OuP_UE}) and (\ref{alg:eNodeB_first}) with number of iterations $n$ for different users and $P_T= 45$. It can be observed that there is no fluctuation in powers when using Algorithm \eqref{alg:OuP_UE} and \eqref{alg:eNodeB_first}. This is due to the introduction of fluctuation decay function in our algorithm which damps the fluctuations and the powers converge for all users.}
\label{fig:sim:powers_iteP_robust}
\end{figure}
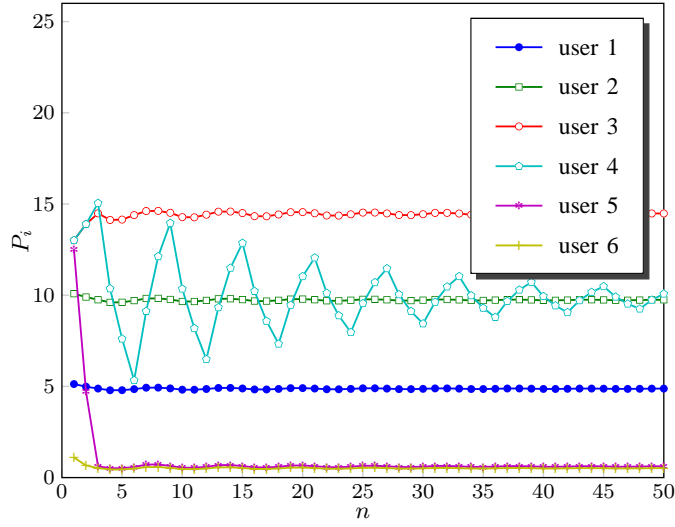
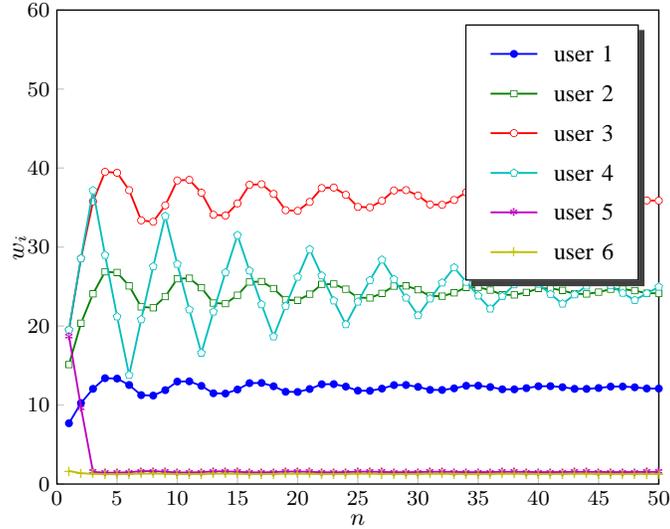
\begin{figure}[tb]
\centering
%
\psset{xunit=0.020000\plotwidth,yunit=0.013145\plotwidth}%
\begin{pspicture}(-5.069124,-6.666667)(50.921659,61.403509)%


\psline[linewidth=\AxesLineWidth,linecolor=GridColor](0.000000,0.000000)(0.000000,0.912883)
\psline[linewidth=\AxesLineWidth,linecolor=GridColor](5.000000,0.000000)(5.000000,0.912883)
\psline[linewidth=\AxesLineWidth,linecolor=GridColor](10.000000,0.000000)(10.000000,0.912883)
\psline[linewidth=\AxesLineWidth,linecolor=GridColor](15.000000,0.000000)(15.000000,0.912883)
\psline[linewidth=\AxesLineWidth,linecolor=GridColor](20.000000,0.000000)(20.000000,0.912883)
\psline[linewidth=\AxesLineWidth,linecolor=GridColor](25.000000,0.000000)(25.000000,0.912883)
\psline[linewidth=\AxesLineWidth,linecolor=GridColor](30.000000,0.000000)(30.000000,0.912883)
\psline[linewidth=\AxesLineWidth,linecolor=GridColor](35.000000,0.000000)(35.000000,0.912883)
\psline[linewidth=\AxesLineWidth,linecolor=GridColor](40.000000,0.000000)(40.000000,0.912883)
\psline[linewidth=\AxesLineWidth,linecolor=GridColor](45.000000,0.000000)(45.000000,0.912883)
\psline[linewidth=\AxesLineWidth,linecolor=GridColor](50.000000,0.000000)(50.000000,0.912883)
\psline[linewidth=\AxesLineWidth,linecolor=GridColor](0.000000,0.000000)(0.600000,0.000000)
\psline[linewidth=\AxesLineWidth,linecolor=GridColor](0.000000,10.000000)(0.600000,10.000000)
\psline[linewidth=\AxesLineWidth,linecolor=GridColor](0.000000,20.000000)(0.600000,20.000000)
\psline[linewidth=\AxesLineWidth,linecolor=GridColor](0.000000,30.000000)(0.600000,30.000000)
\psline[linewidth=\AxesLineWidth,linecolor=GridColor](0.000000,40.000000)(0.600000,40.000000)
\psline[linewidth=\AxesLineWidth,linecolor=GridColor](0.000000,50.000000)(0.600000,50.000000)
\psline[linewidth=\AxesLineWidth,linecolor=GridColor](0.000000,60.000000)(0.600000,60.000000)

{ \footnotesize 
\rput[t](0.000000,-0.912883){$0$}
\rput[t](5.000000,-0.912883){$5$}
\rput[t](10.000000,-0.912883){$10$}
\rput[t](15.000000,-0.912883){$15$}
\rput[t](20.000000,-0.912883){$20$}
\rput[t](25.000000,-0.912883){$25$}
\rput[t](30.000000,-0.912883){$30$}
\rput[t](35.000000,-0.912883){$35$}
\rput[t](40.000000,-0.912883){$40$}
\rput[t](45.000000,-0.912883){$45$}
\rput[t](50.000000,-0.912883){$50$}
\rput[r](-0.600000,0.000000){$0$}
\rput[r](-0.600000,10.000000){$10$}
\rput[r](-0.600000,20.000000){$20$}
\rput[r](-0.600000,30.000000){$30$}
\rput[r](-0.600000,40.000000){$40$}
\rput[r](-0.600000,50.000000){$50$}
\rput[r](-0.600000,60.000000){$60$}
} 

\psframe[linewidth=\AxesLineWidth,dimen=middle](0.000000,0.000000)(50.000000,60.000000)

{ \small 
\rput[b](25.000000,-6.666667){
\begin{tabular}{c}
$n$\\
\end{tabular}
}

\rput[t]{90}(-5.069124,30.000000){
\begin{tabular}{c}
$w_i$\\
\end{tabular}
}
} 

\newrgbcolor{color399.011}{0  0  1}
\psline[plotstyle=line,linejoin=1,showpoints=true,dotstyle=*,dotsize=\MarkerSize,linestyle=solid,linewidth=\LineWidth,linecolor=color399.011]
(1.000000,7.691560)(2.000000,10.250913)(3.000000,12.053781)(4.000000,13.397566)(5.000000,13.350952)
(6.000000,12.535453)(7.000000,11.256890)(8.000000,11.199960)(9.000000,11.884981)(10.000000,12.970963)
(11.000000,13.000781)(12.000000,12.425736)(13.000000,11.490672)(14.000000,11.456357)(15.000000,11.965897)
(16.000000,12.773981)(17.000000,12.800114)(18.000000,12.374154)(19.000000,11.681282)(20.000000,11.656058)
(21.000000,12.031885)(22.000000,12.630477)(23.000000,12.650613)(24.000000,12.334249)(25.000000,11.820585)
(26.000000,11.801878)(27.000000,12.079305)(28.000000,12.522622)(29.000000,12.537847)(30.000000,12.302960)
(31.000000,11.922235)(32.000000,11.908391)(33.000000,12.113387)(34.000000,12.441724)(35.000000,12.453148)
(36.000000,12.278838)(37.000000,11.996694)(38.000000,11.986460)(39.000000,12.138045)(40.000000,12.381237)
(41.000000,12.389771)(42.000000,12.260468)(43.000000,12.051400)(44.000000,12.043838)(45.000000,12.155985)
(46.000000,12.336119)(47.000000,12.342477)(48.000000,12.246590)(49.000000,12.091684)(50.000000,12.086093)

\newrgbcolor{color400.0105}{0         0.5           0}
\psline[plotstyle=line,linejoin=1,showpoints=true,dotstyle=Bsquare,dotsize=\MarkerSize,linestyle=solid,linewidth=\LineWidth,linecolor=color400.0105]
(1.000000,15.123292)(2.000000,20.351827)(3.000000,24.081771)(4.000000,26.871685)(5.000000,26.775280)
(6.000000,25.082703)(7.000000,22.428969)(8.000000,22.311117)(9.000000,23.731243)(10.000000,25.987554)
(11.000000,26.049454)(12.000000,24.854659)(13.000000,22.913271)(14.000000,22.842149)(15.000000,23.899249)
(16.000000,25.578417)(17.000000,25.632713)(18.000000,24.747451)(19.000000,23.308518)(20.000000,23.256196)
(21.000000,24.036291)(22.000000,25.280202)(23.000000,25.322052)(24.000000,24.664515)(25.000000,23.597571)
(26.000000,23.558745)(27.000000,24.134788)(28.000000,25.056034)(29.000000,25.087679)(30.000000,24.599490)
(31.000000,23.808589)(32.000000,23.779846)(33.000000,24.205588)(34.000000,24.887890)(35.000000,24.911634)
(36.000000,24.549361)(37.000000,23.963203)(38.000000,23.941952)(39.000000,24.256815)(40.000000,24.762172)
(41.000000,24.779908)(42.000000,24.511186)(43.000000,24.076826)(44.000000,24.061117)(45.000000,24.294085)
(46.000000,24.668403)(47.000000,24.681615)(48.000000,24.482348)(49.000000,24.160502)(50.000000,24.148889)

\newrgbcolor{color401.0105}{1  0  0}
\psline[plotstyle=line,linejoin=1,showpoints=true,dotstyle=Bo,dotsize=\MarkerSize,linestyle=solid,linewidth=\LineWidth,linecolor=color401.0105]
(1.000000,19.512294)(2.000000,28.560668)(3.000000,35.781419)(4.000000,39.515306)(5.000000,39.405550)
(6.000000,37.198491)(7.000000,33.377328)(8.000000,33.204126)(9.000000,35.276530)(10.000000,38.425722)
(11.000000,38.506633)(12.000000,36.879693)(13.000000,34.087094)(14.000000,33.983082)(15.000000,35.518952)
(16.000000,37.879783)(17.000000,37.953218)(18.000000,36.728867)(19.000000,34.663533)(20.000000,34.587391)
(21.000000,35.716112)(22.000000,37.472058)(23.000000,37.529688)(24.000000,36.611809)(25.000000,35.083135)
(26.000000,35.026881)(27.000000,35.857469)(28.000000,37.161361)(29.000000,37.205414)(30.000000,36.519814)
(31.000000,35.388229)(32.000000,35.346738)(33.000000,35.958884)(34.000000,36.926326)(35.000000,36.959610)
(36.000000,36.448768)(37.000000,35.611030)(38.000000,35.580446)(39.000000,36.032156)(40.000000,36.749611)
(41.000000,36.774590)(42.000000,36.394592)(43.000000,35.774321)(44.000000,35.751769)(45.000000,36.085410)
(46.000000,36.617304)(47.000000,36.635971)(48.000000,36.353628)(49.000000,35.894320)(50.000000,35.877680)

\newrgbcolor{color402.0107}{0        0.75        0.75}
\psline[plotstyle=line,linejoin=1,showpoints=true,dotstyle=Bpentagon,dotsize=\MarkerSize,linestyle=solid,linewidth=\LineWidth,linecolor=color402.0107]
(1.000000,19.512294)(2.000000,28.560668)(3.000000,37.167748)(4.000000,28.980441)(5.000000,21.192433)
(6.000000,13.784251)(7.000000,20.831132)(8.000000,27.534332)(9.000000,33.910613)(10.000000,27.845307)
(11.000000,22.075809)(12.000000,16.587692)(13.000000,21.808150)(14.000000,26.774003)(15.000000,31.497669)
(16.000000,27.004379)(17.000000,22.730230)(18.000000,18.664533)(19.000000,22.531943)(20.000000,26.210738)
(21.000000,29.710115)(22.000000,26.381405)(23.000000,23.215037)(24.000000,20.203095)(25.000000,23.068143)
(26.000000,25.793461)(27.000000,28.385863)(28.000000,25.919894)(29.000000,23.574191)(30.000000,21.342889)
(31.000000,23.465369)(32.000000,25.484334)(33.000000,27.404833)(34.000000,25.577998)(35.000000,23.840258)
(36.000000,22.187270)(37.000000,23.759641)(38.000000,25.255327)(39.000000,26.678068)(40.000000,25.324715)
(41.000000,24.037366)(42.000000,22.812802)(43.000000,23.977644)(44.000000,25.085675)(45.000000,26.139667)
(46.000000,25.137079)(47.000000,24.183387)(48.000000,23.276208)(49.000000,24.139144)(50.000000,24.959994)

\newrgbcolor{color403.0107}{0.75           0        0.75}
\psline[plotstyle=line,linejoin=1,showpoints=true,dotstyle=Basterisk,dotsize=\MarkerSize,linestyle=solid,linewidth=\LineWidth,linecolor=color403.0107]
(1.000000,18.750000)(2.000000,9.701626)(3.000000,1.539147)(4.000000,1.432818)(5.000000,1.435881)
(6.000000,1.496114)(7.000000,1.628128)(8.000000,1.635585)(9.000000,1.555913)(10.000000,1.462268)
(11.000000,1.460099)(12.000000,1.505349)(13.000000,1.599226)(14.000000,1.603305)(15.000000,1.547755)
(16.000000,1.477053)(17.000000,1.475043)(18.000000,1.509800)(19.000000,1.577509)(20.000000,1.580295)
(21.000000,1.541268)(22.000000,1.488359)(23.000000,1.486744)(24.000000,1.513293)(25.000000,1.562572)
(26.000000,1.564534)(27.000000,1.536695)(28.000000,1.497178)(29.000000,1.495916)(30.000000,1.516063)
(31.000000,1.552129)(32.000000,1.553529)(33.000000,1.533452)(34.000000,1.503984)(35.000000,1.503013)
(36.000000,1.518217)(37.000000,1.544709)(38.000000,1.545718)(39.000000,1.531129)(40.000000,1.509185)
(41.000000,1.508445)(42.000000,1.519868)(43.000000,1.539377)(44.000000,1.540108)(45.000000,1.529451)
(46.000000,1.513128)(47.000000,1.512569)(48.000000,1.521122)(49.000000,1.535513)(50.000000,1.536046)

\newrgbcolor{color404.0107}{0.75        0.75           0}
\psline[plotstyle=line,linejoin=1,showpoints=true,dotstyle=B+,dotsize=\MarkerSize,linestyle=solid,linewidth=\LineWidth,linecolor=color404.0107]
(1.000000,1.647918)(2.000000,1.369878)(3.000000,1.281839)(4.000000,1.237400)(5.000000,1.238748)
(6.000000,1.264434)(7.000000,1.315458)(8.000000,1.318136)(9.000000,1.288413)(10.000000,1.250195)
(11.000000,1.249266)(12.000000,1.268234)(13.000000,1.304878)(14.000000,1.306390)(15.000000,1.285228)
(16.000000,1.256475)(17.000000,1.255627)(18.000000,1.270053)(19.000000,1.296715)(20.000000,1.297773)
(21.000000,1.282676)(22.000000,1.261215)(23.000000,1.260541)(24.000000,1.271474)(25.000000,1.290993)
(26.000000,1.291750)(27.000000,1.280867)(28.000000,1.264874)(29.000000,1.264352)(30.000000,1.272598)
(31.000000,1.286939)(32.000000,1.287485)(33.000000,1.279579)(34.000000,1.267675)(35.000000,1.267276)
(36.000000,1.273469)(37.000000,1.284032)(38.000000,1.284429)(39.000000,1.278654)(40.000000,1.269802)
(41.000000,1.269500)(42.000000,1.274136)(43.000000,1.281929)(44.000000,1.282219)(45.000000,1.277984)
(46.000000,1.271407)(47.000000,1.271180)(48.000000,1.274642)(49.000000,1.280398)(50.000000,1.280610)

{ \small 
\rput[tr](48.800000,58.174233){%
\psshadowbox[framesep=0pt,linewidth=\AxesLineWidth]{\psframebox*{\begin{tabular}{l}
\Rnode{a1}{\hspace*{0.0ex}} \hspace*{0.4cm} \Rnode{a2}{~~user 1} \\
\Rnode{a3}{\hspace*{0.0ex}} \hspace*{0.4cm} \Rnode{a4}{~~user 2} \\
\Rnode{a5}{\hspace*{0.0ex}} \hspace*{0.4cm} \Rnode{a6}{~~user 3} \\
\Rnode{a7}{\hspace*{0.0ex}} \hspace*{0.4cm} \Rnode{a8}{~~user 4} \\
\Rnode{a9}{\hspace*{0.0ex}} \hspace*{0.4cm} \Rnode{a10}{~~user 5} \\
\Rnode{a11}{\hspace*{0.0ex}} \hspace*{0.4cm} \Rnode{a12}{~~user 6} \\
\end{tabular}}
\ncline[linestyle=solid,linewidth=\LineWidth,linecolor=color399.011]{a1}{a2} \ncput{\psdot[dotstyle=*,dotsize=\MarkerSize,linecolor=color399.011]}
\ncline[linestyle=solid,linewidth=\LineWidth,linecolor=color400.0105]{a3}{a4} \ncput{\psdot[dotstyle=Bsquare,dotsize=\MarkerSize,linecolor=color400.0105]}
\ncline[linestyle=solid,linewidth=\LineWidth,linecolor=color401.0105]{a5}{a6} \ncput{\psdot[dotstyle=Bo,dotsize=\MarkerSize,linecolor=color401.0105]}
\ncline[linestyle=solid,linewidth=\LineWidth,linecolor=color402.0107]{a7}{a8} \ncput{\psdot[dotstyle=Bpentagon,dotsize=\MarkerSize,linecolor=color402.0107]}
\ncline[linestyle=solid,linewidth=\LineWidth,linecolor=color403.0107]{a9}{a10} \ncput{\psdot[dotstyle=Basterisk,dotsize=\MarkerSize,linecolor=color403.0107]}
\ncline[linestyle=solid,linewidth=\LineWidth,linecolor=color404.0107]{a11}{a12} \ncput{\psdot[dotstyle=B+,dotsize=\MarkerSize,linecolor=color404.0107]}
}%
}%
} 

\end{pspicture}%

\caption{The convergence of bids $w_i(n)$ of Algorithm in (\ref{alg:OuP_UE}) and (\ref{alg:eNodeB_first}) with number of iterations $n$ for different users and $P_T= 45$. It can be observed that there is no fluctuation in bids $w_i(n)$ when using Algorithm \eqref{alg:OuP_UE} and \eqref{alg:eNodeB_first}. This is due to the introduction of fluctuation decay function in our algorithm which damps the fluctuations and the bidding process converges for all users.}
\label{fig:sim:bids_iteP_robust}
\end{figure}
\begin{figure}[tb]
\centering
%
\psset{xunit=0.025000\plotwidth,yunit=1.0\plotwidth}%
\begin{pspicture}(-2.027650,0.622222)(40.737327,1.316374)%
\psline[linewidth=\AxesLineWidth,linecolor=GridColor](0.000000,0.700000)(0.000000,0.710650)
\psline[linewidth=\AxesLineWidth,linecolor=GridColor](5.000000,0.700000)(5.000000,0.710650)
\psline[linewidth=\AxesLineWidth,linecolor=GridColor](10.000000,0.700000)(10.000000,0.710650)
\psline[linewidth=\AxesLineWidth,linecolor=GridColor](15.000000,0.700000)(15.000000,0.710650)
\psline[linewidth=\AxesLineWidth,linecolor=GridColor](20.000000,0.700000)(20.000000,0.710650)
\psline[linewidth=\AxesLineWidth,linecolor=GridColor](25.000000,0.700000)(25.000000,0.710650)
\psline[linewidth=\AxesLineWidth,linecolor=GridColor](30.000000,0.700000)(30.000000,0.710650)
\psline[linewidth=\AxesLineWidth,linecolor=GridColor](35.000000,0.700000)(35.000000,0.710650)
\psline[linewidth=\AxesLineWidth,linecolor=GridColor](40.000000,0.700000)(40.000000,0.710650)
\psline[linewidth=\AxesLineWidth,linecolor=GridColor](1.000000,0.700000)(1.480000,0.700000)
\psline[linewidth=\AxesLineWidth,linecolor=GridColor](1.000000,0.800000)(1.480000,0.800000)
\psline[linewidth=\AxesLineWidth,linecolor=GridColor](1.000000,0.900000)(1.480000,0.900000)
\psline[linewidth=\AxesLineWidth,linecolor=GridColor](1.000000,1.000000)(1.480000,1.000000)
\psline[linewidth=\AxesLineWidth,linecolor=GridColor](1.000000,1.100000)(1.480000,1.100000)
\psline[linewidth=\AxesLineWidth,linecolor=GridColor](1.000000,1.200000)(1.480000,1.200000)
\psline[linewidth=\AxesLineWidth,linecolor=GridColor](1.000000,1.300000)(1.480000,1.300000)
{ \footnotesize 
\rput[t](5.000000,0.689350){$5$}
\rput[t](10.000000,0.689350){$10$}
\rput[t](15.000000,0.689350){$15$}
\rput[t](20.000000,0.689350){$20$}
\rput[t](25.000000,0.689350){$25$}
\rput[t](30.000000,0.689350){$30$}
\rput[t](35.000000,0.689350){$35$}
\rput[t](40.000000,0.689350){$40$}
\rput[r](0.480000,0.700000){$0.7$}
\rput[r](0.480000,0.800000){$0.8$}
\rput[r](0.480000,0.900000){$0.9$}
\rput[r](0.480000,1.000000){$1$}
\rput[r](0.480000,1.100000){$1.1$}
\rput[r](0.480000,1.200000){$1.2$}
\rput[r](0.480000,1.300000){$1.3$}
} 
\psframe[linewidth=\AxesLineWidth,dimen=middle](1.000000,0.700000)(40.000000,1.35)
{ \small 
\rput[b](20.000000,0.6){
\begin{tabular}{c}
$n$\\
\end{tabular}
}
} 
\newrgbcolor{color18.0052}{0  0  1}
\psline[plotstyle=line,linejoin=1,showpoints=false,dotstyle=*,dotsize=\MarkerSize,linestyle=solid,linewidth=\LineWidth,linecolor=color18.0052]
(100.000000,1.000000)(100.000000,1.000000)
\psline[plotstyle=line,linejoin=1,showpoints=true,dotstyle=*,dotsize=\MarkerSize,linestyle=solid,linewidth=\LineWidth,linecolor=color18.0052]
(1.000000,1.333333)(2.000000,0.982353)(3.000000,1.275383)(4.000000,0.946488)(5.000000,1.254769)
(6.000000,0.933878)(7.000000,1.243380)(8.000000,0.926954)(9.000000,1.236732)(10.000000,0.922929)
(11.000000,1.232763)(12.000000,0.920532)(13.000000,1.230368)(14.000000,0.919088)(15.000000,1.228914)
(16.000000,0.918212)(17.000000,1.228029)(18.000000,0.917679)(19.000000,1.227489)(20.000000,0.917354)
(21.000000,1.227159)(22.000000,0.917155)(23.000000,1.226957)(24.000000,0.917034)(25.000000,1.226834)
(26.000000,0.916960)(27.000000,1.226758)(28.000000,0.916914)(29.000000,1.226712)(30.000000,0.916886)
(31.000000,1.226684)(32.000000,0.916870)(33.000000,1.226667)(34.000000,0.916859)(35.000000,1.226656)
(36.000000,0.916853)(37.000000,1.226650)(38.000000,0.916849)(39.000000,1.226646)(40.000000,0.916847)
\newrgbcolor{color19.0048}{0         0.5           0}
\psline[plotstyle=line,linejoin=1,showpoints=false,dotstyle=Bsquare,dotsize=\MarkerSize,linestyle=solid,linewidth=\LineWidth,linecolor=color19.0048]
(100.000000,1.000000)(100.000000,1.000000)
\psline[plotstyle=line,linejoin=1,showpoints=true,dotstyle=Bsquare,dotsize=\MarkerSize,linestyle=solid,linewidth=\LineWidth,linecolor=color19.0048]
(1.000000,1.333333)(2.000000,1.107999)(3.000000,0.813812)(4.000000,0.734635)(5.000000,0.755915)
(6.000000,0.837634)(7.000000,0.953436)(8.000000,1.085794)(9.000000,1.115334)(10.000000,1.097666)
(11.000000,1.045167)(12.000000,0.973572)(13.000000,0.959662)(14.000000,0.980713)(15.000000,1.022077)
(16.000000,1.024671)(17.000000,1.003953)(18.000000,0.969948)(19.000000,0.965782)(20.000000,0.979637)
(21.000000,1.003863)(22.000000,1.006305)(23.000000,0.995611)(24.000000,0.999670)(25.000000,1.012437)
(26.000000,1.011767)(27.000000,1.003072)(28.000000,0.989849)(29.000000,0.987849)(30.000000,0.992638)
(31.000000,1.001342)(32.000000,1.002101)(33.000000,0.998075)(34.000000,0.999507)(35.000000,1.004163)
(36.000000,1.003890)(37.000000,1.000674)(38.000000,0.995798)(39.000000,0.995054)(40.000000,0.996811)
{ \small 
\rput[tr](39.040000,.87){%
\psshadowbox[framesep=0pt,linewidth=\AxesLineWidth]{\psframebox*{\begin{tabular}{l}
\Rnode{a1}{\hspace*{0.0ex}} \hspace*{0.4cm} \Rnode{a2}{~~$p(n)$ of Algorithm in (\ref{alg:UE_first}) and (\ref{alg:eNodeB_first})} \\
\Rnode{a3}{\hspace*{0.0ex}} \hspace*{0.4cm} \Rnode{a4}{~~$p(n)$ of Algorithm in (\ref{alg:OuP_UE}) and (\ref{alg:eNodeB_first})} \\
\end{tabular}}
\ncline[linestyle=solid,linewidth=\LineWidth,linecolor=color18.0052]{a1}{a2}
\ncput{\psdot[dotstyle=*,dotsize=\MarkerSize,linecolor=color18.0052]}
\ncline[linestyle=solid,linewidth=\LineWidth,linecolor=color19.0048]{a3}{a4}
\ncput{\psdot[dotstyle=Bsquare,dotsize=\MarkerSize,linecolor=color19.0048]}
}%
}%
} 
\end{pspicture}%
\caption{The convergence of shadow price $p(n)$ with the number of iterations $n$. There is fluctuation in the shadow price $p(n)$ when using Algorithm \eqref{alg:UE_first} and \eqref{alg:eNodeB_first}. However, there is no fluctuation in the shadow price $p(n)$ when using Algorithm \eqref{alg:OuP_UE} and \eqref{alg:eNodeB_first}. This is due to the introduction of fluctuation decay function in our algorithm which damps the fluctuations and the optimal price converges.}
\label{fig:sim:p_compare}
\end{figure}
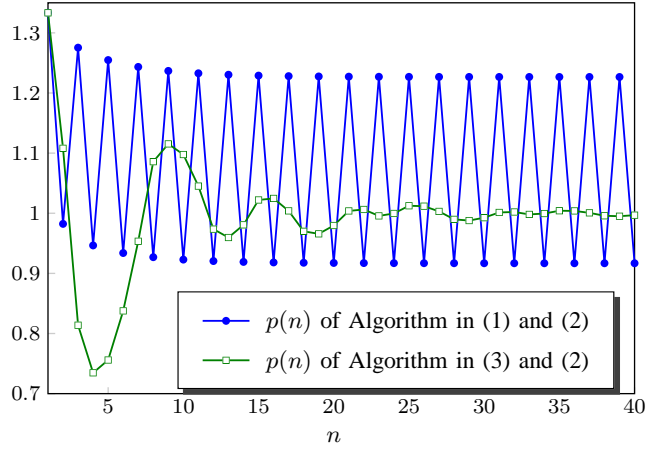
\begin{figure}[tb]
\centering
%
\psset{xunit=0.005128\plotwidth,yunit=0.017527\plotwidth}%
\begin{pspicture}(-14.769585,-5.000000)(205.391705,46.052632)%


\psline[linewidth=\AxesLineWidth,linecolor=GridColor](20.000000,0.000000)(20.000000,0.684663)
\psline[linewidth=\AxesLineWidth,linecolor=GridColor](40.000000,0.000000)(40.000000,0.684663)
\psline[linewidth=\AxesLineWidth,linecolor=GridColor](60.000000,0.000000)(60.000000,0.684663)
\psline[linewidth=\AxesLineWidth,linecolor=GridColor](80.000000,0.000000)(80.000000,0.684663)
\psline[linewidth=\AxesLineWidth,linecolor=GridColor](100.000000,0.000000)(100.000000,0.684663)
\psline[linewidth=\AxesLineWidth,linecolor=GridColor](120.000000,0.000000)(120.000000,0.684663)
\psline[linewidth=\AxesLineWidth,linecolor=GridColor](140.000000,0.000000)(140.000000,0.684663)
\psline[linewidth=\AxesLineWidth,linecolor=GridColor](160.000000,0.000000)(160.000000,0.684663)
\psline[linewidth=\AxesLineWidth,linecolor=GridColor](180.000000,0.000000)(180.000000,0.684663)
\psline[linewidth=\AxesLineWidth,linecolor=GridColor](200.000000,0.000000)(200.000000,0.684663)
\psline[linewidth=\AxesLineWidth,linecolor=GridColor](5.000000,0.000000)(7.340000,0.000000)
\psline[linewidth=\AxesLineWidth,linecolor=GridColor](5.000000,5.000000)(7.340000,5.000000)
\psline[linewidth=\AxesLineWidth,linecolor=GridColor](5.000000,10.000000)(7.340000,10.000000)
\psline[linewidth=\AxesLineWidth,linecolor=GridColor](5.000000,15.000000)(7.340000,15.000000)
\psline[linewidth=\AxesLineWidth,linecolor=GridColor](5.000000,20.000000)(7.340000,20.000000)
\psline[linewidth=\AxesLineWidth,linecolor=GridColor](5.000000,25.000000)(7.340000,25.000000)
\psline[linewidth=\AxesLineWidth,linecolor=GridColor](5.000000,30.000000)(7.340000,30.000000)
\psline[linewidth=\AxesLineWidth,linecolor=GridColor](5.000000,35.000000)(7.340000,35.000000)
\psline[linewidth=\AxesLineWidth,linecolor=GridColor](5.000000,40.000000)(7.340000,40.000000)
\psline[linewidth=\AxesLineWidth,linecolor=GridColor](5.000000,45.000000)(7.340000,45.000000)

{ \footnotesize 
\rput[t](20.000000,-0.684663){$20$}
\rput[t](40.000000,-0.684663){$40$}
\rput[t](60.000000,-0.684663){$60$}
\rput[t](80.000000,-0.684663){$80$}
\rput[t](100.000000,-0.684663){$100$}
\rput[t](120.000000,-0.684663){$120$}
\rput[t](140.000000,-0.684663){$140$}
\rput[t](160.000000,-0.684663){$160$}
\rput[t](180.000000,-0.684663){$180$}
\rput[t](200.000000,-0.684663){$200$}
\rput[r](2.660000,0.000000){$0$}
\rput[r](2.660000,5.000000){$5$}
\rput[r](2.660000,10.000000){$10$}
\rput[r](2.660000,15.000000){$15$}
\rput[r](2.660000,20.000000){$20$}
\rput[r](2.660000,25.000000){$25$}
\rput[r](2.660000,30.000000){$30$}
\rput[r](2.660000,35.000000){$35$}
\rput[r](2.660000,40.000000){$40$}
\rput[r](2.660000,45.000000){$45$}
} 

\psframe[linewidth=\AxesLineWidth,dimen=middle](5.000000,0.000000)(200.000000,45.000000)

{ \small 
\rput[b](102.500000,-5.000000){
\begin{tabular}{c}
$P_T$\\
\end{tabular}
}

\rput[t]{90}(-14.769585,22.500000){
\begin{tabular}{c}
$P_i$\\
\end{tabular}
}
} 

\newrgbcolor{color802.0059}{0  0  1}
\psline[plotstyle=line,linejoin=1,showpoints=true,dotstyle=*,dotsize=\MarkerSize,linestyle=solid,linewidth=\LineWidth,linecolor=color802.0059]
(5.000000,2.975257)(10.000000,4.513525)(15.000000,4.527825)(20.000000,4.724629)(25.000000,4.724927)
(30.000000,4.764323)(35.000000,4.875222)(40.000000,4.873760)(45.000000,4.866270)(50.000000,4.947988)
(55.000000,5.126501)(60.000000,5.126803)(65.000000,5.129056)(70.000000,5.127553)(75.000000,5.162867)
(80.000000,5.271187)(85.000000,5.274509)(90.000000,5.274564)(95.000000,5.274771)(100.000000,5.276038)
(105.000000,5.352247)(110.000000,5.656883)(115.000000,6.055640)(120.000000,6.465827)(125.000000,6.842805)
(130.000000,7.138067)(135.000000,7.331184)(140.000000,7.463590)(145.000000,7.564262)(150.000000,7.652755)
(155.000000,7.708213)(160.000000,7.745787)(165.000000,7.783328)(170.000000,7.839837)(175.000000,7.871280)
(180.000000,7.877231)(185.000000,7.933789)(190.000000,7.977364)(195.000000,8.008106)(200.000000,8.026146)

\newrgbcolor{color803.0044}{0         0.5           0}
\psline[plotstyle=line,linejoin=1,showpoints=true,dotstyle=Bsquare,dotsize=\MarkerSize,linestyle=solid,linewidth=\LineWidth,linecolor=color803.0044]
(5.000000,0.595863)(10.000000,3.631271)(15.000000,8.590023)(20.000000,9.486630)(25.000000,9.487228)
(30.000000,9.561984)(35.000000,9.742587)(40.000000,9.740400)(45.000000,9.729133)(50.000000,9.847131)
(55.000000,10.080688)(60.000000,10.081065)(65.000000,10.083881)(70.000000,10.082003)(75.000000,10.125872)
(80.000000,10.257637)(85.000000,10.261625)(90.000000,10.261690)(95.000000,10.261939)(100.000000,10.263459)
(105.000000,10.354263)(110.000000,10.709607)(115.000000,11.167694)(120.000000,11.636963)(125.000000,12.067885)
(130.000000,12.405345)(135.000000,12.626055)(140.000000,12.777377)(145.000000,12.892432)(150.000000,12.993567)
(155.000000,13.056948)(160.000000,13.099890)(165.000000,13.142794)(170.000000,13.207375)(175.000000,13.243311)
(180.000000,13.250112)(185.000000,13.314750)(190.000000,13.364550)(195.000000,13.399683)(200.000000,13.420301)

\newrgbcolor{color804.0044}{1  0  0}
\psline[plotstyle=line,linejoin=1,showpoints=true,dotstyle=Bo,dotsize=\MarkerSize,linestyle=solid,linewidth=\LineWidth,linecolor=color804.0044]
(5.000000,0.462965)(10.000000,0.648640)(15.000000,0.663692)(20.000000,4.191672)(25.000000,9.190915)
(30.000000,13.944717)(35.000000,14.472241)(40.000000,14.467903)(45.000000,14.445170)(50.000000,14.658638)
(55.000000,14.997990)(60.000000,14.998494)(65.000000,15.002247)(70.000000,14.999744)(75.000000,15.057639)
(80.000000,15.225846)(85.000000,15.230833)(90.000000,15.230915)(95.000000,15.231226)(100.000000,15.233125)
(105.000000,15.345445)(110.000000,15.771824)(115.000000,16.309993)(120.000000,16.858226)(125.000000,17.361109)
(130.000000,17.754840)(135.000000,18.012342)(140.000000,18.188886)(145.000000,18.323119)(150.000000,18.441109)
(155.000000,18.515055)(160.000000,18.565154)(165.000000,18.615208)(170.000000,18.690554)(175.000000,18.732478)
(180.000000,18.740413)(185.000000,18.815824)(190.000000,18.873924)(195.000000,18.914913)(200.000000,18.938967)

\newrgbcolor{color805.0044}{0        0.75        0.75}
\psline[plotstyle=line,linejoin=1,showpoints=true,dotstyle=Bpentagon,dotsize=\MarkerSize,linestyle=solid,linewidth=\LineWidth,linecolor=color805.0044]
(5.000000,0.392909)(10.000000,0.501107)(15.000000,0.508534)(20.000000,0.715272)(25.000000,0.715866)
(30.000000,0.811445)(35.000000,4.976937)(40.000000,9.888626)(45.000000,14.617756)(50.000000,19.191845)
(55.000000,19.834796)(60.000000,19.835553)(65.000000,19.841182)(70.000000,19.837430)(75.000000,19.922553)
(80.000000,20.155243)(85.000000,20.161898)(90.000000,20.162007)(95.000000,20.162422)(100.000000,20.164953)
(105.000000,20.312168)(110.000000,20.845283)(115.000000,21.497488)(120.000000,22.156639)(125.000000,22.760335)
(130.000000,23.232859)(135.000000,23.541872)(140.000000,23.753729)(145.000000,23.914810)(150.000000,24.056400)
(155.000000,24.145135)(160.000000,24.205254)(165.000000,24.265320)(170.000000,24.355734)(175.000000,24.406044)
(180.000000,24.415566)(185.000000,24.506059)(190.000000,24.575779)(195.000000,24.624967)(200.000000,24.653831)

\newrgbcolor{color806.0045}{0.75           0        0.75}
\psline[plotstyle=line,linejoin=1,showpoints=true,dotstyle=Basterisk,dotsize=\MarkerSize,linestyle=solid,linewidth=\LineWidth,linecolor=color806.0045]
(5.000000,0.313682)(10.000000,0.373078)(15.000000,0.376767)(20.000000,0.461620)(25.000000,0.461818)
(30.000000,0.490827)(35.000000,0.615293)(40.000000,0.613070)(45.000000,0.601992)(50.000000,0.758712)
(55.000000,3.871342)(60.000000,8.741221)(65.000000,14.041896)(70.000000,18.819921)(75.000000,23.426411)
(80.000000,24.516925)(85.000000,24.537038)(90.000000,24.537365)(95.000000,24.538608)(100.000000,24.546181)
(105.000000,24.936710)(110.000000,26.012257)(115.000000,27.144661)(120.000000,28.251820)(125.000000,29.259561)
(130.000000,30.047411)(135.000000,30.562506)(140.000000,30.915628)(145.000000,31.184108)(150.000000,31.420099)
(155.000000,31.567994)(160.000000,31.668194)(165.000000,31.768306)(170.000000,31.918999)(175.000000,32.002849)
(180.000000,32.018719)(185.000000,32.169542)(190.000000,32.285744)(195.000000,32.367723)(200.000000,32.415831)

\newrgbcolor{color807.0044}{0.75        0.75           0}
\psline[plotstyle=line,linejoin=1,showpoints=true,dotstyle=B|,dotsize=\MarkerSize,linestyle=solid,linewidth=\LineWidth,linecolor=color807.0044]
(5.000000,0.287971)(10.000000,0.336473)(15.000000,0.339421)(20.000000,0.405107)(25.000000,0.405255)
(30.000000,0.426780)(35.000000,0.513775)(40.000000,0.512297)(45.000000,0.504892)(50.000000,0.603383)
(55.000000,1.092617)(60.000000,1.094113)(65.000000,1.105400)(70.000000,1.097845)(75.000000,1.307813)
(80.000000,4.573169)(85.000000,9.524655)(90.000000,14.517176)(95.000000,19.567203)(100.000000,24.518184)
(105.000000,28.701552)(110.000000,30.996911)(115.000000,32.791284)(120.000000,34.468450)(125.000000,35.983036)
(130.000000,37.165394)(135.000000,37.938175)(140.000000,38.467907)(145.000000,38.870650)(150.000000,39.224650)
(155.000000,39.446500)(160.000000,39.596803)(165.000000,39.746974)(170.000000,39.973017)(175.000000,40.098795)
(180.000000,40.122600)(185.000000,40.348838)(190.000000,40.523142)(195.000000,40.646111)(200.000000,40.718273)

{ \small 
\rput(38.916331,31.066971){%
\psshadowbox[framesep=0pt,linewidth=\AxesLineWidth]{\psframebox*{\begin{tabular}{l}
\Rnode{a1}{\hspace*{0.0ex}} \hspace*{0.4cm} \Rnode{a2}{~~user 1} \\
\Rnode{a3}{\hspace*{0.0ex}} \hspace*{0.4cm} \Rnode{a4}{~~user 2} \\
\Rnode{a5}{\hspace*{0.0ex}} \hspace*{0.4cm} \Rnode{a6}{~~user 3} \\
\Rnode{a7}{\hspace*{0.0ex}} \hspace*{0.4cm} \Rnode{a8}{~~user 4} \\
\Rnode{a9}{\hspace*{0.0ex}} \hspace*{0.4cm} \Rnode{a10}{~~user 5} \\
\Rnode{a11}{\hspace*{0.0ex}} \hspace*{0.4cm} \Rnode{a12}{~~user 6} \\
\end{tabular}}
\ncline[linestyle=solid,linewidth=\LineWidth,linecolor=color802.0059]{a1}{a2} \ncput{\psdot[dotstyle=*,dotsize=\MarkerSize,linecolor=color802.0059]}
\ncline[linestyle=solid,linewidth=\LineWidth,linecolor=color803.0044]{a3}{a4} \ncput{\psdot[dotstyle=Bsquare,dotsize=\MarkerSize,linecolor=color803.0044]}
\ncline[linestyle=solid,linewidth=\LineWidth,linecolor=color804.0044]{a5}{a6} \ncput{\psdot[dotstyle=Bo,dotsize=\MarkerSize,linecolor=color804.0044]}
\ncline[linestyle=solid,linewidth=\LineWidth,linecolor=color805.0044]{a7}{a8} \ncput{\psdot[dotstyle=Bpentagon,dotsize=\MarkerSize,linecolor=color805.0044]}
\ncline[linestyle=solid,linewidth=\LineWidth,linecolor=color806.0045]{a9}{a10} \ncput{\psdot[dotstyle=Basterisk,dotsize=\MarkerSize,linecolor=color806.0045]}
\ncline[linestyle=solid,linewidth=\LineWidth,linecolor=color807.0044]{a11}{a12} \ncput{\psdot[dotstyle=B|,dotsize=\MarkerSize,linecolor=color807.0044]}
}%
}%
} 

\end{pspicture}%

\caption{The allocated powers $P_i$ for different values of $P_T$, $n_0=20$ and $\delta = 10^{-3}$ for Algorithm in (\ref{alg:OuP_UE}) and (\ref{alg:eNodeB_first}). Our algorithm gives priority to users running lower modulation schemes and thus first allocates powers to them until they reach their inflection point. Users running higher modulation schemes are followed by this. However, our algorithm ensures non-zero allocation of powers to all users thus maintaining a minimum QoS of all users.}
\label{fig:sim:powers_robust}
\end{figure}
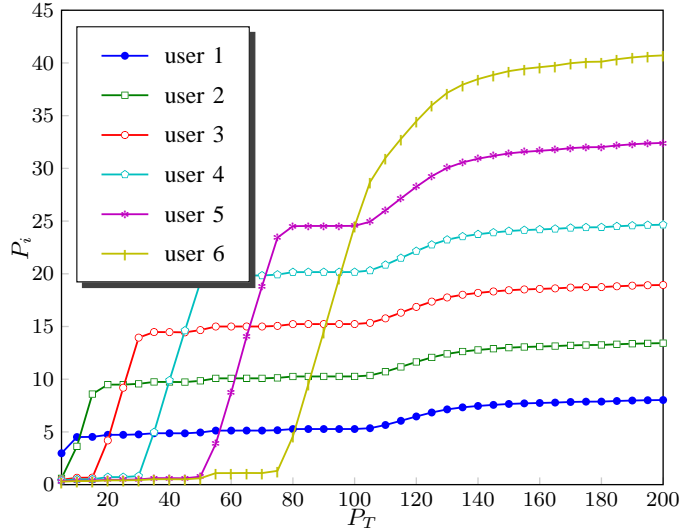
\begin{figure}[tb]
\centering
%
\psset{xunit=0.005128\plotwidth,yunit=0.005634\plotwidth}%
\begin{pspicture}(-17.914747,-15.555556)(205.391705,143.274854)%


\psline[linewidth=\AxesLineWidth,linecolor=GridColor](20.000000,0.000000)(20.000000,2.130061)
\psline[linewidth=\AxesLineWidth,linecolor=GridColor](40.000000,0.000000)(40.000000,2.130061)
\psline[linewidth=\AxesLineWidth,linecolor=GridColor](60.000000,0.000000)(60.000000,2.130061)
\psline[linewidth=\AxesLineWidth,linecolor=GridColor](80.000000,0.000000)(80.000000,2.130061)
\psline[linewidth=\AxesLineWidth,linecolor=GridColor](100.000000,0.000000)(100.000000,2.130061)
\psline[linewidth=\AxesLineWidth,linecolor=GridColor](120.000000,0.000000)(120.000000,2.130061)
\psline[linewidth=\AxesLineWidth,linecolor=GridColor](140.000000,0.000000)(140.000000,2.130061)
\psline[linewidth=\AxesLineWidth,linecolor=GridColor](160.000000,0.000000)(160.000000,2.130061)
\psline[linewidth=\AxesLineWidth,linecolor=GridColor](180.000000,0.000000)(180.000000,2.130061)
\psline[linewidth=\AxesLineWidth,linecolor=GridColor](200.000000,0.000000)(200.000000,2.130061)
\psline[linewidth=\AxesLineWidth,linecolor=GridColor](5.000000,0.000000)(7.340000,0.000000)
\psline[linewidth=\AxesLineWidth,linecolor=GridColor](5.000000,20.000000)(7.340000,20.000000)
\psline[linewidth=\AxesLineWidth,linecolor=GridColor](5.000000,40.000000)(7.340000,40.000000)
\psline[linewidth=\AxesLineWidth,linecolor=GridColor](5.000000,60.000000)(7.340000,60.000000)
\psline[linewidth=\AxesLineWidth,linecolor=GridColor](5.000000,80.000000)(7.340000,80.000000)
\psline[linewidth=\AxesLineWidth,linecolor=GridColor](5.000000,100.000000)(7.340000,100.000000)
\psline[linewidth=\AxesLineWidth,linecolor=GridColor](5.000000,120.000000)(7.340000,120.000000)
\psline[linewidth=\AxesLineWidth,linecolor=GridColor](5.000000,140.000000)(7.340000,140.000000)

{ \footnotesize 
\rput[t](20.000000,-2.130061){$20$}
\rput[t](40.000000,-2.130061){$40$}
\rput[t](60.000000,-2.130061){$60$}
\rput[t](80.000000,-2.130061){$80$}
\rput[t](100.000000,-2.130061){$100$}
\rput[t](120.000000,-2.130061){$120$}
\rput[t](140.000000,-2.130061){$140$}
\rput[t](160.000000,-2.130061){$160$}
\rput[t](180.000000,-2.130061){$180$}
\rput[t](200.000000,-2.130061){$200$}
\rput[r](2.660000,0.000000){$0$}
\rput[r](2.660000,20.000000){$20$}
\rput[r](2.660000,40.000000){$40$}
\rput[r](2.660000,60.000000){$60$}
\rput[r](2.660000,80.000000){$80$}
\rput[r](2.660000,100.000000){$100$}
\rput[r](2.660000,120.000000){$120$}
\rput[r](2.660000,140.000000){$140$}
} 

\psframe[linewidth=\AxesLineWidth,dimen=middle](5.000000,0.000000)(200.000000,140.000000)

{ \small 
\rput[b](102.500000,-15.555556){
\begin{tabular}{c}
$P_T$\\
\end{tabular}
}

\rput[t]{90}(-17.914747,70.000000){
\begin{tabular}{c}
$\sum P_i$\\
\end{tabular}
}
} 

\newrgbcolor{color819.0052}{0  0  1}
\psline[plotstyle=line,linejoin=1,linestyle=solid,linewidth=\LineWidth,linecolor=color819.0052]
(5.000000,5.028648)(10.000000,10.004093)(15.000000,15.006261)(20.000000,19.984929)(25.000000,24.986009)
(30.000000,30.000076)(35.000000,35.196055)(40.000000,40.096055)(45.000000,44.765214)(50.000000,50.007697)
(55.000000,55.003934)(60.000000,59.877250)(65.000000,65.203662)(70.000000,69.964495)(75.000000,75.003155)
(80.000000,80.000008)(85.000000,84.990559)(90.000000,89.983717)(95.000000,95.036168)(100.000000,100.001939)
(105.000000,105.002386)(110.000000,109.992765)(115.000000,114.966759)(120.000000,119.837923)(125.000000,124.274731)
(130.000000,127.743916)(135.000000,130.012135)(140.000000,131.567117)(145.000000,132.749381)(150.000000,133.788580)
(155.000000,134.439847)(160.000000,134.881082)(165.000000,135.321931)(170.000000,135.985516)(175.000000,136.354757)
(180.000000,136.424640)(185.000000,137.088802)(190.000000,137.600503)(195.000000,137.961504)(200.000000,138.173348)

\end{pspicture}%
\caption{The sum of power $\sum P_i$ for different values of $P_T$ and $\delta = 10^{-3}$ for Algorithm in (\ref{alg:OuP_UE}) and (\ref{alg:eNodeB_first}). It can be seen that when the total
power $P_T$ exceeds the sum of the inflection power $\sum P_i = \sum b_i$ of all the users, all the users are assigned power according to their desired modulation scheme's requirement. Thus satisfying users' demands and the power
allocation algorithm reaches a steady state.}
\label{fig:sim:SumOfPower}
\end{figure}
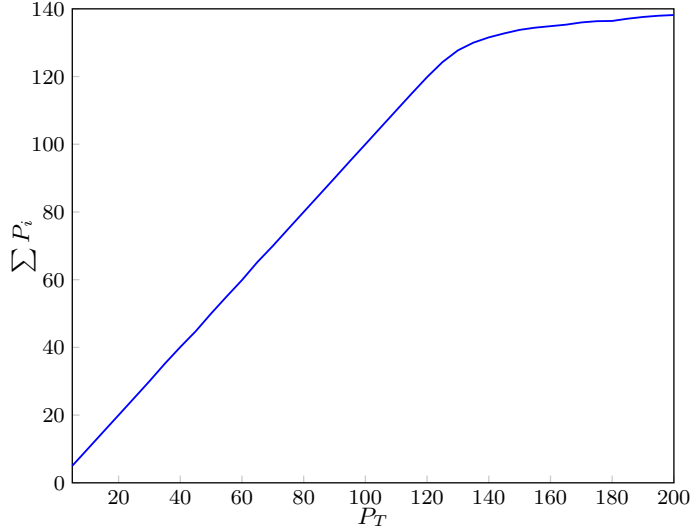

\begin{figure}[tb]
\centering
%
\psset{xunit=0.005128\plotwidth,yunit=0.017527\plotwidth}%
\begin{pspicture}(-14.769585,-5.000000)(205.391705,46.052632)%


\psline[linewidth=\AxesLineWidth,linecolor=GridColor](20.000000,0.000000)(20.000000,0.684663)
\psline[linewidth=\AxesLineWidth,linecolor=GridColor](40.000000,0.000000)(40.000000,0.684663)
\psline[linewidth=\AxesLineWidth,linecolor=GridColor](60.000000,0.000000)(60.000000,0.684663)
\psline[linewidth=\AxesLineWidth,linecolor=GridColor](80.000000,0.000000)(80.000000,0.684663)
\psline[linewidth=\AxesLineWidth,linecolor=GridColor](100.000000,0.000000)(100.000000,0.684663)
\psline[linewidth=\AxesLineWidth,linecolor=GridColor](120.000000,0.000000)(120.000000,0.684663)
\psline[linewidth=\AxesLineWidth,linecolor=GridColor](140.000000,0.000000)(140.000000,0.684663)
\psline[linewidth=\AxesLineWidth,linecolor=GridColor](160.000000,0.000000)(160.000000,0.684663)
\psline[linewidth=\AxesLineWidth,linecolor=GridColor](180.000000,0.000000)(180.000000,0.684663)
\psline[linewidth=\AxesLineWidth,linecolor=GridColor](200.000000,0.000000)(200.000000,0.684663)
\psline[linewidth=\AxesLineWidth,linecolor=GridColor](5.000000,0.000000)(7.340000,0.000000)
\psline[linewidth=\AxesLineWidth,linecolor=GridColor](5.000000,5.000000)(7.340000,5.000000)
\psline[linewidth=\AxesLineWidth,linecolor=GridColor](5.000000,10.000000)(7.340000,10.000000)
\psline[linewidth=\AxesLineWidth,linecolor=GridColor](5.000000,15.000000)(7.340000,15.000000)
\psline[linewidth=\AxesLineWidth,linecolor=GridColor](5.000000,20.000000)(7.340000,20.000000)
\psline[linewidth=\AxesLineWidth,linecolor=GridColor](5.000000,25.000000)(7.340000,25.000000)
\psline[linewidth=\AxesLineWidth,linecolor=GridColor](5.000000,30.000000)(7.340000,30.000000)
\psline[linewidth=\AxesLineWidth,linecolor=GridColor](5.000000,35.000000)(7.340000,35.000000)
\psline[linewidth=\AxesLineWidth,linecolor=GridColor](5.000000,40.000000)(7.340000,40.000000)
\psline[linewidth=\AxesLineWidth,linecolor=GridColor](5.000000,45.000000)(7.340000,45.000000)

{ \footnotesize 
\rput[t](20.000000,-0.684663){$20$}
\rput[t](40.000000,-0.684663){$40$}
\rput[t](60.000000,-0.684663){$60$}
\rput[t](80.000000,-0.684663){$80$}
\rput[t](100.000000,-0.684663){$100$}
\rput[t](120.000000,-0.684663){$120$}
\rput[t](140.000000,-0.684663){$140$}
\rput[t](160.000000,-0.684663){$160$}
\rput[t](180.000000,-0.684663){$180$}
\rput[t](200.000000,-0.684663){$200$}
\rput[r](2.660000,0.000000){$0$}
\rput[r](2.660000,5.000000){$5$}
\rput[r](2.660000,10.000000){$10$}
\rput[r](2.660000,15.000000){$15$}
\rput[r](2.660000,20.000000){$20$}
\rput[r](2.660000,25.000000){$25$}
\rput[r](2.660000,30.000000){$30$}
\rput[r](2.660000,35.000000){$35$}
\rput[r](2.660000,40.000000){$40$}
\rput[r](2.660000,45.000000){$45$}
} 

\psframe[linewidth=\AxesLineWidth,dimen=middle](5.000000,0.000000)(200.000000,45.000000)

{ \small 
\rput[b](102.500000,-5.000000){
\begin{tabular}{c}
$P_T$\\
\end{tabular}
}

\rput[t]{90}(-14.769585,22.500000){
\begin{tabular}{c}
$w_i$\\
\end{tabular}
}
} 

\newrgbcolor{color957.0038}{0  0  1}
\psline[plotstyle=line,linejoin=1,showpoints=true,dotstyle=*,dotsize=\MarkerSize,linestyle=solid,linewidth=\LineWidth,linecolor=color957.0038]
(5.000000,11.890731)(10.000000,15.797316)(15.000000,15.731605)(20.000000,14.184060)(25.000000,14.180731)
(30.000000,13.714513)(35.000000,12.134446)(40.000000,12.157585)(45.000000,12.275260)(50.000000,10.921705)
(55.000000,7.712940)(60.000000,7.707581)(65.000000,7.667650)(70.000000,7.694284)(75.000000,7.076402)
(80.000000,5.326181)(85.000000,5.276786)(90.000000,5.275979)(95.000000,5.272907)(100.000000,5.254153)
(105.000000,4.204585)(110.000000,1.524808)(115.000000,0.349998)(120.000000,0.073291)(125.000000,0.017208)
(130.000000,0.005513)(135.000000,0.002615)(140.000000,0.001568)(145.000000,0.001062)(150.000000,0.000754)
(155.000000,0.000609)(160.000000,0.000526)(165.000000,0.000455)(170.000000,0.000366)(175.000000,0.000324)
(180.000000,0.000316)(185.000000,0.000254)(190.000000,0.000215)(195.000000,0.000191)(200.000000,0.000178)

\newrgbcolor{color958.0033}{0         0.5           0}
\psline[plotstyle=line,linejoin=1,showpoints=true,dotstyle=Bsquare,dotsize=\MarkerSize,linestyle=solid,linewidth=\LineWidth,linecolor=color958.0033]
(5.000000,2.381391)(10.000000,12.709429)(15.000000,29.845419)(20.000000,28.480321)(25.000000,28.473631)
(30.000000,27.524992)(35.000000,24.249336)(40.000000,24.297409)(45.000000,24.541925)(50.000000,21.735598)
(55.000000,15.166628)(60.000000,15.155767)(65.000000,15.074835)(70.000000,15.128813)(75.000000,13.878865)
(80.000000,10.364654)(85.000000,10.266055)(90.000000,10.264443)(95.000000,10.258313)(100.000000,10.220886)
(105.000000,8.134037)(110.000000,2.886766)(115.000000,0.645459)(120.000000,0.131906)(125.000000,0.030349)
(130.000000,0.009580)(135.000000,0.004504)(140.000000,0.002684)(145.000000,0.001811)(150.000000,0.001281)
(155.000000,0.001031)(160.000000,0.000890)(165.000000,0.000768)(170.000000,0.000616)(175.000000,0.000545)
(180.000000,0.000532)(185.000000,0.000426)(190.000000,0.000360)(195.000000,0.000319)(200.000000,0.000297)

\newrgbcolor{color959.0033}{1  0  0}
\psline[plotstyle=line,linejoin=1,showpoints=true,dotstyle=Bo,dotsize=\MarkerSize,linestyle=solid,linewidth=\LineWidth,linecolor=color959.0033]
(5.000000,1.850259)(10.000000,2.270235)(15.000000,2.305949)(20.000000,12.584042)(25.000000,27.584317)
(30.000000,40.141068)(35.000000,36.021465)(40.000000,36.090157)(45.000000,36.438220)(50.000000,32.356049)
(55.000000,22.564823)(60.000000,22.548577)(65.000000,22.427515)(70.000000,22.508259)(75.000000,20.638514)
(80.000000,15.384696)(85.000000,15.237409)(90.000000,15.235001)(95.000000,15.225844)(100.000000,15.169938)
(105.000000,12.054978)(110.000000,4.251282)(115.000000,0.942668)(120.000000,0.191090)(125.000000,0.043660)
(130.000000,0.013712)(135.000000,0.006426)(140.000000,0.003821)(145.000000,0.002573)(150.000000,0.001818)
(155.000000,0.001462)(160.000000,0.001261)(165.000000,0.001088)(170.000000,0.000872)(175.000000,0.000770)
(180.000000,0.000753)(185.000000,0.000603)(190.000000,0.000508)(195.000000,0.000450)(200.000000,0.000419)

\newrgbcolor{color960.0033}{0        0.75        0.75}
\psline[plotstyle=line,linejoin=1,showpoints=true,dotstyle=Bpentagon,dotsize=\MarkerSize,linestyle=solid,linewidth=\LineWidth,linecolor=color960.0033]
(5.000000,1.570278)(10.000000,1.753871)(15.000000,1.766864)(20.000000,2.147355)(25.000000,2.148498)
(30.000000,2.335813)(35.000000,12.387615)(40.000000,24.667160)(45.000000,36.873572)(50.000000,42.362208)
(55.000000,29.841911)(60.000000,29.820561)(65.000000,29.661450)(70.000000,29.767575)(75.000000,27.306531)
(80.000000,20.365520)(85.000000,20.170603)(90.000000,20.167416)(95.000000,20.155297)(100.000000,20.081309)
(105.000000,15.956705)(110.000000,5.618829)(115.000000,1.242490)(120.000000,0.251148)(125.000000,0.057238)
(130.000000,0.017942)(135.000000,0.008398)(140.000000,0.004990)(145.000000,0.003358)(150.000000,0.002371)
(155.000000,0.001906)(160.000000,0.001645)(165.000000,0.001419)(170.000000,0.001136)(175.000000,0.001004)
(180.000000,0.000981)(185.000000,0.000785)(190.000000,0.000661)(195.000000,0.000586)(200.000000,0.000546)

\newrgbcolor{color961.0033}{0.75           0        0.75}
\psline[plotstyle=line,linejoin=1,showpoints=true,dotstyle=Basterisk,dotsize=\MarkerSize,linestyle=solid,linewidth=\LineWidth,linecolor=color961.0033]
(5.000000,1.253643)(10.000000,1.305771)(15.000000,1.309051)(20.000000,1.385854)(25.000000,1.386037)
(30.000000,1.412888)(35.000000,1.531467)(40.000000,1.529301)(45.000000,1.518538)(50.000000,1.674707)
(55.000000,5.824523)(60.000000,13.141459)(65.000000,20.991844)(70.000000,28.240726)(75.000000,32.109038)
(80.000000,24.772708)(85.000000,24.547631)(90.000000,24.543948)(95.000000,24.529938)(100.000000,24.444363)
(105.000000,19.589624)(110.000000,7.011583)(115.000000,1.568879)(120.000000,0.320238)(125.000000,0.073583)
(130.000000,0.023205)(135.000000,0.010903)(140.000000,0.006494)(145.000000,0.004379)(150.000000,0.003097)
(155.000000,0.002493)(160.000000,0.002152)(165.000000,0.001857)(170.000000,0.001489)(175.000000,0.001316)
(180.000000,0.001286)(185.000000,0.001030)(190.000000,0.000869)(195.000000,0.000770)(200.000000,0.000718)

\newrgbcolor{color962.0033}{0.75        0.75           0}
\psline[plotstyle=line,linejoin=1,showpoints=true,dotstyle=B+,dotsize=\MarkerSize,linestyle=solid,linewidth=\LineWidth,linecolor=color962.0033]
(5.000000,1.150886)(10.000000,1.177653)(15.000000,1.179292)(20.000000,1.216192)(25.000000,1.216276)
(30.000000,1.228522)(35.000000,1.278789)(40.000000,1.277924)(45.000000,1.273599)(50.000000,1.331848)
(55.000000,1.643868)(60.000000,1.644878)(65.000000,1.652510)(70.000000,1.647399)(75.000000,1.792534)
(80.000000,4.620881)(85.000000,9.528767)(90.000000,14.521071)(95.000000,19.560289)(100.000000,24.416483)
(105.000000,22.547184)(110.000000,8.355192)(115.000000,1.895237)(120.000000,0.390704)(125.000000,0.090491)
(130.000000,0.028702)(135.000000,0.013534)(140.000000,0.008081)(145.000000,0.005459)(150.000000,0.003866)
(155.000000,0.003115)(160.000000,0.002690)(165.000000,0.002324)(170.000000,0.001864)(175.000000,0.001649)
(180.000000,0.001611)(185.000000,0.001292)(190.000000,0.001090)(195.000000,0.000967)(200.000000,0.000901)

{ \small 
\rput(166.423387,31.072523){%
\psshadowbox[framesep=0pt,linewidth=\AxesLineWidth]{\psframebox*{\begin{tabular}{l}
\Rnode{a1}{\hspace*{0.0ex}} \hspace*{0.4cm} \Rnode{a2}{~~user 1} \\
\Rnode{a3}{\hspace*{0.0ex}} \hspace*{0.4cm} \Rnode{a4}{~~user 2} \\
\Rnode{a5}{\hspace*{0.0ex}} \hspace*{0.4cm} \Rnode{a6}{~~user 3} \\
\Rnode{a7}{\hspace*{0.0ex}} \hspace*{0.4cm} \Rnode{a8}{~~user 4} \\
\Rnode{a9}{\hspace*{0.0ex}} \hspace*{0.4cm} \Rnode{a10}{~~user 5} \\
\Rnode{a11}{\hspace*{0.0ex}} \hspace*{0.4cm} \Rnode{a12}{~~user 6} \\
\end{tabular}}
\ncline[linestyle=solid,linewidth=\LineWidth,linecolor=color957.0038]{a1}{a2} \ncput{\psdot[dotstyle=*,dotsize=\MarkerSize,linecolor=color957.0038]}
\ncline[linestyle=solid,linewidth=\LineWidth,linecolor=color958.0033]{a3}{a4} \ncput{\psdot[dotstyle=Bsquare,dotsize=\MarkerSize,linecolor=color958.0033]}
\ncline[linestyle=solid,linewidth=\LineWidth,linecolor=color959.0033]{a5}{a6} \ncput{\psdot[dotstyle=Bo,dotsize=\MarkerSize,linecolor=color959.0033]}
\ncline[linestyle=solid,linewidth=\LineWidth,linecolor=color960.0033]{a7}{a8} \ncput{\psdot[dotstyle=Bpentagon,dotsize=\MarkerSize,linecolor=color960.0033]}
\ncline[linestyle=solid,linewidth=\LineWidth,linecolor=color961.0033]{a9}{a10} \ncput{\psdot[dotstyle=Basterisk,dotsize=\MarkerSize,linecolor=color961.0033]}
\ncline[linestyle=solid,linewidth=\LineWidth,linecolor=color962.0033]{a11}{a12} \ncput{\psdot[dotstyle=B+,dotsize=\MarkerSize,linecolor=color962.0033]}
}%
}%
} 

\end{pspicture}%

\caption{The final bids $w_i$ for different values of $P_T$ and $\delta = 10^{-3}$ for Algorithm in (\ref{alg:OuP_UE}) and (\ref{alg:eNodeB_first}). It can be seen that as the users bid higher they are allocated higher
power. Since our algorithm gives priority to lower modulation schemes we start by allocating powers to users with lower modulation schemes. The users with lower modulation schemes bid higher
when the resources are scarce and their bids decrease as $P_T$
increases. Therefore, the pricing which is proportional to the
bids is traffic-dependent. }
\label{fig:sim:bids_robust}
\end{figure}
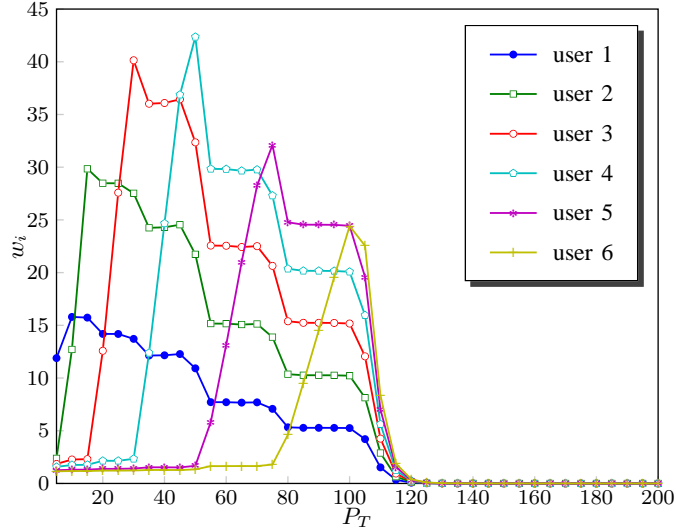
\begin{figure}[tb]
\centering
%
\psset{xunit=0.005128\plotwidth,yunit=0.197177\plotwidth}%
\begin{pspicture}(-16.566820,-0.444444)(205.391705,4.093567)%


\psline[linewidth=\AxesLineWidth,linecolor=GridColor](20.000000,0.000000)(20.000000,0.060859)
\psline[linewidth=\AxesLineWidth,linecolor=GridColor](40.000000,0.000000)(40.000000,0.060859)
\psline[linewidth=\AxesLineWidth,linecolor=GridColor](60.000000,0.000000)(60.000000,0.060859)
\psline[linewidth=\AxesLineWidth,linecolor=GridColor](80.000000,0.000000)(80.000000,0.060859)
\psline[linewidth=\AxesLineWidth,linecolor=GridColor](100.000000,0.000000)(100.000000,0.060859)
\psline[linewidth=\AxesLineWidth,linecolor=GridColor](120.000000,0.000000)(120.000000,0.060859)
\psline[linewidth=\AxesLineWidth,linecolor=GridColor](140.000000,0.000000)(140.000000,0.060859)
\psline[linewidth=\AxesLineWidth,linecolor=GridColor](160.000000,0.000000)(160.000000,0.060859)
\psline[linewidth=\AxesLineWidth,linecolor=GridColor](180.000000,0.000000)(180.000000,0.060859)
\psline[linewidth=\AxesLineWidth,linecolor=GridColor](200.000000,0.000000)(200.000000,0.060859)
\psline[linewidth=\AxesLineWidth,linecolor=GridColor](5.000000,0.000000)(7.340000,0.000000)
\psline[linewidth=\AxesLineWidth,linecolor=GridColor](5.000000,0.500000)(7.340000,0.500000)
\psline[linewidth=\AxesLineWidth,linecolor=GridColor](5.000000,1.000000)(7.340000,1.000000)
\psline[linewidth=\AxesLineWidth,linecolor=GridColor](5.000000,1.500000)(7.340000,1.500000)
\psline[linewidth=\AxesLineWidth,linecolor=GridColor](5.000000,2.000000)(7.340000,2.000000)
\psline[linewidth=\AxesLineWidth,linecolor=GridColor](5.000000,2.500000)(7.340000,2.500000)
\psline[linewidth=\AxesLineWidth,linecolor=GridColor](5.000000,3.000000)(7.340000,3.000000)
\psline[linewidth=\AxesLineWidth,linecolor=GridColor](5.000000,3.500000)(7.340000,3.500000)
\psline[linewidth=\AxesLineWidth,linecolor=GridColor](5.000000,4.000000)(7.340000,4.000000)

{ \footnotesize 
\rput[t](20.000000,-0.060859){$20$}
\rput[t](40.000000,-0.060859){$40$}
\rput[t](60.000000,-0.060859){$60$}
\rput[t](80.000000,-0.060859){$80$}
\rput[t](100.000000,-0.060859){$100$}
\rput[t](120.000000,-0.060859){$120$}
\rput[t](140.000000,-0.060859){$140$}
\rput[t](160.000000,-0.060859){$160$}
\rput[t](180.000000,-0.060859){$180$}
\rput[t](200.000000,-0.060859){$200$}
\rput[r](2.660000,0.000000){$0$}
\rput[r](2.660000,0.500000){$0.5$}
\rput[r](2.660000,1.000000){$1$}
\rput[r](2.660000,1.500000){$1.5$}
\rput[r](2.660000,2.000000){$2$}
\rput[r](2.660000,2.500000){$2.5$}
\rput[r](2.660000,3.000000){$3$}
\rput[r](2.660000,3.500000){$3.5$}
\rput[r](2.660000,4.000000){$4$}
} 

\psframe[linewidth=\AxesLineWidth,dimen=middle](5.000000,0.000000)(200.000000,4.000000)

{ \small 
\rput[b](102.500000,-0.444444){
\begin{tabular}{c}
$P_T$\\
\end{tabular}
}

\rput[t]{90}(-16.566820,2.000000){
\begin{tabular}{c}
$p$\\
\end{tabular}
}
} 

\newrgbcolor{color999.0039}{0  0  1}
\psline[plotstyle=line,linejoin=1,linestyle=solid,linewidth=\LineWidth,linecolor=color999.0039]
(5.000000,3.996539)(10.000000,3.499995)(15.000000,3.474428)(20.000000,3.002154)(25.000000,3.001259)
(30.000000,2.878586)(35.000000,2.489004)(40.000000,2.494498)(45.000000,2.522519)(50.000000,2.207303)
(55.000000,1.504523)(60.000000,1.503389)(65.000000,1.494944)(70.000000,1.500576)(75.000000,1.370634)
(80.000000,1.010433)(85.000000,1.000432)(90.000000,1.000268)(95.000000,0.999647)(100.000000,0.995852)
(105.000000,0.785574)(110.000000,0.269549)(115.000000,0.057797)(120.000000,0.011335)(125.000000,0.002515)
(130.000000,0.000772)(135.000000,0.000357)(155.000000,0.000079)(200.000000,0.000022)

\end{pspicture}%

\caption{The final shadow price $p$ for different values of $P_T$ and $\delta = 10^{-3}$ for Algorithm in (\ref{alg:OuP_UE}) and (\ref{alg:eNodeB_first}). It can be seen that for high-traffic periods or when the power is scarce the network providers charge higher price for the same amount of per unit power. This traffic-dependent pricing regimen provides incentive to users to use network during off-peak hours as they'll pay less for the same amount of per unit power.}
\label{fig:sim:shadow_price_robust}
\end{figure}
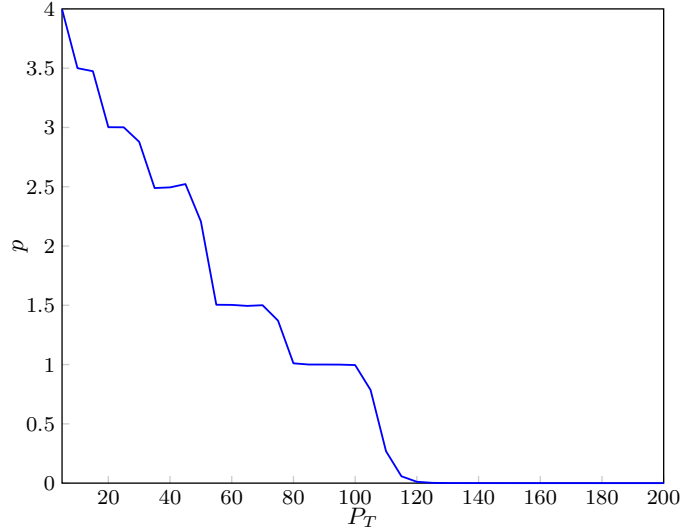

\section{Simulation Results}\label{sec:sim}

We present the simulation results of six utility functions corresponding to six UEs, as shown in Figure \ref{fig:sigmoid}. Algorithm in (\ref{alg:OuP_UE}) and (\ref{alg:eNodeB_first}) was applied to sigmoidal-like utility functions with different parameters using MATLAB. Our simulation results showed convergence to the optimal global powers for all values of the BS power $P_T$.


\subsection{Convergence Dynamics for $P_T$ = 45}
In the following simulations, we set $P_T$ = 45 and the number of iterations $n=40$. 


\textbf{Algorithm (\ref{alg:UE_first}) and (\ref{alg:eNodeB_first}): Non-convergent Powers and Bids} Here, we choose the total BS power $P_T$ to be less than the sum of users' inflection points $\sum b_i$. Therefore, Algorithm in (\ref{alg:UE_first}) and (\ref{alg:eNodeB_first}) does not converge in this region. In Figure \ref{fig:sim:powers_iter}, we show the powers $P_i(n)$ of different users with the number of iterations $n$ for Algorithm in (\ref{alg:UE_first}) and (\ref{alg:eNodeB_first}). It is shown that the powers fluctuate around the optimal powers and so the optimal powers are not achieved and the exit condition is not satisfied (i.e. endless iterations). Similar behavior for bids $w_i(n)$ with the number of iterations $n$ is shown in Figure \ref{fig:sim:bids_iter}. 

\textbf{Algorithm (\ref{alg:OuP_UE}) and (\ref{alg:eNodeB_first}): Convergent Powers and Bids} The behavior is more robust in Algorithm (\ref{alg:OuP_UE}) and (\ref{alg:eNodeB_first}) due to the fluctuation decay function. It damps the fluctuation with every iteration so the network reaches the optimal powers of the optimization 
problem (\ref{eqn:opt_prob_fairness_UE}). The powers $P_i(n)$ and bids $w_i(n)$ of Algorithm in (\ref{alg:OuP_UE}) and (\ref{alg:eNodeB_first}) are shown in Figures \ref{fig:sim:powers_iteP_robust} and \ref{fig:sim:bids_iteP_robust}, respectively. 

\textbf{Optimal Shadow Price $p(n)$:} In Figure \ref{fig:sim:p_compare}, we show the new shadow price $p(n)$ of the fluctuation example in Section \ref{sec:Fluctuation_example} when using Algorithm in (\ref{alg:OuP_UE}) and (\ref{alg:eNodeB_first}). The shadow price $p(n)$ fluctuation decreases with every iteration $n$ and converges to the optimal shadow price that corresponds to the optimal powers. On the contrary, when using Algorithm in (\ref{alg:UE_first}) and (\ref{alg:eNodeB_first}), the shadow price fluctuates and doesn't reach optimal value.


\subsection{Power Allocation and Pricing for $5\le P_T\le100$}
In the following simulations, we set $\delta =10^{-3}$ and the total power of the BS $P_T$ takes values between 5 and 100 with a step of 5. 

\textbf{Optimal Power Allocation:} In Figure \ref{fig:sim:powers_robust}, we show the final powers of different users with different BS power $P_T$. Our distributed algorithm is set to avoid the situation of allocating zero power to any user (i.e. no user is dropped). However, the BS allocates the majority of the powers to the UEs running low modulation schemes until they reach the inflection power $P_i = b_i$. When the total power $P_T$ exceeds the sum of the inflection powers $\sum b_i$ of all the users with lower modulation schemes, BS allocates more powers to the UEs with higher modulation schemes, as shown in Figure \ref{fig:sim:powers_robust}, when BS power exceeds $P_T= 65$. This behavior is similar to that in \ref{alg:UE_first} but with including BS power $P_T<60$ where the power is scarce with respect to the users' modulation schemes. In Figure \ref{fig:sim:SumOfPower}, we show the sum of powers $\sum P_i$ for different BS power $P_T$ values. When the total power $P_T$ exceeds the sum of the 
inflection power $\sum P_i = \sum b_i$ of all the users, the users demands are met and the power allocation reaches a steady state. This is similar in behavior to individual users, Figure \ref{fig:sim:powers_robust}, where after the inflection point the power allocation is steady. However, Figure \ref{fig:sim:SumOfPower} gives an overall relation between total power available and demands of users in the network.

\textbf{Traffic-dependent Bidding/Pricing:} In Figure \ref{fig:sim:bids_robust}, we 
show the final bids of different users with different BS total power $P_T$. The higher the user bids the higher the allocated power. The users with lower modulation schemes bid higher when the resources are scarce and their bids decrease as $P_T$ increases. Therefore, the pricing which is proportional to the bids is traffic-dependent. This gives service providers an option to increase the service price for subscribers when the traffic load on the system is high. Therefore, service providers can motivate subscribers to use the network when the traffic load on the network is low, as they pay less for the same service. The shadow price $p(n)$ represents the total price per unit power for all users. In Figure \ref{fig:sim:shadow_price_robust}, we show the shadow price $p(n)$ with BS power $P_T$. The price is high for high-traffic (i.e. fixed number of 
users but less resources, $P_T$ is small) which decreases for low-traffic (i.e. same number of users but more resources, $P_T$ is large). A large decrease in the price is apparent after $P=\{10, 30, 60\}$ which are the points where one of the users utility exceed the inflection point. This large decrease occurs at the sum of inflection points $\sum_{i=1}^{k} P_i^{\text{inf}}$, where $k = \{1,2,...,M\}$ is the users index and $M$ is the number of users.

\section{Conclusion}\label{sec:conclude}

In this paper, we addressed the problem of optimal allocation of BS powers to users running different modulation schemes by taking into account the price paid per unit power by each user and the total power available at the BS. We used sigmoidal-like utility functions to represent different users' modulation schemes due to the resemblance in shape of sigmoidal-like utility functions and the probability of packet transmission success of modulation schemes. We showed that our network utility maximization problem is convex and thus has optimal solution. Our distributed power allocation algorithm ensured fairness in the allocation of powers as we gave priority to users running lower modulation schemes while ensuring non-zero power allocation to users running higher modulation schemes. Moreover, the power allocation algorithm was convergent, for all network traffic conditions, and there was no oscillation in the power allocation process due to the introduction of fluctuation damping parameter in our algorithm. We 
also illustrated that our algorithm provided a pricing approach for network providers that could be used to reduce the demand in power i.e. flatten traffic loads during peak traffic hours. We showed that the price per unit power, the power allocation algorithm charged from users depended upon the modulation scheme used and the total power available at the BS. This provided an opportunity for network providers to flatten load on their network by motivating users to use the network at off-peak hours as they paid more for the same amount of per unit power in peak hours than in off-peak hours.

\bibliographystyle{ieeetr}
\bibliography{pubs}
\end{document}